\documentclass[twoside]{article}

\usepackage[preprint]{aistats2027}

\usepackage[round]{natbib}

\usepackage{amsmath,amssymb,amsthm,mathtools}
\usepackage{booktabs}
\usepackage{array}
\usepackage{float}
\newfloat{algofloat}{tbp}{loa}
\floatname{algofloat}{Algorithm}
\usepackage{placeins}
\usepackage{needspace}
\usepackage{graphicx}
\usepackage[dvipsnames]{xcolor}
\usepackage{tikz}
\usepackage[colorlinks=true,linkcolor=BrickRed,citecolor=MidnightBlue,urlcolor=MidnightBlue]{hyperref}
\hypersetup{pdftitle={Exact Regret Frontiers and Externality
Scheduling in Centralized Serial-Dictatorship Bandits},
pdfauthor={},
pdfkeywords={matching bandits, serial dictatorship, Graves-Lai lower
bound, exploration externality, Pareto frontier}}

\newtheorem{theorem}{Theorem}[section]
\newtheorem{lemma}[theorem]{Lemma}
\newtheorem{proposition}[theorem]{Proposition}
\newtheorem{corollary}[theorem]{Corollary}
\theoremstyle{definition}
\newtheorem{definition}[theorem]{Definition}
\newtheorem{assumption}[theorem]{Assumption}
\newtheorem{example}[theorem]{Example}
\theoremstyle{remark}
\newtheorem{remark}[theorem]{Remark}

\newcommand{\E}{\mathbb{E}}
\newcommand{\Prob}{\mathbb{P}}
\newcommand{\R}{\mathbb{R}}
\DeclareMathOperator{\KL}{KL}
\newcommand{\Normal}{\mathcal{N}}
\newcommand{\ind}[1]{\mathbf{1}\{#1\}}
\newcommand{\mstar}{m^{\star}}
\newcommand{\Cclass}{\mathcal{C}}
\newcommand{\Mset}{\mathcal{M}}
\newcommand{\Eset}{\mathcal{E}}
\newcommand{\Lset}{\Lambda}
\newcommand{\Thstrict}{\Theta_{\mathrm{strict}}}
\newcommand{\Thsep}{\Theta_{\mathrm{sep}}}
\newcommand{\UGL}{\mathcal{U}_{\mathrm{GL}}}
\newcommand{\Xset}{\mathcal{X}}

\begin{document}

\runningauthor{Xu, Ma, Weng, Xia}
\runningtitle{Exact Regret Frontiers and Externality Scheduling in
Serial-Dictatorship Bandits}

\twocolumn[

\aistatstitle{Exact Regret Frontiers and Externality Scheduling\\
in Centralized Serial-Dictatorship Bandits}

\aistatsauthor{
Lishang Xu
\And
Guodong Ma
\And
Pengcheng Weng
\And
Zixuan Xia$^{*}$
}

\aistatsaddress{
University of Bern
\And
EPFL \\ University of Bern
\And
University of Bern
\And
University of Bern
}
]

{\renewcommand{\thefootnote}{}\footnotetext{$^{*}$Corresponding:
\texttt{zixuan.xia@students.unibe.ch}}}

\begin{abstract}
Exploration in centralized serial-dictatorship matching bandits must use complete matchings, so learning one player--arm pair can impose regret on others. We study this externality under a known common priority order and Gaussian rewards with unit variance. We show that the matching-level Graves--Lai constraints reduce to finitely many pairwise exploration quotas and, at top-choice-separated instances, yield a polynomial-size marginal linear program. At these instances, the exact attainable set of expected logarithmic regret coefficients is $G(\theta)\Xset(\theta)$, where $\Xset$ is the feasible matching-allocation set and $G$ maps allocations to player regret. The usual upper-closed Graves--Lai region can be strictly larger despite having the same Pareto-minimal boundary. We further show that identical exploration quotas can induce very different regret through their scheduling. Finally, we construct estimate--solve--track policies, uniformly good on the full row-strict class, that attain every fixed positively weighted optimum without assuming optimizer uniqueness. Every Pareto-minimal point is pointwise attainable, possibly through an instance-calibrated target.
\end{abstract}

\section{INTRODUCTION}\label{sec:intro}

A centralized matching platform repeatedly assigns $N$ players to
$K\ge N$ arms while learning their preferences from noisy rewards.
We consider a known common priority order, so the stable matching is
given by serial dictatorship. Such priority-based assignment arises,
for instance, when a platform allocates shifts or tasks to workers by
seniority or rank. Since the platform must choose a
complete matching in every round, exploring one player--arm pair can
displace other players from their stable arms. Thus, the player whose
preferences are being learned need not be the only one who pays for
that learning.

We study this externality under Gaussian rewards with known unit
variance and strict within-player preferences. Our main question is:
\emph{which vectors of player-wise logarithmic regret coefficients
are attainable by uniformly good policies?} The key observation is
that the problem splits in two: pairwise exploration quotas fix what
must be learned, while the complete matchings that serve them
determine who pays (Figure~\ref{fig:schematic}). A platform can thus
trade one player's regret against another's; we characterize exactly
how far, and how to get there.

Our information-theoretic reduction holds on the full row-strict
parameter space $\Thstrict$. The regret-region and attainability
results are stated for \emph{top-choice-separated} instances, where
players have distinct best arms and all regret is therefore the price
of learning, while policies remain uniformly good over all of
$\Thstrict$.

\begin{figure}[t]
\centering
\begin{tikzpicture}[x=1pt,y=1pt,font=\scriptsize,
pl/.style={circle,draw,inner sep=0.5pt,minimum size=10pt},
arm/.style={rectangle,draw,inner sep=1.5pt,minimum size=10pt},
ex/.style={->,thick,BrickRed},
exb/.style={->,thick,MidnightBlue,densely dashed}]
\begin{scope}
\node[align=center] at (40,50) {\textbf{parallel}: $m_{\mathrm B}$};
\foreach \i/\y in {1/32,2/17,3/2}{
\node[pl] (p\i) at (0,\y) {$p_\i$};
\node[arm] (a\i) at (56,\y) {$a_\i$};}
\draw[ex] (p1) -- (a2);
\draw[ex] (p2) -- (a3);
\draw[ex] (p3) -- (a1);
\node[anchor=west] at (64,32) {$r_1=22$};
\node[anchor=west] at (64,17) {$r_2=20$};
\node[anchor=west] at (64,2) {$r_3=\mathbf{202}$};
\end{scope}
\begin{scope}[xshift=118pt]
\node[align=center] at (44,50)
{\textbf{protective}: $m_{\mathrm A}$, then $m_{\mathrm C}$};
\foreach \i/\y in {1/32,2/17,3/2}{
\node[pl] (q\i) at (0,\y) {$p_\i$};
\node[arm] (b\i) at (56,\y) {$a_\i$};}
\draw[ex] (q1) -- (b2);
\draw[ex] (q2) -- (b1);
\draw[ex] (q3) -- (b3);
\draw[exb] (q1) -- (b1);
\draw[exb] (q2) -- (b3);
\draw[exb] (q3) -- (b2);
\node[anchor=west] at (64,32) {$r_1=22$};
\node[anchor=west] at (64,17) {$r_2=\mathbf{222}$};
\node[anchor=west] at (64,2) {$r_3=2.02$};
\end{scope}
\end{tikzpicture}
\caption{The same exploration quotas, scheduled differently, make
different players pay. Left: $m_{\mathrm B}$ serves the large quotas of
$p_1$ and $p_2$ at once and leaves $p_3$ on its worst arm ($r_3=202$).
Right: $m_{\mathrm A}$, then $m_{\mathrm C}$ (dashed), spare $p_3$
($r_3=2.02$) but make $p_2$ pay for both large quotas ($r_2=222$).
$r_i$ is $p_i$'s expected logarithmic regret coefficient
(Example~\ref{ex:33}); the stable matching is $p_i\mapsto a_i$; a small
quota is not drawn.}
\label{fig:schematic}
\end{figure}

\paragraph{Contributions.}
Our main results are as follows.

\begin{enumerate}
\item \textbf{Exact pairwise quotas and a compact LP.}
The Graves--Lai constraints are exactly equivalent to one quota per
eligible player--arm pair $(i,a)$: it must be played at least
$(2/\Delta_{i,a}^2)\log T$ times, where $\Delta_{i,a}$ is its gap to
the stable arm (Theorem~\ref{thm:reduction}).
At top-choice-separated instances, optimizing any positively
weighted regret reduces to a polynomial-size LP over player--arm
marginals, and every feasible marginal solution decomposes into at
most $NK-N+1$ complete matchings.

\item \textbf{Exact attainable regret region.}
Let $\Xset(\theta)$ denote the feasible matching-allocation
polyhedron and $G(\theta)$ the map from allocations to player-wise
regret coefficients. Writing $\mathcal A(\theta)$ for the set of
attainable regret-coefficient vectors and $\UGL(\theta)$ for the usual
upper-closed Graves--Lai region, we prove

\[
\begin{aligned}
  \mathcal A(\theta)&=G(\theta)\Xset(\theta),\\
  \UGL(\theta)&=\mathcal A(\theta)+\R_+^N.
\end{aligned}
\]

Hence $\UGL(\theta)$ has the same Pareto-minimal boundary as
$\mathcal A(\theta)$ but, exactly when $K=N$, contains regret vectors
that no complete-matching schedule can realize
(Theorem~\ref{thm:closure}, Corollary~\ref{cor:rigid}).
We characterize this frontier explicitly in a three-player example
and provide higher-dimensional families.

\item \textbf{Policies attaining the frontier.}
We construct estimate--solve--track policies. For every fixed positive
weight vector they attain the weighted optimum even when the optimal
schedule is not unique, and a data-only target rule selects which
optimal schedule is realized. Any Pareto-minimal point of an instance
is attained by such a policy with a target calibrated to that
instance.
\end{enumerate}

\section{RELATED WORK}\label{sec:related}

\paragraph{Learning in matching markets.}
\citet{liu2020competing} formulated a competing-bandits model for
matching markets and established stable-regret guarantees.
Subsequent work studied serial-dictatorship markets and sharper regret guarantees
\citep{sankararaman2021dominate,kong2024improved,wang2024optimal}.
Related work also considers player-optimal stable regret in general
decentralized markets \citep{kong2023playeroptimal}, as well as
sleeping and contextual variants
\citep{uba2026sleeping,lin2026contextual}.
Our setting is instead centralized with a known common priority order,
which makes it possible to study how complete-matching schedules
distribute exploration cost across players.

The closest lower-bound work is \citet{basu2025competing}.
Using single-entry alternatives, Basu derives matching-level
information constraints and a matching-cover form of pairwise
constraints for a relaxed binary stable-regret program.
Pairwise constraints are therefore not new in themselves. In our
known-common-priority model, we prove that the corresponding pairwise
quotas exactly describe the full Graves--Lai feasible set over
confusing alternatives. We then
characterize the image of this feasible set in player-wise regret
space and establish its attainability under complete-matching
policies. Recent work on stable-matching identification studies
one-sided and two-sided uncertainty
\citep{hosseini2024putting,pagare2025optimal,athanasopoulos2025probably,
athanasopoulos2026stable}, the last with a regret-minimization
extension; our focus is instead the
exact attainable set of asymptotic player-regret coefficients. \citet{li2025survey}
provide a broader survey.

\paragraph{Graves--Lai optimization.}
Our lower bound builds on the asymptotic information framework of
\citet{lai1985asymptotically,graves1997asymptotically}, its use in
structured and combinatorial bandits
\citep{combes2015combinatorial,combes2017minimal}, and the
fundamental inequality of \citet{garivier2019explore}.
For combinatorial semi-bandits,
\citet{cuvelier2021asymptotically} solve the Graves--Lai program
to any accuracy $\delta>0$ in polynomial time given an exact
budgeted linear-maximization oracle; for matchings the oracle is only
approximate, costing an extra fixed factor. In our model, the special structure
of the confusing
alternatives yields a more direct finite description: the
information constraints reduce to explicit pairwise quotas, and at
top-choice-separated instances the scheduling problem becomes a
polynomial-size LP over player--arm marginals.
This reduction is what allows us to describe the full vector-valued
regret region rather than only a scalar optimum.

\paragraph{Exploration externalities and allocation selection.}
Exploration can impose costs on agents other than the one whose
information is being collected
\citep{raghavan2018externalities,baek2024fair}.
In particular, \citet{baek2024fair} study bargaining over attainable
utilities in grouped bandits. Here the externality arises
from the complete-matching constraint, and the resulting asymptotic
regret vectors form an explicit polyhedral image whose Pareto
boundary we characterize.

Our attainment argument is also related to adaptive tracking under
nonunique optimal allocations.
\citet{degenne2019multiple} introduce sticky tracking for
multiple-answer pure exploration, while
\citet{vanparys2024optimal} construct continuous selections for
near-optimal exploration allocations.
Our objective is different: we select and track a prescribed point
on the exact cost-optimal face of matching marginals, which permits
different attainable divisions of regret among players without
requiring optimizer uniqueness.

Appendix~\ref{app:related} adds comparisons with work on
transfers, welfare, indifference, and Graves--Lai computation.

\section{MODEL}\label{sec:model}

There are $N$ players $p_1,\dots,p_N$ and $K\ge N$ arms
$a_1,\dots,a_K$; all arms share the common priority order
$p_1\succ\cdots\succ p_N$, known to the platform
(\emph{serial dictatorship}). We assume $K\ge2$, since for $K=1$
every policy has zero regret. Player
$p_i$'s utility for arm $a$ has unknown mean $\mu_{i,a}$, with
Gaussian rewards of known common variance, normalized to one:
$\nu_{i,a}=\Normal(\mu_{i,a},1)$. At each round
$t=1,\dots,T$, a centralized platform selects an injective matching
$M_t\in\Mset$ ($\Mset$ is the set of injective maps $[N]\to[K]$) and observes
all matched-pair rewards (centralized semi-bandit feedback). The
policy is non-anticipating; conditionally on $M_t=m$ the
observation has independent coordinates
$X_{i,t}\sim\Normal(\mu_{i,m(i)},1)$, independent of the past, and
the same history-to-action kernel is used under every instance.
Write $N_m(T):=\sum_{t\le T}\ind{M_t=m}$ and
$N_{i,a}(T):=\sum_{t\le T}\ind{M_t(i)=a}$; the policy's counts
$N_{i,a}(t)$ and empirical means $\widehat\mu(t)$ exclude round $t$.

\begin{assumption}[Complete matching]\label{ass:complete}
$M_t\in\Mset$ in every round, so every player is matched.
Remark~\ref{rem:equality} discusses partial matchings in the
two-by-two market; the general case is left open.
\end{assumption}

The serial-dictatorship stable matching $\mstar(\theta)$ lets
$p_1$ take its best arm, then $p_2$ its best remaining arm, and so
on. Individual and weighted regret are
$R_i(T;\theta)=T\mu_{i,\mstar(i)}-\E\sum_{t\le T}\mu_{i,M_t(i)}$
and $R_w=\sum_iw_iR_i$, $w>0$, studied asymptotically in
$\log T$. With $g_i(m;\theta):=\mu_{i,\mstar(i)}-\mu_{i,m(i)}$
($g_i(\mstar;\theta)=0$), complete matching gives the exact
accounting identity
\begin{equation}\label{eq:regretcount}
R_i(T;\theta)=\sum_{m\in\Mset}\E_\theta[N_m(T)]\,g_i(m;\theta).
\end{equation}

\paragraph{Parameter space.}
The global parameter space is the set of \emph{row-strict}
instances
$\Thstrict:=\{\mu\in\R^{N\times K}:\mu_{i,a}\neq\mu_{i,b}\
\forall i,a\neq b\}$, each with a unique stable matching; entire
rows must be strict because the alternatives of
\S\ref{sec:information} change the remaining set that lower-priority
players face. The instances we \emph{analyze} are additionally
\emph{top-choice separated},
$\Thsep:=\{\theta\in\Thstrict:\arg\max_a\mu_{i,a}\text{ distinct
across }i\}$, so that $\mstar(i)=\arg\max_a\mu_{i,a}$ and
$g_i(m;\theta)\ge0$ everywhere. Every player then holds its favorite
arm, so the stable matching involves no conflict and regret arises
only from learning; priority acts through exploration alone, by
fixing which pairs must be sampled and which players those samples
displace.

\begin{remark}[Separation is local]\label{rem:local}
Separation restricts the analyzed instance, not the parameter
space: the lower-bound alternatives raise a single entry above
$\mu_{i,\mstar(i)}$ and typically leave $\Thsep$ while remaining in
$\Thstrict$; quantifying over $\Thsep$ only would make them
inadmissible and change even the two-by-two constant. Separation
makes every off-stable gap positive, bounding normalized matching
counts by normalized regret and the cost-optimal marginals by their
cost. Theorem~\ref{thm:reduction} holds on all of $\Thstrict$.
Proposition~\ref{prop:region} and Theorem~\ref{thm:weighted} need
uniform goodness only at $\theta$ and on its single-entry alternatives
(Remark~\ref{rem:finiteclass}).
\end{remark}

\begin{definition}[Uniformly good]\label{def:unifgood}
A policy is \emph{uniformly good on $\Theta'\subseteq\Thstrict$} if
$(R_i(T;\lambda))_+ = o(T^{\alpha})$ for every
$\lambda\in\Theta'$, $i$, $\alpha>0$.
\end{definition}

The positive part matters: for $\lambda\notin\Thsep$ stable regret
can be negative; Lemma~\ref{lem:nocancel} shows the definition
already forces $|R_i|=o(T^{\alpha})$.

\section{Information Requirements and Scheduling Geometry}
\label{sec:general-lb}

\subsection{Exact Information Requirements}\label{sec:information}

When matching $m$ is executed, the platform observes a draw from
$P^m_\theta=\bigotimes_{i=1}^N\Normal(\mu_{i,m(i)},1)$. For
$\theta,\lambda\in\Thstrict$ set
\begin{equation}\label{eq:Dm}
D_m(\theta,\lambda)
:=\KL\bigl(P^m_\theta\,\big\|\,P^m_\lambda\bigr)
=\tfrac12\sum_{i=1}^{N}
\bigl(\mu_{i,m(i)}-\lambda_{i,m(i)}\bigr)^2 ,
\end{equation}
where additivity follows from the conditional independence of matched
pairs.

The divergence between the laws of the whole interaction is
$\sum_m\E_\theta[N_m(T)]D_m(\theta,\lambda)$
(Lemma~\ref{lem:divdecomp}, Appendix~\ref{app:c1-com}).

By hierarchical no-cancellation, a uniformly good policy executes
$\mstar(\theta)$ in all but $o(T^{\alpha})$ rounds in expectation, even
when individual stable regret can be negative
(Lemma~\ref{lem:nocancel}, Appendix~\ref{app:c1-nocancel}). Hence
alternatives with $D_{\mstar(\theta)}(\theta,\lambda)>0$ impose no
constraint at the $\log T$ scale; the binding ones agree with
$\theta$ on every stable pair yet change the stable matching:
\begin{equation}\label{eq:confusing}
\begin{split}
\Lset(\theta)
=\bigl\{\lambda\in\Thstrict:\;
&\mstar(\lambda)\neq\mstar(\theta),\\[-2pt]
&D_{\mstar(\theta)}(\theta,\lambda)=0\bigr\},
\end{split}
\end{equation}
the confusing-parameter set of the Graves--Lai program
\citep{graves1997asymptotically,combes2017minimal}.

\begin{theorem}[Information allocation]\label{thm:infoalloc}
Let the policy be uniformly good on $\Thstrict$ and satisfy
Assumption~\ref{ass:complete}. Then for every $\theta\in\Thstrict$
and every $\lambda\in\Lset(\theta)$,
\[
\liminf_{T\to\infty}\;\frac{1}{\log T}
\sum_{m\neq\mstar(\theta)}
\E_\theta[N_m(T)]\,D_m(\theta,\lambda)\;\ge\;1 .
\]
\end{theorem}

Appendix~\ref{app:c1-infoalloc} applies the fundamental
inequality of \citet{garivier2019explore} to
$Z_T=N_{\mstar(\theta)}(T)/T$. Define the \emph{matching-allocation
feasible set}, with $x$ indexed by $\Mset\setminus\{\mstar\}$,
\begin{equation}\label{eq:Xset}
\Xset(\theta)
:=\Bigl\{x\ge0:\,
\sum_{m\neq\mstar}x_m D_m(\theta,\lambda)\ge1
\ \ \forall\lambda\in\Lset(\theta)\Bigr\}.
\end{equation}

Every finite cluster point of the normalized expected matching
counts, along a common subsequence, belongs to $\Xset(\theta)$
(Proposition~\ref{prop:cluster}, Appendix~\ref{app:c1-region}).
The two-by-two market provides a warm-up
(Appendix~\ref{app:twobytwo}).

Let the \emph{eligible set} $\Eset$ consist of the free off-stable pairs $(i,a)$ with
$a\notin\{\mstar(1),\dots,\mstar(i-1)\}$ and $a\neq\mstar(i)$.
For $(i,a)\in\Eset$, let
$\Delta_{i,a}:=\mu_{i,\mstar(i)}-\mu_{i,a}$ and
$q_{i,a}:=2/\Delta_{i,a}^2$.
These gaps are positive for every $\theta\in\Thstrict$: serial
dictatorship assigns player $i$ its unique best remaining arm.
Other off-stable gaps can be negative.
For an allocation $x$ define the \emph{pair marginals}
$z_{i,a}(x):=\sum_{m\neq\mstar:\,m(i)=a}x_m$.

\begin{theorem}[Exact pairwise-quota reduction]\label{thm:reduction}
For every $\theta\in\Thstrict$,
\[
\Xset(\theta)
=\{x\ge0:
z_{i,a}(x)\ge q_{i,a}
\ \forall(i,a)\in\Eset\}.
\]
\end{theorem}

\emph{Proof sketch} (Appendix~\ref{app:c2-proofs}). Necessity: for
each $(i,a)\in\Eset$, raising only $\mu_{i,a}$ just above
$\mu_{i,\mstar(i)}$ creates a confusing alternative whose first
deviation is $(i,a)$ (Lemma~\ref{lem:singleentry}). Its constraint gives
$z_{i,a}(x)\ge2/(\Delta_{i,a}+\varepsilon)^2$, and
$\varepsilon\downarrow0$ yields the closed quota. The change may
cascade to lower-priority players; only its first deviation is needed.

Sufficiency: for any $\lambda\in\Lset(\theta)$, the earliest-deviation
pair alone gives $D_m(\theta,\lambda)\ge\Delta_{i,a}^2/2$ whenever
$m(i)=a$, so its quota certifies $\lambda$'s information constraint.
The quotas therefore characterize information feasibility, but complete
matchings couple their execution. In an eligible pair $(i,a)$ the arm
is unassigned or the stable arm of a lower-priority player, so a quota
displaces only players further down the order, and the matching that
serves it decides who bears the cost. The next subsection computes
these scheduling costs.

\subsection{Marginal Scheduling and Regret Geometry}\label{sec:finite-lp}

Let $\theta\in\Thsep$ and $w\in\R^N_{++}$. The same allocation
$x$ determines the losses of all players through \eqref{eq:regretcount}.
Let $G(\theta)$ be the \emph{gap matrix} with columns $g(m;\theta)$,
$m\neq\mstar$. On $\Thsep$ every entry is nonnegative and every column
has a strictly positive entry (Lemma~\ref{lem:positivity}). Define
the weighted allocation value
\begin{equation}\label{eq:Cw}
C_w(\theta)
:=\inf_{x\in\Xset(\theta)}
\sum_{m\neq\mstar}x_m\,\bigl\langle w,\,g(m;\theta)\bigr\rangle .
\end{equation}

By Theorem~\ref{thm:reduction}, $C_w(\theta)$ is the
\emph{matching-cover LP}
\begin{equation}\label{eq:coverlp}
\begin{split}
C_w(\theta)=
\min_{x\ge0}\ &\sum_{m\neq\mstar} x_m
\bigl\langle w,g(m;\theta)\bigr\rangle\\[-2pt]
&\text{s.t.}\ \
z_{i,a}(x)\ge q_{i,a}\ \ \forall(i,a)\in\Eset
\end{split}
\end{equation}
(Corollary~\ref{cor:quota}). The objective and the quotas depend on $x$
only through its pair marginals, so the exponentially many matching
variables can be projected out. The variable $\rho$ below is the total
schedule mass at the $\log T$ scale, rather than a per-round probability.

\begin{theorem}[Compact edge-marginal LP]\label{thm:compact}
For $\theta\in\Thsep$ and $w\in\R^N_{++}$, with
$\Delta_{i,a}:=\mu_{i,\mstar(i)}-\mu_{i,a}\ge0$ now defined for all
pairs,
\begin{equation}\label{eq:compactlp}
\begin{aligned}
C_w(\theta)=\min_{z\ge0,\;\rho\ge0}\ &
\sum_{i,a} w_i\,\Delta_{i,a}\,z_{i,a}
\quad\text{s.t.}\\
z_{i,a}&\ge q_{i,a}, && (i,a)\in\Eset,\\
\textstyle\sum_{a} z_{i,a}&=\rho, && i\in[N],\\
\textstyle\sum_{i} z_{i,a}&\le\rho, && a\in[K],
\end{aligned}
\end{equation}
a linear program with $NK+1$ variables and $O(NK)$ constraints.
\end{theorem}

Integrality of the rectangular assignment polytope lifts feasible
marginals to complete matchings (Lemma~\ref{lem:integrality},
Appendix~\ref{app:c2-integrality}). Write $\zeta^m$ for the incidence
matrix of $m$, with $\zeta^m_{i,a}=\ind{m(i)=a}$.

\begin{proposition}[Sparse schedule recovery]\label{prop:sparse}
Every feasible $(z,\rho)$ of \eqref{eq:compactlp} with $\rho>0$ can
be written as a schedule $z=\sum_k c_k\,\zeta^{m_k}$, $c_k>0$,
$\sum_kc_k=\rho$, using at most $NK-N+1$ matchings; discarding a
possible $\mstar$-term preserves quotas and weighted cost.
\end{proposition}

The recovery is constructive (Appendix~\ref{app:c2-sparse}); computing
a minimum-size decomposition is NP-hard \citep{dufosse2016notes} and
is not needed.

Define the \emph{upper-closed Graves--Lai region}
\begin{equation}\label{eq:UGL}
\UGL(\theta)
:=G(\theta)\,\Xset(\theta)+\R_+^N .
\end{equation}
It is a closed convex polyhedron (Lemmas~\ref{lem:closed}
and~\ref{lem:polyhedral}, Appendix~\ref{app:c1-region}).

\begin{proposition}[Necessary coefficient image]\label{prop:region}
Let $\theta\in\Thsep$ and let the policy be uniformly good on
$\Thstrict$ and satisfy Assumption~\ref{ass:complete}. If
$R(T_k;\theta)/\log T_k\to r\in\R^N$ along some subsequence, then
$r\in G(\theta)\Xset(\theta)\subseteq\UGL(\theta)$.
\end{proposition}

The proposition concerns cluster points of the full vector because
the coordinatewise liminf vector need not lie in
$\UGL(\theta)$ (Remark~\ref{rem:liminf}; proofs in
Appendix~\ref{app:c1-region}).

\begin{theorem}[Weighted Graves--Lai lower bound]\label{thm:weighted}
Let $\theta\in\Thsep$ and $w\in\R^N_{++}$. Every policy that is uniformly good on $\Thstrict$ and
satisfies Assumption~\ref{ass:complete} obeys
\[
\liminf_{T\to\infty}\;\frac{w^\top R(T;\theta)}{\log T}\;\ge\;C_w(\theta).
\]
\end{theorem}

For $w>0$,
$C_w(\theta)=\inf\{\langle w,u\rangle:u\in\UGL(\theta)\}$: the
weighted bounds are the lower scalarizations of the region, whose
Pareto-minimal boundary we call the \emph{Graves--Lai necessary
lower boundary}. Every point of this boundary minimizes some
strictly positive weighted cost:

\begin{proposition}[Every Pareto-minimal point is properly supported]
\label{prop:proper}
Let $\theta\in\Thsep$. (i)~For every $u\in\UGL(\theta)$ there is a
Pareto-minimal $u'\in\UGL(\theta)$ with $u'\le u$. (ii)~A point
$u\in\UGL(\theta)$ is Pareto-minimal if and only if
$\langle w,u\rangle=C_w(\theta)$ for some $w\in\R^N_{++}$. Hence
the necessary lower boundary is the union over $w>0$ of the faces
$\arg\min_{\UGL(\theta)}\langle w,\cdot\rangle$.
\end{proposition}

Strictly positive weights support every Pareto-minimal point;
they need not select it uniquely. The proof uses polyhedrality
(Appendix~\ref{app:c1-region}). The dual of \eqref{eq:coverlp} prices
each quota: an optimal schedule uses a matching only if its weighted
loss is exactly paid for by the quotas it serves
(Appendix~\ref{app:c2-dual}). The compact LP \eqref{eq:compactlp} is
solvable in polynomial time for rational data.

\subsection{A Three-Player Externality Frontier}\label{sec:frontier}

The two-by-two instance has a singleton frontier
(Corollary~\ref{cor:closed22}). We next construct a three-player
instance whose frontier is a segment.

\begin{example}\label{ex:33}
Take $p_1\succ p_2\succ p_3$, $K=N=3$, and
\[
\mu^{(3)}=\begin{pmatrix}1&0.9&0\\0&1&0.9\\0&0.99&1\end{pmatrix},
\qquad \mstar=(a_1,a_2,a_3),
\]
with quotas $q_{1,a_2}=q_{2,a_3}=2/0.1^2=200$, $q_{1,a_3}=2$
(player $3$, holding the last remaining arm, generates none). The five non-stable matchings are
\[
\begin{aligned}
m_{\mathrm A}&=(a_2,a_1,a_3),& m_{\mathrm B}&=(a_2,a_3,a_1),\\
m_{\mathrm C}&=(a_1,a_3,a_2),& m_{\mathrm D}&=(a_3,a_1,a_2),\\
m_{\mathrm E}&=(a_3,a_2,a_1).
\end{aligned}
\]
Their gap vectors and quotas served are given in
Appendix~\ref{app:c2-frontier}.
\end{example}

\begin{proposition}[Exact frontier of the $3\times3$ instance]
\label{prop:frontier33}
For the instance $\mu^{(3)}$, the Pareto-minimal boundary of
$\UGL(\mu^{(3)})$ is exactly the segment
\[
\bigl\{\,r(t)=(22,\;220-t,\;4+0.99t)\;:\;t\in[-2,200]\,\bigr\},
\]
traced by the allocations $x_{m_{\mathrm B}}=h$,
$x_{m_{\mathrm A}}=x_{m_{\mathrm C}}=200-h$,
$x_{m_{\mathrm D}}=y$, $x_{m_{\mathrm E}}=2-y$, with $t=h-y$,
$h\in[0,200]$, $y\in[0,2]$, all meeting the three quotas with
equality.
\end{proposition}

The proof (Appendix~\ref{app:c2-frontier}) maps every feasible regret
vector to a componentwise no-larger family point and shows that no two
segment points are componentwise ordered.

The same binding quotas are filled at every point, but the cost is
distributed differently. Player $1$'s two quotas can never
share a matching, so $r_1=22$ throughout. At the \emph{parallel}
end ($t=200$),
the schedule $\{m_{\mathrm B}\!:200,\ m_{\mathrm E}\!:2\}$ fills both
large quotas at once: player $2$ reaches
its minimum $r_2=20$ while player $3$ spends those rounds on its
worst arm ($r_3=202$). At the \emph{protective} end ($t=-2$), the
schedule $\{m_{\mathrm A}\!:200,\ m_{\mathrm C}\!:200,\ m_{\mathrm D}\!:2\}$
serves every quota through matchings cheap for player $3$
($r_3=2.02$) while player $2$ pays both large quotas in sequence
($r_2=222$).
Interior points are supported only by the weight ratio
$w_2/w_3=0.99$; other weights select an endpoint, since the
region is a polyhedron (Lemma~\ref{lem:polyhedral}) and a generic
weight selects a vertex.

The construction extends to every $N\ge3$. There the top player's
regret $r_1=c_1^{(N)}:=\sum_{a\neq a_1}2/\Delta_{1,a}$ is the same at
every Pareto-minimal point, so only lower-priority players trade off,
and every face of the necessary lower boundary has dimension at most
$N-2$.
Exact rational certificates exhibit faces of dimension exactly $N-2$
for $N\in\{3,4,5,6\}$
(Proposition~\ref{prop:family}, Appendix~\ref{app:c2-family});
a closed form for general $N$ is open.

\section{Attainability and Frontier Selection}\label{sec:algorithm}

The schedules of \S\ref{sec:finite-lp} are described in terms of
unknown means. This section shows that they can be attained without
knowing the means, in expected logarithmic regret at separated
instances. For every positive weight vector, a policy that is uniformly
good on $\Thstrict$ attains the weighted optimum, and a data-only
target rule selects the split of regret
(Theorem~\ref{thm:frontier-achieve}). Every point of
$G(\theta)\Xset(\theta)$ is attainable (Theorem~\ref{thm:closure}).
Neither result needs a unique optimum. The anytime policies learn the
means on the way: they solve for a schedule of complete matchings from
the estimated gaps, play it, and in each remaining round exploit the
estimated stable matching when a confidence test certifies it and
sample otherwise. Table~\ref{tab:scope}
(Appendix~\ref{app:scope}) lists the hypotheses of each guarantee, and
\S\ref{sec:experiments} quantifies the finite-horizon surcharge.
Tracking follows \citet{garivier2016track}; the tracked object here is
a schedule of complete matchings.

\subsection{The Estimate--Solve--Track Framework}\label{sec:c4-policy}

\paragraph{Parameters and thresholds.}
Fix $\alpha>\xi>3$, let $\ell(t):=\log(t+e)$, and set
$b(t):=\ell(t)+\alpha\log\ell(t)$,
$c(t):=\ell(t)+\xi\log\ell(t)$,
$\beta(t):=\log(\ell(t)+e)$: $b$ scales the exploration targets,
$c$ the certification test ($\xi>3$ makes the peeling bound of
Lemma~\ref{lem:peel} summable), $\beta$ the noise floor; $b>c$ and
$b(t)/\log t\to1$. Rounds are grouped into epochs $(T_{k-1},T_k]$
with $T_k:=\lceil e^{k^2}\rceil$, so $b(T_k)\sim k^2$.

After the repair phase of epoch $k$ the policy computes the empirical
matching $\widehat m_k=\mstar(\widehat\mu)$, the empirical eligible
set $\widehat\Eset_k$, and on it the \emph{noise-floor regularized}
gap
$\widehat g_{i,a}
:=\max\bigl(
\widehat\mu_{i,\widehat m_k(i)}-\widehat\mu_{i,a},\;
\sqrt{2\beta(T_{k-1})/N_{i,a}}
+\sqrt{2\beta(T_{k-1})/N_{i,\widehat m_k(i)}}\bigr)$
and capped quota
$\widehat q_{i,a}:=\min(2/\widehat g_{i,a}^{\,2},\,\kappa_k)$,
$\kappa_k:=k$. The floor keeps early quotas small without biasing
the limit. Write $\widehat{\mathcal P}_k$ for the feasible set of the compact LP
\eqref{eq:compactlp} with empirical quotas $\widehat q$, and use
surrogate costs for $w\in\R^N_{++}$,
$\widehat c_{k,i,a}:=w_i(\widehat\mu_{i,\widehat m_k(i)}-
\widehat\mu_{i,a})_+$. Each solve rule produces a schedule meeting
these quotas; the canonical rule has mass at most $NK\kappa_k$
(Lemma~\ref{lem:mass}), while the regularized and uncosted rules
impose this cap explicitly (Lemma~\ref{lem:capfeas}). The marginal
solution is decomposed into at most $NK-N+1$ matchings
(Proposition~\ref{prop:sparse}); empirical stable components are
discarded before exploration.

\paragraph{Phases and the certificate.}
Each epoch runs \emph{repair} (top every pair count up to
$\sqrt{b(T_k)}$), \emph{solve}, \emph{explore} (play the schedule
with increment execution, plus a stable-matching replay of
$\widehat m_k$ that keeps stable-side certification radii
negligible), and \emph{certified exploitation with a dynamic
candidate}: at each remaining round the policy recomputes
$\widetilde m_t:=\mstar(\widehat\mu(t))$ (no LP re-solve) and plays
it if for every level $i$ and every free arm
$a\neq\widetilde m_t(i)$,
$\widehat\mu_{i,\widetilde m_t(i)}(t)-\widehat\mu_{i,a}(t)
\ge\sqrt{2c(t)/N_{i,\widetilde m_t(i)}(t)}
+\sqrt{2c(t)/N_{i,a}(t)}$,
and otherwise samples the least-sampled pair of the first
uncertain comparison (an \emph{estimation round}).
Algorithm~\ref{alg:policy} (Appendix~\ref{app:c4}) specifies all
phases, tie-breaks, the round-robin prefix and the starting epoch; a
round costs $O(NK)$ operations, and an epoch solves $O(NK)$ linear
programs of size $O(NK)$.

\begin{remark}[Design necessities]\label{rem:free}
\emph{(a) The candidate must be dynamic}: a frozen wrong candidate
may remain uncertified despite further sampling, costing
$\approx e^{k^2}$ estimation rounds.
\emph{(b) Certification is asymptotically free}: quota supply and
stable replay support the demand--supply accounting of
Lemma~\ref{lem:certify}, leaving $o(\log T)$ residual estimation
rounds; certification also bounds expected wrong exploitation, and
without it the exploit-only ablation of \S\ref{sec:experiments} can
keep a wrong matching.
\end{remark}

The default canonical policy $\pi_w$ uses three ordered programs:
minimum cost, minimum mass with vanishing cost slack, then a fixed
lexicographic choice. Its full definition, Assumption~\ref{ass:unique}
and Theorem~\ref{thm:achieve} are in
Appendix~\ref{app:c4-canonical}; that theorem concerns this solve
rule, separately from the regularized rule below.

\paragraph{Proof roadmap.}\label{sec:c4-proof}
Repair identifies the structure (Lemma~\ref{lem:ident}); convergence
of the selected targets then gives expected-count tracking by
increment execution and summable bad-epoch accounting
(Lemma~\ref{lem:track}). Certification and regret accounting
(Lemmas~\ref{lem:certify} and~\ref{lem:account}) turn off-stable
counts into expected-regret coefficients. Quota supply, stable replay
and the demand--supply argument control the residual estimation
rounds. Uniform goodness is proved
separately on every $\lambda\in\Thstrict$ in
Lemma~\ref{lem:ug} (Appendix~\ref{app:c4-strict}), for all four solve
rules, without separation or canonical uniqueness.

\subsection{Weighted Optimality and Frontier Selection}\label{sec:c4-frontier}

A positive weight can support an entire face without selecting one
point: at $w^\ast=(1,0.99,1)$ every point of the three-player segment
is optimal. The weight determines the cost-optimal face, while a
fixed data-only rule supplies a target in marginal space.

For $\theta\in\Thsep$, let $\mathcal P(\theta)$ be the feasible set of \eqref{eq:compactlp}
and $c(\theta,w)$ its surrogate cost. Define
$F_w(\theta):=\operatorname{proj}_{\mathrm{off}}
\arg\min_{(z,\rho)\in\mathcal P(\theta)}c(\theta,w)^{\top}z$, the
off-stable projection of the cost-optimal face without
$\rho$-minimization, a nonempty compact polytope
(Lemma~\ref{lem:fwcompact}, Appendix~\ref{app:c4-fw}).

Projection is in off-stable marginal coordinates, not in regret space;
the regret coefficient of player $i$ is then $\sum_a\Delta_{i,a}z_{i,a}$.

\paragraph{The policy $\pi_{w,\Psi}$.}
A \emph{target rule} $\Psi$ maps each empirical structure
$(\widehat m_k,\widehat\Eset_k,\widehat q_k,\widehat\Delta_k)$ to
$\tau_k:=\Psi(\widehat m_k,\widehat\Eset_k,\widehat q_k,
\widehat\Delta_k)\in\R_+^{N\times K}$, locally Lipschitz in
$(\widehat q_k,\widehat\Delta_k)$ for each fixed structure; it reads
only off-stable coordinates, and a constant rule fixes $\tau$.
Epoch $k$ replaces the solve step by the \emph{capped regularized program}
\begin{equation}\label{eq:regprog}
\min_{\substack{(z,\rho)\in\widehat{\mathcal P}_k\\
\rho\le NK\kappa_k}}\;
\widehat c_k^{\top}z
+\eta_k\bigl\|z^{\mathrm{off}}-\tau_k^{\mathrm{off}}\bigr\|^2,
\quad \eta_k:=1/\log(k+1),
\end{equation}
strictly convex in $z^{\mathrm{off}}$; among its minimizers the
policy takes the smallest $\rho$. The cap must be imposed, since the
charging argument of Lemma~\ref{lem:mass} breaks under
regularization (Lemma~\ref{lem:capfeas}, Remark~\ref{rem:whycap}).
Intuitively, the cost term confines the limit to the cost-optimal
face, the vanishing regularizer selects the point of that face
closest to the target, and the cap bounds each epoch's exploration on
every sample path, as uniform goodness requires.

\begin{theorem}[Rule-selected frontier achievability]
\label{thm:frontier-achieve}
Fix $w\in\R^N_{++}$, a target rule $\Psi$ and $\alpha>\xi>3$. The
policy $\pi_{w,\Psi}$ is uniformly good on $\Thstrict$. For
$\theta\in\Thsep$, let $\tau:=\Psi(\mstar,\Eset,q,\Delta)$
be the rule's value at the true structure and $y^\tau(\theta)$ the
Euclidean projection of $\tau^{\mathrm{off}}$ onto $F_w(\theta)$,
off-stable coordinates taken with respect to $\mstar(\theta)$. Then
for every off-stable pair,
\[
\frac{\E_\theta[N_{i,a}(T)]}{\log T}\to y^\tau_{i,a}(\theta),
\quad
\frac{R_i(T;\theta)}{\log T}
\to\sum_a\Delta_{i,a}y^\tau_{i,a}(\theta),
\]
and $R_w(T;\theta)/\log T\to C_w(\theta)$. No uniqueness assumption
is made, and the rule sees data only.
\end{theorem}

A fixed $\Psi$ defines one policy for all instances; for example,
$\Psi\equiv0$ selects the minimum-norm point of each $F_w(\theta)$.
To attain a prescribed Pareto-minimal point of a given instance, take
a supporting weight $w>0$ and the constant target $y\in F_w(\theta)$
with that regret vector, which projects to itself
(Corollary~\ref{cor:frontier-general}). Such a target depends on the
instance.

For the three-player instance, the explicit policy $\pi_s$ uses a
fixed $s\in[0,1]$ to split the empirical quotas and attains
$r(202s-2)=(22,222-202s,2.02+199.98s)$
(Proposition~\ref{prop:frontier-attain}, Appendix~\ref{app:c4-frontier33}).
Algorithm~\ref{alg:policy} gives its full definition, including the
capped fallback; the diagnostic experiments run a practical variant
of it. At $w^\ast$ a Lipschitz target rule reaches the same limit
through a different sequence of solves (Remark~\ref{rem:usplit-rule}).

The proof (Appendix~\ref{app:c4-c4b}) combines exact regularization
\citep{mangasarian1979nonlinear,friedlander2007exact} with a Hoffman
bound \citep{hoffman1952approximate}.

\subsection{The Exact Attainable Set}\label{sec:c4-attainable}

Tracking a cost-optimal face realizes its regret coefficients. To
characterize all attainable coefficients, the solve rule must also
track feasible targets that are not cost-optimal.

Call $r$ \emph{attainable at $\theta$} if $R(T;\theta)/\log T\to r$
under a complete-matching policy uniformly good on $\Thstrict$, and
write $\mathcal A(\theta)$ for all such vectors. For
$\tau\in\R_+^{N\times K}$ let $\pi_\tau$ remove the cost from
\eqref{eq:regprog}, so the solve step becomes
$\min\{\|z^{\mathrm{off}}-\tau^{\mathrm{off}}\|^2:(z,\rho)\in
\widehat{\mathcal P}_k,\ \rho\le NK\kappa_k\}$, smallest $\rho$
among minimizers, and let
$Y_\infty(\theta):=\operatorname{proj}_{\mathrm{off}}\mathcal P(\theta)$
(off-stable marginals of feasible schedules).

\Needspace{11\baselineskip}
\begin{theorem}[Exact attainable set]\label{thm:closure}
Fix $\alpha>\xi>3$. (i)~For every $\tau\in\R_+^{N\times K}$ the
policy $\pi_\tau$ is uniformly good on $\Thstrict$ and, at every
$\theta\in\Thsep$, $\E_\theta[N_{i,a}(T)]/\log T$ converges on
every off-stable pair to the corresponding coordinate of the
Euclidean projection of $\tau^{\mathrm{off}}$ onto $Y_\infty(\theta)$;
in particular every $\tau^{\mathrm{off}}\in Y_\infty(\theta)$ is
tracked exactly. (ii)~For every $\theta\in\Thsep$,
$\mathcal A(\theta)=G(\theta)\Xset(\theta)$, a polyhedron, so that
$\UGL(\theta)=\mathcal A(\theta)+\R^N_+$ and the Pareto-minimal set
of $\mathcal A(\theta)$ is the Graves--Lai necessary lower boundary.
\end{theorem}

The proof (Appendix~\ref{app:c4-closure}) is short: $\subseteq$ in
(ii) is Proposition~\ref{prop:region},
and $\supseteq$ tracks the marginals of $x\in\Xset(\theta)$ by (i).

\Needspace{8\baselineskip}
\begin{corollary}[Rigidity]\label{cor:rigid}
For $\theta\in\Thsep$, $\mathcal A(\theta)=\UGL(\theta)$ iff
$K>N$, or equivalently iff
$\R^N_+\subseteq\mathrm{cone}\{g(m;\theta):m\neq\mstar\}$.
At $\mu^{(3)}$, $(22,20,202)$ is
attainable and $(22,20,202+\varepsilon)$ is not for any
$\varepsilon>0$, since $r_1=22$, $r_2=20$ force $x_{m_{\mathrm B}}=200$,
$x_{m_{\mathrm E}}=2$, all other $x_m=0$: at the parallel end no schedule
can add regret for player $3$ alone.
\end{corollary}

\FloatBarrier
\section{DIAGNOSTIC EXPERIMENTS}\label{sec:experiments}
We simulate \emph{practical variants} of the policy on
$\mu^{(3)}$, horizon $T=10^7$, $50$ seeds per configuration ($20$
per split in the $s$-sweep of Figure~\ref{fig:frontier});
secondary experiments are in Appendix~\ref{app:experiments}.
All plotted regret is the sample-path \emph{pseudo-regret}
$\widetilde R_i(T):=\sum_a\Delta_{i,a}N_{i,a}(T)$, normalized by $b(T)$
($b(T)/\log T=1.19$ at $10^7$). The experiments diagnose the
scheduling mechanism and its finite-horizon surcharge.
The practical variants of Algorithm~\ref{alg:policy} use geometric epochs, thresholds
$\alpha=1.1$ and $\xi=1.05$, deficit-based execution, batched
exploitation and a cap of $2000$; certification is varied across
configurations. The theorems cover
Algorithm~\ref{alg:policy}, not these variants.
Table~\ref{tab:ablation}
(Appendix~\ref{app:experiments}) isolates the mechanisms.
The exploit-only ablation stops repair and exploration after
commitment and can retain a wrong matching; removing the noise floor
increases exploration.

\begin{figure}[t]
\centering
\includegraphics[width=\columnwidth]{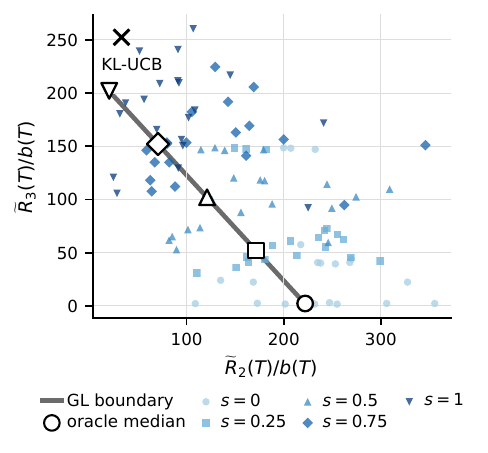}
\caption{Scheduling trade-off at $T=10^7$: final
$(\widetilde R_2/b(T),\,\widetilde R_3/b(T))$ for $\pi_s$,
$s\in\{0,0.25,0.5,0.75,1\}$, $20$ seeds each (filled), against the
segment $r_3=221.8-0.99\,r_2$ of the asymptotic Graves--Lai (GL)
boundary; open markers, shaped as the corresponding split, are oracle
$s$-split medians over $20$ seeds and sit on it. The plug-in
splits are ordered along the segment but displaced outward by the
finite-horizon transient. The cross is layered KL-UCB (mean over
$50$ seeds).}
\label{fig:frontier}
\end{figure}

\paragraph{Oracle and plug-in trackers.}
Both trackers use $w=(1,1,1)$. Given the true $\mstar$ and the
canonical marginals $z^\circ$
(Assumption~\ref{ass:unique}), the oracle tracker reaches
median $\widetilde R(T)/b(T)=(22.1,20.3,202.0)$ and
$w^\top\widetilde R/(C_wb)=1.002$. The plug-in tracker reaches median
$w^\top\widetilde R/(C_wb)=1.23$. Because the two
endpoints differ in weighted cost by only $0.8\%$, individual seeds
select either vertex; their coordinatewise median lies between the
endpoints and is not itself a coefficient limit. The certified
variant had no wrong-exploitation round in any of its $50$ seeds, but
its finite-horizon certification surcharge is substantial
(Table~\ref{tab:ablation}).

\paragraph{Scheduling and baselines.}
The oracle $s$-splits trace the segment
(Figure~\ref{fig:frontier}): the same quotas, scheduled with
different $s$, move regret between players $2$ and $3$ along the
asymptotic Graves--Lai boundary. On $33$ random markets up to $5\times5$ the
LPs \eqref{eq:coverlp} and \eqref{eq:compactlp} agree to machine precision; the plug-in runs also show
finite-horizon excess over $C_w$. Two baselines at the same horizon
$T=10^7$ compare scheduling costs (Table~\ref{tab:random},
Appendix~\ref{app:experiments}). On the three probe instances,
dedicated-cover exploration, which serves each quota with its own
matching, pays $1.7\times$ to $2.0\times$ the tracker's weighted
pseudo-regret. Layered KL-UCB, the natural centralized index policy
(each priority level plays its remaining arm of highest index), is
competitive in weighted pseudo-regret, but its mean vector on
$\mu^{(3)}$, $(27.8,32.9,252.4)$, lies beyond the parallel end (the
cross in Figure~\ref{fig:frontier}), and it has no explicit control
parameter for moving along the segment.

\paragraph{Certification and face selection.}
A truncated run with theorem-style certification thresholds and a
practical cap (Appendix~\ref{app:experiments}) has no
wrong-exploitation round in any of its $40$ runs ($20$ seeds each
for the oracle and plug-in modes), but its finite-horizon surcharge
remains large. For regularized face
selection, the oracle marginal optimizer matches all three interior
targets to solver precision; after execution, the maximum normalized
coordinate error is below $0.065$, and plug-in medians remain ordered by the
target. These runs quantify the finite-horizon error in selecting
and executing the schedules.

\section{CONCLUSION}\label{sec:discussion}

At top-choice-separated instances, scheduling determines who pays for
exploration, and we characterize exactly which splits are possible.
Pairwise quotas specify what must be learned; the attainable
logarithmic regret vectors are exactly the image of the
matching-allocation polyhedron $\Xset(\theta)$ under the gap matrix,
and the marginal LP exposes how regret can be distributed across
players. At these instances,
estimate--solve--track policies attain every weighted optimum, and
capped regularization lets a target rule choose among optimal
schedules when the optimum is not unique.

\paragraph{Limitations.}
The analysis assumes Gaussian rewards with known unit variance, a
known common priority order, complete matchings and row-strict
preferences; the regret-region and attainability results also require
top-choice separation. Extensions beyond common priorities and to
two-sided uncertainty remain open. All guarantees are
logarithmic-asymptotic: we give no finite-horizon or
uniform-in-instance rates, and prescribing an arbitrary boundary point
may require instance calibration. The thresholds $\alpha>\xi>3$ and
the epoch and cap schedules $T_k=\lceil e^{k^2}\rceil$ and $\kappa_k=k$
serve the proofs; at
$T=10^7$ the cap would be $\kappa_5=5$, far below the quota $200$ of
$\mu^{(3)}$, so the experiments use practical variants
(\S\ref{sec:experiments}, Appendix~\ref{app:experiments}), whose
finite-horizon surcharge remains substantial.

\clearpage
\section*{AI Use Statement}

In this work, we used generative AI tools for the following tasks that
require disclosure: providing critical ingredients for proofs of
mathematical claims and assisting in the writing of proofs; designing
the experiments and providing feedback on them; implementing the
methods and the simulation code, including the seeded scripts that
generate all synthetic problem instances (no data set was produced
directly by an AI model); and interpreting results. We have not used
generative AI tools to develop theoretical models or conceptual
frameworks, to formulate mathematical claims, to propose or refine
hypotheses, or for translation; cleaning or reformatting data sets and
qualitative or thematic data analysis are not applicable to this work.
Additionally, we used generative AI tools to suggest experimental
parameters, to create and edit software code and figures, to check
proofs and experiments, to identify relevant literature, and to edit
the manuscript for readability.

We have reviewed all AI-assisted work. The authors checked every proof
in Appendices~\ref{app:twobytwo}--\ref{app:c4} line by line, and
reviewed and tested the simulation code. Every number reported in the paper was
reconciled against the archived result files of the supplementary
material, and the exact certificates of Proposition~\ref{prop:family}
are re-verified in rational arithmetic by \texttt{frontier\_family.py}.
Every reference, including the specific theorem, equation and appendix
numbers we cite, was checked against the original source. We take
responsibility for the final content of this work, including text,
claims, code and results produced with the aid of generative AI.

\bibliographystyle{plainnat}
\bibliography{refs}

\begin{thebibliography}{42}
\providecommand{\natexlab}[1]{#1}
\providecommand{\url}[1]{\texttt{#1}}
\expandafter\ifx\csname urlstyle\endcsname\relax
  \providecommand{\doi}[1]{doi: #1}\else
  \providecommand{\doi}{doi: \begingroup \urlstyle{rm}\Url}\fi

\bibitem[Athanasopoulos et~al.(2025)Athanasopoulos, George, and
  Dimitrakakis]{athanasopoulos2025probably}
Andreas Athanasopoulos, Anne-Marie George, and Christos Dimitrakakis.
\newblock Probably correct optimal stable matching for two-sided markets under
  uncertainty.
\newblock In \emph{Proceedings of the 24th International Conference on
  Autonomous Agents and Multiagent Systems (AAMAS)}, pages 152--160, 2025.
\newblock arXiv:2501.03018.

\bibitem[Athanasopoulos et~al.(2026)Athanasopoulos, George, and
  Dimitrakakis]{athanasopoulos2026stable}
Andreas Athanasopoulos, Anne-Marie George, and Christos Dimitrakakis.
\newblock Probably correct optimal stable matching under two-sided uncertainty.
\newblock In \emph{Proceedings of the 42nd Conference on Uncertainty in
  Artificial Intelligence (UAI)}, volume 337 of \emph{Proceedings of Machine
  Learning Research}, pages 228--254. PMLR, 2026.

\bibitem[Baek and Farias(2024)]{baek2024fair}
Jackie Baek and Vivek~F. Farias.
\newblock Fair exploration via axiomatic bargaining.
\newblock \emph{Management Science}, 70\penalty0 (12):\penalty0 8922--8939,
  2024.

\bibitem[Basu(2025)]{basu2025competing}
Soumya Basu.
\newblock Competing bandits in matching markets via super stability.
\newblock In \emph{Proceedings of the 42nd International Conference on Machine
  Learning (ICML)}, volume 267 of \emph{Proceedings of Machine Learning
  Research}, pages 3226--3250. PMLR, 2025.
\newblock arXiv:2506.15926v1.

\bibitem[Capp{\'e} et~al.(2013)Capp{\'e}, Garivier, Maillard, Munos, and
  Stoltz]{cappe2013kullback}
Olivier Capp{\'e}, Aur{\'e}lien Garivier, Odalric-Ambrym Maillard, R{\'e}mi
  Munos, and Gilles Stoltz.
\newblock Kullback--{L}eibler upper confidence bounds for optimal sequential
  allocation.
\newblock \emph{The Annals of Statistics}, 41\penalty0 (3):\penalty0
  1516--1541, 2013.

\bibitem[Cen and Shah(2022)]{cen2022regret}
Sarah~H. Cen and Devavrat Shah.
\newblock Regret, stability \& fairness in matching markets with bandit
  learners.
\newblock In \emph{Proceedings of the 25th International Conference on
  Artificial Intelligence and Statistics (AISTATS)}, volume 151 of
  \emph{Proceedings of Machine Learning Research}, pages 8938--8968. PMLR,
  2022.

\bibitem[Combes et~al.(2015)Combes, {Talebi Mazraeh Shahi}, Prouti{\`e}re, and
  Lelarge]{combes2015combinatorial}
Richard Combes, Mohammad~Sadegh {Talebi Mazraeh Shahi}, Alexandre
  Prouti{\`e}re, and Marc Lelarge.
\newblock Combinatorial bandits revisited.
\newblock In \emph{Advances in Neural Information Processing Systems 28
  (NIPS)}, 2015.

\bibitem[Combes et~al.(2017)Combes, Magureanu, and
  Prouti{\`e}re]{combes2017minimal}
Richard Combes, Stefan Magureanu, and Alexandre Prouti{\`e}re.
\newblock Minimal exploration in structured stochastic bandits.
\newblock In \emph{Advances in Neural Information Processing Systems 30
  (NIPS)}, 2017.
\newblock arXiv:1711.00400.

\bibitem[Cuvelier et~al.(2021)Cuvelier, Combes, and
  Gourdin]{cuvelier2021asymptotically}
Thibaut Cuvelier, Richard Combes, and Eric Gourdin.
\newblock Asymptotically optimal strategies for combinatorial semi-bandits in
  polynomial time.
\newblock In \emph{Proceedings of the 32nd International Conference on
  Algorithmic Learning Theory (ALT)}, volume 132 of \emph{Proceedings of
  Machine Learning Research}, pages 505--528. PMLR, 2021.

\bibitem[Das and Kamenica(2005)]{das2005two}
Sanmay Das and Emir Kamenica.
\newblock Two-sided bandits and the dating market.
\newblock In \emph{Proceedings of the 19th International Joint Conference on
  Artificial Intelligence (IJCAI)}, pages 947--952, 2005.

\bibitem[Degenne and Koolen(2019)]{degenne2019multiple}
R{\'e}my Degenne and Wouter~M. Koolen.
\newblock Pure exploration with multiple correct answers.
\newblock In \emph{Advances in Neural Information Processing Systems 32
  (NeurIPS)}, 2019.

\bibitem[Dufoss{\'e} and U{\c c}ar(2016)]{dufosse2016notes}
Fanny Dufoss{\'e} and Bora U{\c c}ar.
\newblock Notes on {B}irkhoff--von {N}eumann decomposition of doubly stochastic
  matrices.
\newblock \emph{Linear Algebra and its Applications}, 497:\penalty0 108--115,
  2016.

\bibitem[Friedlander and Tseng(2007)]{friedlander2007exact}
Michael~P. Friedlander and Paul Tseng.
\newblock Exact regularization of convex programs.
\newblock \emph{SIAM Journal on Optimization}, 18\penalty0 (4):\penalty0
  1326--1350, 2007.

\bibitem[Gale and Shapley(1962)]{gale1962college}
David Gale and Lloyd~S. Shapley.
\newblock College admissions and the stability of marriage.
\newblock \emph{The American Mathematical Monthly}, 69\penalty0 (1):\penalty0
  9--15, 1962.

\bibitem[Garivier and Kaufmann(2016)]{garivier2016track}
Aur{\'e}lien Garivier and Emilie Kaufmann.
\newblock Optimal best arm identification with fixed confidence.
\newblock In \emph{Proceedings of the 29th Annual Conference on Learning Theory
  (COLT)}, volume~49 of \emph{Proceedings of Machine Learning Research}, pages
  998--1027. PMLR, 2016.

\bibitem[Garivier et~al.(2019)Garivier, M{\'e}nard, and
  Stoltz]{garivier2019explore}
Aur{\'e}lien Garivier, Pierre M{\'e}nard, and Gilles Stoltz.
\newblock Explore first, exploit next: The true shape of regret in bandit
  problems.
\newblock \emph{Mathematics of Operations Research}, 44\penalty0 (2):\penalty0
  377--399, 2019.

\bibitem[Geoffrion(1968)]{geoffrion1968proper}
Arthur~M. Geoffrion.
\newblock Proper efficiency and the theory of vector maximization.
\newblock \emph{Journal of Mathematical Analysis and Applications}, 22\penalty0
  (3):\penalty0 618--630, 1968.

\bibitem[Graves and Lai(1997)]{graves1997asymptotically}
Todd~L. Graves and Tze~Leung Lai.
\newblock Asymptotically efficient adaptive choice of control laws in
  controlled {M}arkov chains.
\newblock \emph{SIAM Journal on Control and Optimization}, 35\penalty0
  (3):\penalty0 715--743, 1997.

\bibitem[Hoffman(1952)]{hoffman1952approximate}
Alan~J. Hoffman.
\newblock On approximate solutions of systems of linear inequalities.
\newblock \emph{Journal of Research of the National Bureau of Standards},
  49\penalty0 (4):\penalty0 263--265, 1952.

\bibitem[Hosseini and Zhang(2024)]{hosseini2024bandit}
Hadi Hosseini and Duohan Zhang.
\newblock Bandit learning in matching markets: Utilitarian and {R}awlsian
  perspectives, 2024.
\newblock arXiv:2412.00301.

\bibitem[Hosseini et~al.(2024)Hosseini, Roy, and Zhang]{hosseini2024putting}
Hadi Hosseini, Sanjukta Roy, and Duohan Zhang.
\newblock Putting {G}ale \& {S}hapley to work: Guaranteeing stability through
  learning.
\newblock In \emph{Advances in Neural Information Processing Systems 37
  (NeurIPS)}, 2024.
\newblock \doi{10.52202/079017-2206}.

\bibitem[Isermann(1974)]{isermann1974proper}
Heinz Isermann.
\newblock Proper efficiency and the linear vector maximum problem.
\newblock \emph{Operations Research}, 22\penalty0 (1):\penalty0 189--191, 1974.

\bibitem[Kong and Li(2023)]{kong2023playeroptimal}
Fang Kong and Shuai Li.
\newblock Player-optimal stable regret for bandit learning in matching markets.
\newblock In \emph{Proceedings of the 2023 Annual ACM-SIAM Symposium on
  Discrete Algorithms (SODA)}, pages 1512--1522. SIAM, 2023.

\bibitem[Kong et~al.(2024)Kong, Wang, and Li]{kong2024improved}
Fang Kong, Zilong Wang, and Shuai Li.
\newblock Improved analysis for bandit learning in matching markets.
\newblock In \emph{Advances in Neural Information Processing Systems 37
  (NeurIPS)}, 2024.
\newblock \doi{10.52202/079017-2916}.

\bibitem[Kong et~al.(2025)Kong, Tang, Li, Lu, Lui, and
  Li]{kong2025indifference}
Fang Kong, Jingqi Tang, Mingzhu Li, Pinyan Lu, John C.~S. Lui, and Shuai Li.
\newblock Bandit learning in matching markets with indifference.
\newblock In \emph{Proceedings of the 13th International Conference on Learning
  Representations (ICLR)}, 2025.

\bibitem[Koolen(2013)]{koolen2013pareto}
Wouter~M. Koolen.
\newblock The {P}areto regret frontier.
\newblock In \emph{Advances in Neural Information Processing Systems 26
  (NIPS)}, 2013.

\bibitem[Lai and Robbins(1985)]{lai1985asymptotically}
Tze~Leung Lai and Herbert Robbins.
\newblock Asymptotically efficient adaptive allocation rules.
\newblock \emph{Advances in Applied Mathematics}, 6\penalty0 (1):\penalty0
  4--22, 1985.

\bibitem[Lattimore(2015)]{lattimore2015pareto}
Tor Lattimore.
\newblock The {P}areto regret frontier for bandits.
\newblock In \emph{Advances in Neural Information Processing Systems 28
  (NIPS)}, 2015.

\bibitem[Lattimore and Szepesv{\'a}ri(2020)]{lattimore2020bandit}
Tor Lattimore and Csaba Szepesv{\'a}ri.
\newblock \emph{Bandit Algorithms}.
\newblock Cambridge University Press, 2020.

\bibitem[Li et~al.(2025)Li, Wang, and Kong]{li2025survey}
Shuai Li, Zilong Wang, and Fang Kong.
\newblock A survey on bandit learning in matching markets.
\newblock In \emph{Proceedings of the 34th International Joint Conference on
  Artificial Intelligence (IJCAI)}, pages 10546--10554, 2025.

\bibitem[Lin et~al.(2026)Lin, Mauras, Perchet, and Merlis]{lin2026contextual}
Shiyun Lin, Simon Mauras, Vianney Perchet, and Nadav Merlis.
\newblock Adaptive bandit algorithms for contextual matching markets.
\newblock In \emph{Proceedings of the 43rd International Conference on Machine
  Learning (ICML)}, volume 306 of \emph{Proceedings of Machine Learning
  Research}, pages 74800--74843. PMLR, 2026.

\bibitem[Liu and Sellke(2022)]{liu2022pareto}
Allen~X. Liu and Mark Sellke.
\newblock The {P}areto frontier of instance-dependent guarantees in
  multi-player multi-armed bandits with no communication.
\newblock In \emph{Proceedings of the 35th Conference on Learning Theory
  (COLT)}, volume 178 of \emph{Proceedings of Machine Learning Research}, page
  3094. PMLR, 2022.
\newblock Full version: arXiv:2202.09653.

\bibitem[Liu et~al.(2020)Liu, Mania, and Jordan]{liu2020competing}
Lydia~T. Liu, Horia Mania, and Michael~I. Jordan.
\newblock Competing bandits in matching markets.
\newblock In \emph{Proceedings of the 23rd International Conference on
  Artificial Intelligence and Statistics (AISTATS)}, volume 108 of
  \emph{Proceedings of Machine Learning Research}, pages 1618--1628. PMLR,
  2020.

\bibitem[Mangasarian and Meyer(1979)]{mangasarian1979nonlinear}
Olvi~L. Mangasarian and Robert~R. Meyer.
\newblock Nonlinear perturbation of linear programs.
\newblock \emph{SIAM Journal on Control and Optimization}, 17\penalty0
  (6):\penalty0 745--752, 1979.

\bibitem[Pagare et~al.(2025)Pagare, Bandyopadhyay, and
  Juneja]{pagare2025optimal}
Tejas Pagare, Agniv Bandyopadhyay, and Sandeep Juneja.
\newblock Optimal algorithms for bandit learning in matching markets, 2025.
\newblock arXiv:2509.14466.

\bibitem[Raghavan et~al.(2018)Raghavan, Slivkins, {Wortman Vaughan}, and
  Wu]{raghavan2018externalities}
Manish Raghavan, Aleksandrs Slivkins, Jennifer {Wortman Vaughan}, and
  Zhiwei~Steven Wu.
\newblock The externalities of exploration and how data diversity helps
  exploitation.
\newblock In \emph{Proceedings of the 31st Conference on Learning Theory
  (COLT)}, volume~75 of \emph{Proceedings of Machine Learning Research}, pages
  1724--1738. PMLR, 2018.

\bibitem[R{\'e}veillard and Combes(2025)]{reveillard2025multimodal}
William R{\'e}veillard and Richard Combes.
\newblock Multimodal bandits: Regret lower bounds and optimal algorithms.
\newblock In \emph{Advances in Neural Information Processing Systems 38
  (NeurIPS)}, pages 97654--97684, 2025.
\newblock \doi{10.52202/085713-3264}.

\bibitem[Sankararaman et~al.(2021)Sankararaman, Basu, and
  Sankararaman]{sankararaman2021dominate}
Abishek Sankararaman, Soumya Basu, and Karthik~Abinav Sankararaman.
\newblock Dominate or delete: Decentralized competing bandits in serial
  dictatorship.
\newblock In \emph{Proceedings of the 24th International Conference on
  Artificial Intelligence and Statistics (AISTATS)}, volume 130 of
  \emph{Proceedings of Machine Learning Research}, pages 1252--1260. PMLR,
  2021.

\bibitem[Sarkar et~al.(2026)Sarkar, Dutta, and
  Ray~Chowdhury]{sarkar2026fairness}
Dhruv Sarkar, Soumyadeep Dutta, and Sayak Ray~Chowdhury.
\newblock Price of fairness in bandits: A tight minimax characterization, 2026.
\newblock arXiv:2607.13402.

\bibitem[Uba and Yamaguchi(2026)]{uba2026sleeping}
Shinnosuke Uba and Yutaro Yamaguchi.
\newblock Regret analysis of sleeping competing bandits, 2026.
\newblock arXiv:2603.19700.

\bibitem[{Van Parys} and Golrezaei(2024)]{vanparys2024optimal}
Bart P.~G. {Van Parys} and Negin Golrezaei.
\newblock Optimal learning for structured bandits.
\newblock \emph{Management Science}, 70\penalty0 (6):\penalty0 3951--3998,
  2024.
\newblock \doi{10.1287/mnsc.2020.02108}.
\newblock arXiv:2007.07302v3.

\bibitem[Wang and Li(2024)]{wang2024optimal}
Zilong Wang and Shuai Li.
\newblock Optimal analysis for bandit learning in matching markets with serial
  dictatorship.
\newblock \emph{Theoretical Computer Science}, 1010:\penalty0 114703, 2024.

\end{thebibliography}

\section*{Checklist}

\begin{enumerate}

  \item For all models and algorithms presented, check if you include:
  \begin{enumerate}
    \item A clear description of the mathematical setting, assumptions, algorithm, and/or model. [\textbf{Yes}. \S\ref{sec:model} fixes the model and the four standing modeling conventions; Table~\ref{tab:scope} (Appendix~\ref{app:scope}) lists the instance set and hypotheses of every main result; Algorithm~\ref{alg:policy} specifies the policy with all tie-breaks.]
    \item An analysis of the properties and complexity (time, space, sample size) of any algorithm. [\textbf{Yes}. The static compact LP has $NK+1$ variables and $O(NK)$ constraints and is polynomial-time solvable for rational input (\S\ref{sec:finite-lp}); schedules use at most $NK-N+1$ matchings. Regret and exploration counts are analyzed in \S\ref{sec:algorithm} and Appendix~\ref{app:c4}, which also states the cost of $\pi_w$: $O(NK)$ operations per round and $O(NK)$ linear programs per epoch, over $O(\sqrt{\log T})$ epochs.]
    \item (Optional) Anonymized source code, with specification of all dependencies, including external libraries. [\textbf{Yes}. Anonymized code with pinned dependencies and a one-command driver is in the supplementary material.]
  \end{enumerate}

  \Needspace{5\baselineskip}
  \item For any theoretical claim, check if you include:
  \begin{enumerate}
    \item Statements of the full set of assumptions of all theoretical results. [\textbf{Yes}. Standing assumptions in \S\ref{sec:model}; per-result hypotheses in each statement and in Table~\ref{tab:scope}.]
    \item Complete proofs of all theoretical results. [\textbf{Yes}. Appendices~\ref{app:twobytwo}--\ref{app:c4}.]
    \item Clear explanations of any assumptions. [\textbf{Yes}. \S\ref{sec:model} (parameter-space paragraph, Remark~\ref{rem:local}) and Assumption~\ref{ass:unique} with its discussion.]
  \end{enumerate}

  \item For all figures and tables that present empirical results, check if you include:
  \begin{enumerate}
    \item The code, data, and instructions needed to reproduce the main experimental results (either in the supplemental material or as a URL). [\textbf{Yes}. Supplementary material; all data is synthetic and generated by the seeded scripts.]
    \item All the training details (e.g., data splits, hyperparameters, how they were chosen). [\textbf{Yes}. The theoretical and diagnostic policy parameters are stated in \S\ref{sec:experiments} and Appendix~\ref{app:experiments}; there is no training.]
    \item A clear definition of the specific measure or statistics and error bars (e.g., with respect to the random seed after running experiments multiple times). [\textbf{Yes}. Pseudo-regret is defined in \S\ref{sec:experiments}; Figures~\ref{fig:tracking}--\ref{fig:weighted} and Table~\ref{tab:ablation} report medians and interquartile ranges; Figure~\ref{fig:frontier} shows individual runs, oracle medians and the layered KL-UCB mean. Table~\ref{tab:random} reports means and standard deviations. The theorem-style diagnostic reports mean and standard deviation for pseudo-regret and medians for round counts; seed counts are stated with each experiment.]
    \item A description of the computing infrastructure used. (e.g., type of GPUs, internal cluster, or cloud provider). [\textbf{Yes}. A single laptop CPU (no GPU); the full driver, including the parallel same-horizon KL-UCB runs, takes about two hours and ten minutes; the supplementary README records the Python, NumPy/SciPy/Matplotlib, BLAS, OS and CPU/RAM details, and describes the observed sensitivity to SciPy and numerical-library builds.]
  \end{enumerate}

  \item If you are using existing assets (e.g., code, data, models) or curating/releasing new assets, check if you include:
  \begin{enumerate}
    \item Citations of the creator If your work uses existing assets. [\textbf{Not Applicable}. No third-party assets beyond standard open-source scientific libraries, which are listed with versions in the supplementary material.]
    \item The license information of the assets, if applicable. [\textbf{Yes}. The supplementary code is released under the MIT license (anonymized \texttt{LICENSE} file in the package); third-party dependencies retain their respective licenses, documented in the supplementary README.]
    \item New assets either in the supplemental material or as a URL, if applicable. [\textbf{Yes}. The simulation code is provided in the supplementary material.]
    \item Information about consent from data providers/curators. [\textbf{Not Applicable}. All data is synthetic.]
    \item Discussion of sensible content if applicable, e.g., personally identifiable information or offensive content. [\textbf{Not Applicable}.]
  \end{enumerate}

  \item If you used crowdsourcing or conducted research with human subjects, check if you include:
  \begin{enumerate}
    \item The full text of instructions given to participants and screenshots. [\textbf{Not Applicable}.]
    \item Descriptions of potential participant risks, with links to Institutional Review Board (IRB) approvals if applicable. [\textbf{Not Applicable}.]
    \item The estimated hourly wage paid to participants and the total amount spent on participant compensation. [\textbf{Not Applicable}.]
  \end{enumerate}

\end{enumerate}

\clearpage
\appendix
\thispagestyle{empty}
\onecolumn
\raggedbottom
\aistatstitle{Exact Regret Frontiers and Externality Scheduling in
Centralized Serial-Dictatorship Bandits:\\
Supplementary Materials}

\section{Scope of the main results}\label{app:scope}

Table~\ref{tab:scope} collects, for each main result, its instance
set, additional hypotheses, and conclusion type.

\begin{table}[ht]
\centering\small
\begin{tabular}{>{\raggedright\arraybackslash}p{0.30\textwidth} l
>{\raggedright\arraybackslash}p{0.26\textwidth}
>{\raggedright\arraybackslash}p{0.14\textwidth}}
\toprule
Result & Instances & Additional hypotheses & Type \\
\midrule
Thm~\ref{thm:lb22} ($2\times2$ constants) &
$\theta_{\Delta,\gamma}$ & uniform goodness on the one-row class
$\Cclass$ & necessary + attained \\
Thm~\ref{thm:infoalloc} (information allocation) &
$\Thstrict$ & uniform goodness on $\Thstrict$ & necessary \\
Prop~\ref{prop:region}, Thm~\ref{thm:weighted} (regret region,
$C_w$) & $\Thsep$ & uniform goodness on $\Thstrict$ & necessary \\
Thm~\ref{thm:reduction} (exact quota reduction) &
$\Thstrict$ & none & structural \\
Thm~\ref{thm:compact}, Prop~\ref{prop:sparse} (compact LP,
schedules) & $\Thsep$ & $w>0$ & structural \\
Prop~\ref{prop:frontier33} ($3\times3$ frontier) &
$\mu^{(3)}$ & none & necessary (explicit) \\
Thm~\ref{thm:achieve} (canonical achievability) &
$\Thsep$ & canonical marginal uniqueness (Assumption~\ref{ass:unique});
policy uniformly good on $\Thstrict$ & attainable \\
Prop~\ref{prop:frontier-attain}, Cor~\ref{cor:frontier-closed}
($3\times3$ frontier attained) & $\mu^{(3)}$ & split parameter $s$
instance-independent & attainable \\
Thm~\ref{thm:frontier-achieve}, Cor~\ref{cor:frontier-general}
(face selection) & $\Thsep$ & target rule $\Psi$ locally
Lipschitz, data-only; limit instance-dependent & attainable \\
Prop~\ref{prop:proper}, Thm~\ref{thm:closure}, Cor~\ref{cor:rigid}
(attainable set $=G(\theta)\Xset(\theta)$; boundary = attainable
frontier) & $\Thsep$ & target $\tau$ arbitrary & necessary + attained \\
Prop~\ref{prop:family} ($N$-player family) & $\mu^{(N)}$ &
parts (i)--(ii) for every $N\ge3$; part (iii) for $N\in\{3,4,5,6\}$ (exact certificates) & structural \\
\bottomrule
\end{tabular}
\caption{Scope of the main results. The four standing modeling conventions
  (Gaussian rewards with known common variance normalized to one,
  known common priority order, complete
matching, centralized semi-bandit feedback) apply throughout and are
not repeated.}
\label{tab:scope}
\end{table}

\section{Extended related work}\label{app:related}

Table~\ref{tab:related-map} fixes the comparison axes used below.  In
particular, the word ``frontier'' is not used for a trade-off
between instances or benchmark gaps.

\begin{table}[ht]
\centering\small
\setlength{\tabcolsep}{4pt}
\caption{Comparison by object and conclusion.  The present paper studies
player-level utility-regret coefficients at a fixed top-choice-separated
instance.}
\label{tab:related-map}
\begin{tabular}{p{0.24\textwidth}p{0.29\textwidth}p{0.37\textwidth}}
\toprule
Line of work & Object & Difference from this paper \\
\midrule
Matching-bandit order-level work & Stable regret up to order or constant factors & Does not characterize the full player-level coefficient set or who pays under alternative complete-matching schedules. \\
Binary stable-regret work & A scalar count of unstable rounds & Studies a scalar binary criterion and may allow two-sided uncertainty; our target is the vector of utility regrets under one-sided Gaussian learning. \\
Group-fair exploration & Attainable incremental-utility profiles and a bargaining point & Uses grouped arrivals and Nash bargaining; our matching constraints yield an explicit polyhedral coefficient set and recoverable schedules. \\
Pareto regret work & Trade-offs across benchmark experts or arms, or across gap regimes & The trade-off is across players at a fixed instance, with exact quotas and attainable schedules. \\
\bottomrule
\end{tabular}
\end{table}

\paragraph{Bandits in matching markets.}
Bandit learning in two-sided matching markets goes back to
\citet{das2005two}.
\citet{liu2020competing} introduced the competing-bandits model and
proved logarithmic stable-regret bounds for centralized
explore-then-commit (ETC) and UCB variants; their Example~2 identifies the dependence on another
player's reward gap without giving an exact asymptotic constant.
\citet{sankararaman2021dominate} initiated the serial-dictatorship
line in a decentralized setting, and \citet{wang2024optimal} matched
the serial-dictatorship lower-bound rate.
\citet{uba2026sleeping} extend the model to time-varying
availability, with an instance-dependent
$\Omega(N(K-N+1)\log T_i/\Delta^2)$ lower bound, where $T_i$ counts
player $i$'s available rounds; \citet{lin2026contextual} study
contextual utilities.
\citet{kong2023playeroptimal} obtain player-optimal stable-regret
guarantees for every player in general decentralized markets,
matching earlier lower bounds under special preference conditions.
\citet{kong2024improved} improve the player-optimal stable-regret
bound to $O(N^2\log T/\Delta^2+K\log T/\Delta)$, removing the
dependence on the number of arms $K$ from the leading term.
We study exact instance-dependent constants for the regret vector
and their dependence on the exploration schedule.
\citet{li2025survey} survey the model variants (many-to-one markets,
non-stationary markets, two-sided unknown preferences, objectives
beyond regret). Indifferent preferences
\citep{kong2025indifference} are excluded here by the row-strict
class $\Thstrict$.

\citet{basu2025competing} studies two-sided preference uncertainty
using super-stability and proves a centralized instance-dependent
lower bound for binary stable regret. The arXiv version
(2506.15926v1, Appendix F.1) derives inverse-KL pair requirements
by restricting the confusing alternatives to single-pair
perturbations, then bounds the resulting matching-allocation
program through row and column loads. Thus pairwise information
requirements already appear in that lower-bound relaxation.
Under our one-sided fixed-priority model, we prove equivalence
with the full information constraints and characterize the exact
attainable set of player utility-regret coefficient vectors at
top-choice-separated instances.
\citet{cen2022regret} study regret, stability and fairness when
the platform can set costs and transfers (not necessarily monetary).
\citet{hosseini2024bandit} optimize utilitarian and Rawlsian
welfare across the two sides among stable matchings. Our platform
sets neither costs nor transfers; its frontier describes feasible
divisions of the exploration burden.

\paragraph{Externalities of exploration and regret trade-offs.}
\citet{raghavan2018externalities} exhibit negative group
externalities of exploration in linear contextual bandits.
\citet{baek2024fair} show how regret-optimal exploration can burden
some groups and formulate Nash bargaining over attainable
incremental utilities, measured by regret reductions relative to
learning separately. Here complete-matching constraints yield an
explicit polyhedral set of player-regret coefficients, with
recoverable schedules
(\S\ref{sec:finite-lp}--\S\ref{sec:algorithm}).
Other Pareto regret frontiers compare worst-case guarantees
against different benchmark experts or arms for a single learner
\citep{koolen2013pareto,lattimore2015pareto}.
\citet{liu2022pareto} characterize trade-offs across gap regimes
in cooperative multi-player bandits without communication.
Our frontier compares players at a fixed instance.
\citet{sarkar2026fairness} quantify the minimax price of fairness
across rounds using generalized $p$-means, addressing the burden
on early participants.

\paragraph{Stable-matching identification and regret.}
\citet{hosseini2024putting} initiate the study of the sample
complexity of finding a stable matching when one side of the market
learns its preferences.
\citet{pagare2025optimal} study fixed-confidence identification
with a single player--arm sample per round; their protocol does
not require complete matchings.
\citet{athanasopoulos2025probably} study identification under
one-sided uncertainty. \citet{athanasopoulos2026stable} extend
this objective to two-sided uncertainty with centralized
semi-bandit feedback, and also derive regret bounds that avoid
dependence on the minimum reward gap. These sample-complexity
and regret guarantees address related learning objectives; our
target is the exact attainable player utility-regret coefficient set.
In multiple-answer pure exploration, \citet{degenne2019multiple}
show that oracle allocations can form a nonconvex set and propose
Sticky Track-and-Stop to select a consistent answer. Our
cost-optimal face is convex, and regularization selects a
specified projection onto it; attaining a prescribed boundary
point may require instance calibration (\S\ref{sec:c4-frontier}).

\paragraph{Structured bandit lower bounds.}
Our general program instantiates the asymptotic machinery of
\citet{lai1985asymptotically}, \citet{graves1997asymptotically}
and \citet{combes2017minimal}, with semi-bandit additivity
\citep{combes2015combinatorial} and the fundamental inequality of
\citet{garivier2019explore}.
\citet[Table~1 and Theorem~3]{cuvelier2021asymptotically} study
polynomial-time computation of the combinatorial Graves--Lai
program. Their framework uses budgeted linear maximization; for
matchings, its available oracle is approximate.
\citet{reveillard2025multimodal} give an efficient solution for
multimodal reward structures. Under serial dictatorship, our
information constraints reduce exactly to pairwise quotas before
optimization. The resulting LP has $NK+1$ variables and is
solvable in time polynomial in the bit size of rational data
(\S\ref{sec:finite-lp}).

\paragraph{Selecting exploration allocations.}
\citet{vanparys2024optimal} construct continuous selections for
nonunique near-optimal exploration allocations in structured
regret minimization. The deep update of their DUal Structure-based
Algorithm (DUSA) minimizes the squared
norm of exploration rates and dual variables subject to a
near-optimal cost constraint; Equation (20), Proposition 4 and
Appendix H.3 of arXiv:2007.07302v3 give the rule and its continuity
argument. Our use of quadratic selection builds on this general
idea. Here a capped program tracks the projection of a fixed,
locally Lipschitz target onto the exact cost-optimal face, yielding
player-regret coefficients on the matching frontier. Pointwise
attainment of a prescribed boundary point may require an
instance-calibrated target.

\section{The two-by-two market: tools and deferred proofs}
\label{app:twobytwo}

\subsection{Setting and exact constants}\label{sec:twobytwo}

The smallest nontrivial market already contains the two mechanisms
the theory is built on: exploration for the high-priority player is
an \emph{information quota} of $\tfrac{2}{\Delta^2}\log T$ rounds,
each paid for by the low-priority player at unit price $\gamma$.
Let $N=K=2$ with priority $p_1\succ p_2$, let
$\theta_{\Delta,\gamma}$ have means $\mu_{1,\cdot}=(\Delta,0)$,
$\mu_{2,\cdot}=(0,\gamma)$, $\Delta,\gamma>0$, so
$\mstar=(a_1,a_2)$; this is the instance of Example~2 of
\citet{liu2020competing}, with the Gaussian means of their Figure~1
($\gamma$ generalizes their value $1$).
Let $N_{1,2}(T):=\sum_{t\le T}\ind{m_t(p_1)=a_2}$ count the rounds
in which $p_1$ occupies the wrong arm, exactly the rounds in which
$p_2$ is displaced, and let $\Cclass$ be the class keeping the
$p_2$-row fixed while the $p_1$-row ranges over all \emph{strict}
pairs; here \emph{uniformly good} means uniformly good on
$\Cclass$. The key device is a \emph{simulation reduction}
(Lemma~\ref{lem:reduction}, Appendix~\ref{app:22reduction}): every
complete-matching policy $A$ induces a two-armed bandit algorithm
$A^{\#}$, uniformly good whenever $A$ is, under which
$T_{a_2}(T)$ has the law of $N_{1,2}(T)$; the crux is that
$A^{\#}$ simulates the \emph{entire} market internally, generating
$p_2$'s reward from the known fixed $p_2$-row, so no Gale--Shapley
implementation \citep{gale1962college} is invoked even when $A$'s matchings depend on
$p_2$'s past rewards.

\begin{theorem}[Exact two-by-two constants]\label{thm:lb22}
Let $A$ be uniformly good on $\Cclass$ and satisfy
Assumption~\ref{ass:complete}. Then on $\theta_{\Delta,\gamma}$,
\begin{equation}
\liminf_{T\to\infty}\tfrac{\E[N_{1,2}(T)]}{\log T}
\ge\tfrac{2}{\Delta^2},
\label{eq:lbcount}
\end{equation}
and, since $R_2(T)\ge\gamma\,\E[N_{1,2}(T)]$ while
$R_1(T)=\Delta\,\E[N_{1,2}(T)]$,
$\liminf_T R_2(T)/\log T\ge2\gamma/\Delta^2$ and
$\liminf_T R_1(T)/\log T\ge2/\Delta$. Conversely, the policy
$A^{\mathrm{ucb}}$, in which $p_1$'s arm is
selected by KL-UCB with $f(t)=\log t+3\log\log t$
\citep[Theorem~1]{cappe2013kullback} on $p_1$'s own rewards and
$p_2$ receives the remaining arm, is uniformly good on $\Cclass$
and attains all three bounds with equality.
\end{theorem}

The lower bound (Appendix~\ref{app:22lb}) uses a strictly optimal
alternative $\Normal(\Delta+\varepsilon,1)$,
$\varepsilon\downarrow0$, plus the externality inequality:
whenever $p_1$ holds $a_2$, injectivity displaces $p_2$ to a
mean-$0$ arm. The upper bound (Appendix~\ref{app:22ub}) combines
KL-UCB optimality with leftover-arm accounting;
Appendix~\ref{app:22partial} treats partial matchings for the
two-by-two lower bound, and
one-shot balanced explore-then-commit pays a factor of two even
when oracle-tuned (Remark~\ref{rem:etc}, Appendix~\ref{app:etc}).

\begin{corollary}[Degenerate frontier in the two-by-two market]
\label{cor:closed22}
The unique Pareto-minimal attainable asymptotic regret vector on
$\theta_{\Delta,\gamma}$ is $(2/\Delta,\,2\gamma/\Delta^2)$: in
this instance there is no scheduling trade-off. A non-degenerate
frontier is exhibited at $N=3$ (\S\ref{sec:frontier}).
\end{corollary}

\subsection{Tools}\label{app:22tools}

We use two standard facts; proofs can be found in
\citet{lattimore2020bandit}.

\begin{lemma}[Gaussian KL]\label{lem:gausskl}
$\KL\bigl(\Normal(\mu_0,1)\,\|\,\Normal(\mu_1,1)\bigr)
=(\mu_0-\mu_1)^2/2$.
\end{lemma}

\begin{theorem}[{Lai--Robbins;
\citealp{lai1985asymptotically}; see also
\citealp[\S16]{lattimore2020bandit}}]\label{thm:lr}
Let $A^{\#}$ be a two-armed bandit algorithm that is uniformly good
over all unit-variance Gaussian instances. Suppose arm $a$ is
suboptimal under $\theta=(\nu_1,\nu_2)$. Then for \emph{every}
alternative unit-variance Gaussian law $\nu'_a$ under which arm $a$ becomes \emph{strictly}
optimal,
\[
\liminf_{T\to\infty}\frac{\E_\theta[T_a(T)]}{\log T}
\;\ge\;\frac{1}{\KL(\nu_a\,\|\,\nu'_a)},
\]
where $T_a(T)$ is the number of pulls of arm $a$ in $T$ rounds.
\end{theorem}

Note the strictness requirement: an alternative that merely
\emph{ties} the two arms carries no lower-bound information, because
every algorithm has zero bandit regret on a tied instance. Such
instances lie outside $\Cclass$, but the simulated algorithm of
Lemma~\ref{lem:reduction} is also defined on them. This is why the
optimization over alternatives in the proof of
Theorem~\ref{thm:lb22} is a limit, not a minimum.

\subsection{Proof of Lemma~\ref{lem:reduction}}\label{app:22reduction}

\begin{lemma}[Simulation reduction]\label{lem:reduction}
For every (possibly randomized) platform policy $A$ satisfying
Assumption~\ref{ass:complete} there is a two-armed bandit algorithm
$A^{\#}$ such that for every $\theta'\in\Cclass$ and every $T$, the
law of $N_{1,2}(T)$ under $A$ in the market $\theta'$ equals the
law of $T_{a_2}(T)$ under $A^{\#}$ facing the bandit
$(\Normal(\mu'_{1,a_1},1),\Normal(\mu'_{1,a_2},1))$. If $A$ is
uniformly good on $\Cclass$, then $A^{\#}$ is uniformly good over
all unit-variance Gaussian two-armed instances.
\end{lemma}

$A^{\#}$ maintains a simulated market history $H_{t-1}$ (all past
matchings and both players' rewards) and reproduces $A$'s internal
randomness. At round $t$ it computes the matching $m_t=A(H_{t-1})$.
It then
\begin{itemize}
  \item pulls the external arm $m_t(p_1)$, receives the real reward
  $X_t$, and writes $X_t$ into $H_t$ as $p_1$'s reward; and
  \item draws $p_2$'s reward $Y_t\sim\Normal(\mu_{2,m_t(p_2)},1)$
  \emph{internally} and writes it into $H_t$.
\end{itemize}
The internal draw is possible because the $p_2$-row $(0,\gamma)$ is the
same for every $\theta'\in\Cclass$, hence known to $A^{\#}$. Under any
$\theta'\in\Cclass$ the true market feedback at round $t$ is an
independent pair with exactly these laws, so the simulated history has
the same distribution as $A$'s true interaction with the market. This
holds even though $A$'s future matchings may depend on $p_2$'s past
rewards: those rewards are generated from the correct laws inside the
simulation and constitute internal randomization of $A^{\#}$.
By construction $A^{\#}$ pulls $a_2$ at round $t$ exactly when the
simulated $A$ matches $p_1$ to $a_2$, so
$T_{a_2}(T)=N_{1,2}(T)$ pathwise under the natural coupling.

For uniform goodness: fix a bandit instance
$(\mu'_{1,a_1},\mu'_{1,a_2})$ with gap $\Delta'>0$. Since $p_1$ has top
priority, its stable partner in $\theta'$ is its best arm, and, under
Assumption~\ref{ass:complete}, $p_1$ is matched in every round, so the
bandit regret of $A^{\#}$ equals
$\Delta'\,\E[\text{wrong-arm pulls}]=R_1^{A}(T;\theta')=o(T^{\alpha})$
for every $\alpha>0$. Tied bandit instances ($\Delta'=0$) lie outside
$\Cclass$, but on them every algorithm has zero bandit regret, so
$A^{\#}$ is uniformly good over the full Gaussian family.
\hfill$\qed$

\subsection{Proof of the lower bounds in
Theorem~\ref{thm:lb22}}\label{app:22lb}

\begin{proof}
\emph{(i) Forced exploration of $p_1$.}
By Lemma~\ref{lem:reduction}, $A^{\#}$ is a uniformly good bandit
algorithm; under $\theta_{\Delta,\gamma}$ its arm laws are
$\Normal(\Delta,1)$ (optimal) and $\Normal(0,1)$ (suboptimal). Fix
$\varepsilon>0$. The alternative
$\nu'_{a_2}=\Normal(\Delta+\varepsilon,1)$ makes arm $a_2$
\emph{strictly} optimal (note that $\Normal(\Delta,1)$ would only
tie the arms and is therefore inadmissible in Theorem~\ref{thm:lr}),
so
\[
\liminf_{T\to\infty}\frac{\E[N_{1,2}(T)]}{\log T}
=\liminf_{T\to\infty}\frac{\E[T_{a_2}(T)]}{\log T}
\;\ge\;
\frac{1}{\KL\bigl(\Normal(0,1)\,\|\,\Normal(\Delta+\varepsilon,1)\bigr)}
=\frac{2}{(\Delta+\varepsilon)^2},
\]
using the Gaussian divergence identity of Lemma~\ref{lem:gausskl}.
Letting $\varepsilon\downarrow0$ yields \eqref{eq:lbcount}.

\emph{(ii) The externality inequality.}
Consider any round with $m_t(p_1)=a_2$. By injectivity $p_2$ cannot
hold $a_2$, so its mean reward in that round is at most
$\mu_{2,a_1}=0=\mu_{2,a_2}-\gamma$; every other round contributes a
nonnegative amount to $R_2$. Taking expectations,
$R_2(T)\ge\gamma\,\E[N_{1,2}(T)]$, and the $R_2$ bound follows from
\eqref{eq:lbcount}. (This step uses neither
Assumption~\ref{ass:complete} nor any property of $A$.)

\emph{(iii) Player 1.}
Each round with $m_t(p_1)=a_2$ costs $p_1$ exactly $\Delta$ and
every other round costs $0$ under Assumption~\ref{ass:complete}, so
$R_1(T)=\Delta\,\E[N_{1,2}(T)]$ and the $R_1$ bound follows.
\end{proof}

\subsection{Proof of the upper bound in
Theorem~\ref{thm:lb22}}\label{app:22ub}

Since the index of $A^{\mathrm{ucb}}$ depends only on $p_1$'s past
rewards, $p_1$'s assignment process is a genuine two-armed bandit run
of KL-UCB. By the asymptotic optimality of KL-UCB for exponential
families \citep[Theorem~1]{cappe2013kullback},
\[
\E[N_{1,2}(T)]
\;\le\;
\frac{\log T}{\KL(\Normal(0,1)\,\|\,\Normal(\Delta,1))}\,(1+o(1))
=\Bigl(\frac{2}{\Delta^2}+o(1)\Bigr)\log T,
\]
and the matching lower bound \eqref{eq:lbcount} turns the $\limsup$
into a limit. The policy outputs a complete matching with $p_2$ on the
leftover arm, so $R_1(T)=\Delta\,\E[N_{1,2}(T)]$ and
$R_2(T)=\gamma\,\E[N_{1,2}(T)]$, which gives the stated constants.

For uniform goodness, fix any $\theta'\in\Cclass$, let $\Delta'>0$ be
$p_1$'s gap and let $S(T)$ count the rounds $p_1$ spends on its
suboptimal arm; KL-UCB gives $\E_{\theta'}[S(T)]=O(\log T)$. Then
$R_1(T;\theta')=\Delta'\,\E[S(T)]$ and
$|R_2(T;\theta')|=\gamma\,\E[S(T)]$: if $p_1$'s best arm is $a_1$, each
suboptimal round displaces $p_2$ to $a_1$ at cost $\gamma$, so
$R_2=\gamma\,\E[S(T)]\ge0$; if $p_1$'s best arm is $a_2$, each
suboptimal round assigns $p_2$ its preferred arm $a_2$, so
$R_2=-\gamma\,\E[S(T)]\le0$. (In the latter case $N_{1,2}(T)$ itself is
linear in $T$, so a bound in terms of $\E[N_{1,2}]$ would not
suffice.) In
both cases $R_1(T;\theta')$ and $|R_2(T;\theta')|$ are $O(\log T)$,
hence $o(T^{\alpha})$ for every $\alpha>0$.
\hfill$\qed$

\subsection{Partial matchings and the equality case}
\label{app:22partial}

\begin{remark}[When the externality inequality is an equality]
\label{rem:equality}
For a general policy only $R_2(T)\ge\gamma\,\E[N_{1,2}(T)]$ can be
asserted. Every round with $m_t(p_1)=a_2$ contributes \emph{exactly}
$\gamma$ (by injectivity $p_2$ then holds $a_1$ or is unmatched, both
of mean $0$), so equality holds if and only if, with probability one,
$p_2$ receives $a_2$ in every round in which $p_1$ does \emph{not}
receive $a_2$. Under Assumption~\ref{ass:complete} this is automatic in
the two-by-two market. The lower bound itself does not need the
assumption: step (ii) of the proof of Theorem~\ref{thm:lb22} holds
for arbitrary policies, and the
reduction extends to policies that may leave $p_1$ unmatched by giving
$A^{\#}$ a third, uninformative action of known reward $0$. On the
two instances used in the change-of-measure argument the optimal
$p_1$-mean is positive ($\Delta$ and $\Delta+\varepsilon$,
respectively), so the unmatched action is strictly suboptimal there, and
uniform goodness controls it.
\end{remark}

\subsection{One-shot balanced ETC}\label{app:etc}

\begin{remark}[The price of one-shot balanced ETC]\label{rem:etc}
The balanced explore-then-commit platform (each of the two complete
matchings is played $h$ times, after which the platform commits to the
empirically stable matching for the remaining rounds) with oracle-tuned $h$ (the optimal $h^{\star}$
depends on the unknown $\Delta$) achieves
$R_2^{\mathrm{ETC}}(T)=\bigl(\tfrac{4\gamma}{\Delta^2}+o(1)\bigr)\log
T$, a factor of two above the optimum $\tfrac{2\gamma}{\Delta^2}$. The
factor is the price of \emph{this} schedule: its one-shot balanced test
has error exponent $\Delta^2/4$, whereas the information optimum is
$\KL(\Normal(0,1)\,\|\,\Normal(\Delta,1))=\Delta^2/2$. We do not claim
that every non-adaptive schedule incurs this loss.
\end{remark}

The balanced ETC platform
with parameter $h$ plays each of the two complete matchings
$(p_1\!\to\!a_1,\,p_2\!\to\!a_2)$ and $(p_1\!\to\!a_2,\,p_2\!\to\!a_1)$
exactly $h$ times during the first $2h$ rounds, then commits for the
remaining $T-2h$ rounds to the stable matching induced by the empirical
means $\widehat\mu_{i,j}$.

\paragraph{Exploration phase.}
$p_2$ spends $h$ rounds on each arm, so its exploration regret is
$h(\gamma-0)+h\cdot 0=h\gamma$; similarly $p_1$'s is $h\Delta$.

\paragraph{Commit phase.}
$p_2$ keeps its stable partner $a_2$ unless $p_1$'s empirical ranking
is inverted, i.e., unless
$\widehat\mu_{1,a_2}\ge\widehat\mu_{1,a_1}$; in that event $p_1$
(having priority) seizes $a_2$ and $p_2$ is displaced to $a_1$ for all
remaining rounds. Each empirical mean averages $h$ unit-variance
samples, so
$\widehat\mu_{1,a_1}-\widehat\mu_{1,a_2}\sim\Normal(\Delta,2/h)$ and
\[
\Prob\bigl(\widehat\mu_{1,a_2}\ge\widehat\mu_{1,a_1}\bigr)
=\Phi\Bigl(-\Delta\sqrt{h/2}\Bigr),
\]
where $\Phi$ and $\phi$ are the standard normal CDF and density. Hence
\begin{equation}\label{eq:etcregret}
R_2^{\mathrm{ETC}}(T)=h\gamma+(T-2h)\,\gamma\,
\Phi\Bigl(-\Delta\sqrt{h/2}\Bigr),
\qquad
R_1^{\mathrm{ETC}}(T)=h\Delta+(T-2h)\,\Delta\,
\Phi\Bigl(-\Delta\sqrt{h/2}\Bigr).
\end{equation}

\paragraph{Optimal exploration length.}
The exploration length below is tuned with oracle knowledge of
$\theta_{\Delta,\gamma}$: the optimal $h^{\star}$ depends on the
unknown gap $\Delta$. Treating $h$ as continuous and minimizing the
bracket in \eqref{eq:etcregret}, the first-order condition is
\[
1=2\Phi\Bigl(-\Delta\sqrt{h/2}\Bigr)
+(T-2h)\,\frac{\Delta}{2\sqrt{2h}}\,
\phi\Bigl(\Delta\sqrt{h/2}\Bigr).
\]
The choice $h=\lceil4\log T/\Delta^2\rceil$ gives an objective
of order $\log T$, so a continuous minimizer has $h^\star=O(\log T)$;
it also satisfies $h^\star\to\infty$, since bounded $h$ leaves a
positive error probability and a linear commit cost. Thus the
minimizer is interior for large $T$, and its first-order condition
gives $T\,\tfrac{\Delta}{2\sqrt{2h^\star}}
\phi(\Delta\sqrt{h^\star/2})=1+o(1)$.
Taking logarithms and iterating once,
\begin{equation}\label{eq:hstar}
h^{\star}
=\frac{4}{\Delta^2}\Bigl(\log T-\tfrac12\log\log T\Bigr)+O(1)
=\frac{4}{\Delta^2}\log T+O(\log\log T).
\end{equation}
At $h=h^{\star}$ the commit term is $O(1)$: by the tail estimate
$\Phi(-x)\sim\phi(x)/x$ one gets
$(T-2h^{\star})\,\Phi(-\Delta\sqrt{h^{\star}/2})\to 4/\Delta^2$. The
exploration term therefore dominates and
\begin{equation}\label{eq:etcfinal}
R_2^{\mathrm{ETC}}(T)=\gamma\,h^{\star}+O(1)
=\Bigl(\frac{4\gamma}{\Delta^2}+o(1)\Bigr)\log T .
\end{equation}
Rounding $h^\star$ to an integer changes the regret by $O(1)$ and
preserves the leading constant. The error exponent $\Delta^2/4$ announced in Remark~\ref{rem:etc} is
that of $\Phi(-\Delta\sqrt{h/2})$ in $h$.

\section{Information lower bounds and regret geometry}\label{app:c1}

\subsection{Proof of Lemma~\ref{lem:nocancel} (hierarchical
no-cancellation)}\label{app:c1-nocancel}

For $\lambda\in\Thstrict$ define the expected deviation count
$S(T;\lambda)
:=\E_\lambda\sum_{t\le T}\ind{M_t\neq\mstar(\lambda)}$; a policy is
\emph{uniformly stabilizing on $\Theta'$} if
$S(T;\lambda)=o(T^{\alpha})$ for every $\lambda\in\Theta'$,
$\alpha>0$. In a general market a lower-priority player may profit
from a higher-priority deviation, so regret and stabilization can
diverge; strict serial dictatorship excludes such cancellation at
leading order.

\begin{lemma}[Hierarchical no-cancellation]\label{lem:nocancel}
Fix $\lambda\in\Thstrict$ and a complete-matching policy. The
following are equivalent:
(i) $(R_i(T;\lambda))_+=o(T^{\alpha})$ for every $i,\alpha>0$;
(ii) $S(T;\lambda)=o(T^{\alpha})$ for every $\alpha>0$;
(iii) $|R_i(T;\lambda)|=o(T^{\alpha})$ for every $i,\alpha>0$.
In particular, uniformly good $\iff$ uniformly stabilizing, on
every $\Theta'\subseteq\Thstrict$.
\end{lemma}

Fix $\lambda\in\Thstrict$ and write $\mstar=\mstar(\lambda)$. For a
complete matching $m$ define the \emph{earliest deviation level}
\[
\mathrm{lev}(m):=\min\{j\in[N]:m(j)\neq\mstar(j)\},
\qquad \mathrm{lev}(\mstar):=\infty,
\]
and for $i\in[N]$ the expected level-$i$ deviation count
\[
B_i(T):=\E_\lambda\sum_{t=1}^{T}\ind{\mathrm{lev}(M_t)=i},
\qquad\text{so that}\qquad
S(T;\lambda)=\sum_{i=1}^{N}B_i(T).
\]
Let $\Gamma:=\max_i\max_{a,b}(\lambda_{i,a}-\lambda_{i,b})\ge0$ bound
the magnitude of any per-round utility deviation. For level $i$ let
$F_i:=[K]\setminus\{\mstar(1),\dots,\mstar(i-1)\}$ be the arms still
unclaimed by higher-priority players in their stable positions, and set
\[
\delta_i:=\lambda_{i,\mstar(i)}
-\max\bigl\{\lambda_{i,a}:a\in F_i\setminus\{\mstar(i)\}\bigr\}>0
\]
whenever $F_i\setminus\{\mstar(i)\}\neq\emptyset$; positivity holds
because, by the serial-dictatorship construction and row strictness,
$\mstar(i)$ is the unique maximizer of row $i$ over $F_i$. If
$F_i\setminus\{\mstar(i)\}=\emptyset$ (possible only when $K=N$ and
$i=N$), no round can have $\mathrm{lev}(M_t)=i$: injectivity and
$M_t(j)=\mstar(j)$ for $j<i$ force $M_t(i)\in F_i=\{\mstar(i)\}$. In
that case $B_i(T)=0$ and we set $\delta_i:=1$.

\paragraph{Key decomposition.}
Consider the contribution of round $t$ to
$R_i(T;\lambda)=\E_\lambda\sum_t
(\lambda_{i,\mstar(i)}-\lambda_{i,M_t(i)})$, classified by
$\mathrm{lev}(M_t)$:
\begin{itemize}
  \item $\mathrm{lev}(M_t)>i$: then $M_t(i)=\mstar(i)$ and the contribution is
  $0$;
  \item $\mathrm{lev}(M_t)=i$: then $M_t(j)=\mstar(j)$ for $j<i$ and, by
  injectivity, $M_t(i)\in F_i\setminus\{\mstar(i)\}$, so the
  contribution is at least $\delta_i$;
  \item $\mathrm{lev}(M_t)<i$: the contribution is at least $-\Gamma$.
\end{itemize}
Taking expectations,
\begin{equation}\label{eq:keydecomp}
R_i(T;\lambda)\;\ge\;\delta_i\,B_i(T)-\Gamma\sum_{j<i}B_j(T),
\qquad i\in[N].
\end{equation}

\paragraph{(iii) $\Rightarrow$ (i).} Trivial.

\paragraph{(i) $\Rightarrow$ (ii).}
Fix $\alpha>0$ and induct on $i$. For $i=1$, \eqref{eq:keydecomp}
gives $B_1(T)\le(R_1(T))_+/\delta_1=o(T^{\alpha})$. If
$B_j(T)=o(T^{\alpha})$ for all $j<i$, then \eqref{eq:keydecomp} gives
\[
B_i(T)\le\frac{(R_i(T))_++\Gamma\sum_{j<i}B_j(T)}{\delta_i}
=o(T^{\alpha}).
\]
Summing over $i$ yields $S(T;\lambda)=o(T^{\alpha})$.

\paragraph{(ii) $\Rightarrow$ (iii).}
The contribution of round $t$ to $R_i$ is nonzero only if
$M_t(i)\neq\mstar(i)$, which implies $M_t\neq\mstar$, and its magnitude
is at most $\Gamma$. Hence
$|R_i(T;\lambda)|\le\Gamma\,S(T;\lambda)=o(T^{\alpha})$.
\hfill$\qed$

\subsection{Change-of-measure toolbox}\label{app:c1-com}

\begin{lemma}[{Data-processing inequality;
\citealp[Lemma~1]{garivier2019explore}}]\label{lem:fundamental}
For a fixed non-anticipating policy, horizon $T$, and instances
$\theta,\lambda$, every random variable $Z\in[0,1]$ measurable with
respect to the interaction $(M_1,X_1,\dots,M_T,X_T)$ satisfies
\[
\KL\bigl(\Prob^T_\theta\,\big\|\,\Prob^T_\lambda\bigr)
\;\ge\;\operatorname{kl}\bigl(\E_\theta Z,\;\E_\lambda Z\bigr),
\]
where $\operatorname{kl}(x,y)$ is the binary relative entropy.
\end{lemma}

\begin{lemma}[Divergence decomposition]\label{lem:divdecomp}
For a fixed non-anticipating policy, let $\Prob^T_\theta$ denote
the law of the interaction under $\theta$. For all
$\theta,\lambda\in\Thstrict$,
$\KL(\Prob^T_\theta\,\|\,\Prob^T_\lambda)
=\sum_{m\in\Mset}\E_\theta[N_m(T)]\,D_m(\theta,\lambda)$.
\end{lemma}

\begin{proof}[Proof of Lemma~\ref{lem:divdecomp}]
Write $I_T=(M_1,X_1,\dots,M_T,X_T)$ and let
$\pi_t(\cdot\mid I_{t-1})$ denote the action kernel at round $t$,
which by the protocol of \S\ref{sec:model} is the \emph{same} under
both instances. With respect to a common dominating measure, the
density of $\Prob^T_\theta$ factors as
$\prod_{t\le T}\pi_t(M_t\mid I_{t-1})\,p^{M_t}_\theta(X_t)$, where
$p^{m}_\theta$ is the density of $P^m_\theta$, and likewise under
$\lambda$. The policy factors are identical and cancel in the
log-likelihood ratio, leaving
\[
\log\frac{d\Prob^T_\theta}{d\Prob^T_\lambda}(I_T)
=\sum_{t=1}^{T}
\log\frac{p^{M_t}_\theta(X_t)}{p^{M_t}_\lambda(X_t)} .
\]
Taking $\E_\theta$ and conditioning the $t$-th term on
$(I_{t-1},M_t)$ replaces that term by its conditional expectation
$D_{M_t}(\theta,\lambda)$, whence
\[
\KL\bigl(\Prob^T_\theta\,\big\|\,\Prob^T_\lambda\bigr)
=\E_\theta\sum_{t=1}^{T}D_{M_t}(\theta,\lambda)
=\sum_{m\in\Mset}\E_\theta[N_m(T)]\,D_m(\theta,\lambda).
\]
The additive form of $D_m$ over matched pairs in \eqref{eq:Dm} reflects
the product structure of $P^m_\theta$. Compare
\citet[Lemma~15.1]{lattimore2020bandit} for the single-agent case and
\citet{combes2015combinatorial} for the semi-bandit form.
\end{proof}

We will use the elementary bound
\begin{equation}\label{eq:klbound}
\operatorname{kl}(x,y)\;\ge\;x\log\tfrac1y-\log2,
\qquad x\in[0,1],\ y\in(0,1).
\end{equation}

\subsection{Proof of Theorem~\ref{thm:infoalloc}}
\label{app:c1-infoalloc}

Fix $\theta\in\Thstrict$ and $\lambda\in\Lset(\theta)$. By
Lemma~\ref{lem:nocancel}, the policy is uniformly stabilizing, so for
every $\alpha>0$,
\[
S(T;\theta)=o(T^{\alpha})
\qquad\text{and}\qquad
S(T;\lambda)=o(T^{\alpha}).
\]
Define $Z_T:=N_{\mstar(\theta)}(T)/T\in[0,1]$. Under $\theta$,
\[
\E_\theta Z_T=1-\frac{S(T;\theta)}{T}\;\longrightarrow\;1 .
\]
Under $\lambda$, every round executing $\mstar(\theta)$ is a deviation
from $\mstar(\lambda)$, because $\mstar(\lambda)\neq\mstar(\theta)$ as
injective maps; hence
\[
\E_\lambda Z_T\le\frac{S(T;\lambda)}{T}=o(T^{\alpha-1})
\qquad\text{for every }\alpha>0 .
\]
Since $D_{\mstar(\theta)}(\theta,\lambda)=0$, Lemmas
\ref{lem:divdecomp} and \ref{lem:fundamental}, which together form
the fundamental inequality of \citet{garivier2019explore}, give
\[
\sum_{m\neq\mstar(\theta)}\E_\theta[N_m(T)]\,D_m(\theta,\lambda)
=\KL\bigl(\Prob^T_\theta\,\big\|\,\Prob^T_\lambda\bigr)
\ge\operatorname{kl}\bigl(\E_\theta Z_T,\E_\lambda Z_T\bigr).
\]
The left-hand side is finite (finitely many Gaussian
Kullback--Leibler terms), so $\E_\lambda Z_T>0$ whenever
$\E_\theta Z_T>0$ and \eqref{eq:klbound} applies. Fix
$\alpha\in(0,1)$. For all $T$ large enough,
$\E_\lambda Z_T\le T^{\alpha-1}$ and $\E_\theta Z_T\ge\tfrac12$, so
\[
\sum_{m\neq\mstar}\E_\theta[N_m(T)]\,D_m(\theta,\lambda)
\;\ge\;\E_\theta Z_T\,(1-\alpha)\log T-\log2 .
\]
Dividing by $\log T$ and using $\E_\theta Z_T\to1$,
\[
\liminf_{T\to\infty}
\frac{\sum_{m\neq\mstar}\E_\theta[N_m(T)]\,D_m(\theta,\lambda)}
{\log T}\;\ge\;1-\alpha .
\]
Letting $\alpha\downarrow0$ completes the proof. \hfill$\qed$

\subsection{The upper-closed Graves--Lai region: deferred proofs}
\label{app:c1-region}

\begin{proposition}[Cluster points of matching frequencies]
\label{prop:cluster}
Under the assumptions of Theorem~\ref{thm:infoalloc}, if along a
subsequence $T_k\to\infty$ we have
$\E_\theta[N_m(T_k)]/\log T_k\to x_m\in[0,\infty)$ for every
$m\neq\mstar(\theta)$, then $x\in\Xset(\theta)$.
\end{proposition}

\begin{proof}[Proof of Proposition~\ref{prop:cluster}]
$\Mset$ is finite, so for each $\lambda\in\Lset(\theta)$,
\[
\sum_{m\neq\mstar}x_mD_m(\theta,\lambda)
=\lim_k\frac{1}{\log T_k}\sum_{m\neq\mstar}\E_\theta[N_m(T_k)]\,
D_m(\theta,\lambda)\;\ge\;1
\]
by Theorem~\ref{thm:infoalloc}.
\end{proof}

\begin{lemma}[Positivity under separation]\label{lem:positivity}
For $\theta\in\Thsep$, every entry of $G(\theta)$ is nonnegative,
and every column has a strictly positive entry.
\end{lemma}

\begin{proof}
For $\theta\in\Thsep$, $\mstar(i)=\arg\max_a\mu_{i,a}$, so
$g_i(m;\theta)\ge0$ with equality iff $m(i)=\mstar(i)$. If
$m\neq\mstar$ then $m(i)\neq\mstar(i)$ for some $i$, and row
strictness makes the corresponding gap strictly positive.
\end{proof}

\begin{lemma}[$\UGL(\theta)$ is a closed convex upper set]
\label{lem:closed}
For $\theta\in\Thsep$, $\UGL(\theta)$ is convex, closed, and
satisfies $\UGL(\theta)+\R_+^N=\UGL(\theta)$.
\end{lemma}

\begin{proof}[Proof of Lemma~\ref{lem:closed}]
$\Xset(\theta)$ is the intersection of the nonnegative orthant with
closed halfspaces, hence closed and convex; convexity and the
upper-set property of $\UGL(\theta)$ follow because $G(\theta)$ is
linear and $\R_+^N$ is a convex cone. For closedness, let
$u_n=G(\theta)x_n+s_n\to u$ with $x_n\in\Xset(\theta)$ and
$s_n\ge0$. With
$c:=\min_{m\neq\mstar}\sum_ig_i(m;\theta)>0$
(Lemma~\ref{lem:positivity}),
$\langle\mathbf1,u_n\rangle\ge\langle\mathbf1,G(\theta)x_n\rangle
\ge c\,\|x_n\|_1$, so $(x_n)$ is bounded; along a subsequence
$x_n\to x$, and $x\in\Xset(\theta)$ by closedness. Then
$s_n=u_n-G(\theta)x_n\to u-G(\theta)x\ge0$, so
$u=G(\theta)x+\bigl(u-G(\theta)x\bigr)\in\UGL(\theta)$.
\end{proof}

\begin{lemma}[Polyhedrality]\label{lem:polyhedral}
For $\theta\in\Thsep$, $\Xset(\theta)$ and $\UGL(\theta)$ are
polyhedra, and $\UGL(\theta)\subseteq\R^N_+$.
\end{lemma}

\begin{proof}
By Theorem~\ref{thm:reduction} (whose proof uses only
$\theta\in\Thstrict$ and is independent of this section),
$\Xset(\theta)=\{x\ge0:z_{i,a}(x)\ge q_{i,a}\ \forall(i,a)\in\Eset\}$
is cut out by finitely many linear inequalities. The image of a
polyhedron under the linear map $G(\theta)$ is a polyhedron, and
the Minkowski sum of two polyhedra is a polyhedron
(Minkowski--Weyl), so $\UGL(\theta)=G(\theta)\Xset(\theta)+\R^N_+$
is a polyhedron. Nonnegativity follows from $x\ge0$ and
Lemma~\ref{lem:positivity}.
\end{proof}

\begin{proof}[Proof of Proposition~\ref{prop:proper}]
\emph{(i)} Fix $u\in\UGL(\theta)$ and let
$L:=\UGL(\theta)\cap(u-\R^N_+)$, a nonempty closed set
(Lemma~\ref{lem:closed}) contained in the box $[0,u]$ by
Lemma~\ref{lem:polyhedral}, hence compact. Let $u'$ minimize the
continuous function $\langle\mathbf 1,\cdot\rangle$ over $L$. If
some $u''\in\UGL(\theta)$ satisfied $u''\le u'$, $u''\neq u'$, then
$u''\in L$ and $\langle\mathbf 1,u''\rangle<\langle\mathbf 1,u'\rangle$,
a contradiction; so $u'$ is Pareto-minimal and $u'\le u$.

\emph{(ii), ``if''.} If $\langle w,u\rangle=C_w(\theta)
=\min_{\UGL(\theta)}\langle w,\cdot\rangle$ with $w>0$ and
$u''\le u$, $u''\neq u$, $u''\in\UGL(\theta)$, then
$\langle w,u''\rangle<\langle w,u\rangle$, impossible.

\emph{(ii), ``only if''.} Let $u$ be Pareto-minimal and write
$u=G(\theta)x+v$ with $x\in\Xset(\theta)$, $v\ge0$; minimality
forces $v=0$, since $G(\theta)x\in\UGL(\theta)$ lies below $u$.
Then $x$ is an efficient solution of the multiobjective linear
program $\min\{G(\theta)x:x\in\Xset(\theta)\}$: an $x'\in\Xset(\theta)$
with $G(\theta)x'\le G(\theta)x$, $\neq$, would produce a point of
$\UGL(\theta)$ componentwise below $u$. We give a direct argument, a
variant of the linear-programming duality proof of
\citet[Theorem~1]{isermann1974proper} that every efficient solution of
a linear vector optimization problem is optimal for some strictly
positive weight vector (by \citealp[Theorem~1]{geoffrion1968proper},
such optima are properly efficient, which gives
\citealp[Theorem~2]{isermann1974proper}). Fix any finite
$H$-representation $\UGL(\theta)=\{v:Av\ge b\}$, which exists by
Lemma~\ref{lem:polyhedral}, and consider the linear program
\[
\max\{\langle\mathbf 1,s\rangle: As\le Au-b,\ s\ge0\},
\]
whose feasible $s$ are exactly the vectors with $u-s\in\UGL(\theta)$.
It is feasible ($s=0$) and bounded ($\UGL(\theta)\subseteq\R^N_+$
gives $s\le u$), and its value is $0$ because $u$ is Pareto-minimal.
By strong duality its dual $\min\{\langle Au-b,y\rangle:A^\top y\ge
\mathbf 1,\ y\ge0\}$ has an optimal $y$ with $\langle Au-b,y\rangle=0$.
Put $w:=A^\top y\ge\mathbf 1$, so $w\in\R^N_{++}$. For every
$v\in\UGL(\theta)$, $\langle w,v\rangle=\langle y,Av\rangle\ge\langle
y,b\rangle=\langle y,Au\rangle=\langle w,u\rangle$; hence
$\langle w,u\rangle=\min_{\UGL(\theta)}\langle w,\cdot\rangle=C_w(\theta)$
by \eqref{eq:Cw}. (Equivalently: the normal cone of $\UGL(\theta)$
at a Pareto-minimal point contains $-w$, with all coordinates at
most $-1$.) The final assertion combines (ii) with
the fact that each $\arg\min_{\UGL(\theta)}\langle w,\cdot\rangle$
is a face of the polyhedron $\UGL(\theta)$.
\end{proof}

\begin{proof}[Proof of Proposition~\ref{prop:region}]
Set $c:=\min_{m\neq\mstar}\sum_i g_i(m;\theta)>0$
(Lemma~\ref{lem:positivity}, finiteness of $\Mset$). By
\eqref{eq:regretcount},
$\sum_iR_i(T_k)\ge c\sum_{m\neq\mstar}\E_\theta[N_m(T_k)]$, so the
normalized counts are bounded along the subsequence; extract a
further subsequence along which they converge to some $x\ge0$. By
Proposition~\ref{prop:cluster}, $x\in\Xset(\theta)$, and
\eqref{eq:regretcount} passes to the limit coordinatewise (finitely
many matchings), giving $r=G(\theta)x$.
\end{proof}

\begin{remark}[Coordinatewise liminfs]\label{rem:liminf}
Proposition~\ref{prop:region} is phrased through cluster points of
the whole vector $R(T;\theta)/\log T$ because the coordinatewise
liminf $\underline r:=\bigl(\liminf_{T}R_i(T;\theta)/\log T\bigr)_{i}$
can lie outside $\UGL(\theta)$. Take $\theta=\mu^{(3)}$ and let
$\tau^{\mathrm{par}}$ and $\tau^{\mathrm{pro}}$ be the pair marginals
of the endpoint allocations $\{m_{\mathrm B}\!:200,\ m_{\mathrm E}\!:2\}$
and $\{m_{\mathrm A}\!:200,\ m_{\mathrm C}\!:200,\ m_{\mathrm D}\!:2\}$
of Proposition~\ref{prop:frontier33}, whose regret vectors are
$(22,20,202)$ and $(22,222,2.02)$. Let $k^{(1)}:=k_0$ and
$k^{(j+1)}:=(k^{(j)})^2$, and run the uncosted policy $\pi_\tau$ of
\S\ref{sec:c4-attainable} with the epoch-dependent target
$\tau^{\mathrm{par}}$ in the epochs $k\in(k^{(j)},k^{(j+1)}]$ with $j$
odd and $\tau^{\mathrm{pro}}$ in those with $j$ even.

\emph{Uniform goodness.} The proof of Lemma~\ref{lem:ug} uses only the
quota constraints and the mass cap of the solved program, which do not
involve the target, so this complete-matching policy is uniformly good
on $\Thstrict$.

\emph{Counts.} Both targets are off-stable blocks of points of
$\mathcal P(\theta)$ (proof of Theorem~\ref{thm:closure}(ii)), so
Lemma~\ref{lem:trackplugin}, applied to each of the two targets,
shows that on $A_k\cap B_k$ and for
$k\ge\max\{k_2(\theta,\tau^{\mathrm{par}}),k_2(\theta,\tau^{\mathrm{pro}})\}$
the selected off-stable marginals are within $Ck^{-1/3}$ of the target
of epoch $k$, with $C$ the larger of the two constants. The repair, certification and accounting arguments of
Theorem~\ref{thm:closure}(i) apply unchanged, because they use only
the two target-free properties above. Only the Ces\`aro step of
Lemma~\ref{lem:track} changes. Since $b_k\sim k^2$, the exploration
of the epochs $k\le k^{(j)}$ contributes
$O(b_{k^{(j)}})=o(b_{k^{(j+1)}})$ rounds to each pair in
expectation, while the block $(k^{(j)},k^{(j+1)}]$ contributes
$b_{k^{(j+1)}}$ times its target plus $o(b_{k^{(j+1)}})$. Hence, by
Step~1 of Appendix~\ref{app:c2-frontier},
$R(T_{k^{(j+1)}};\theta)/\log T_{k^{(j+1)}}$ tends to $(22,20,202)$
along odd $j$ and to $(22,222,2.02)$ along even $j$.

\emph{Conclusion.} Consequently $\underline r\le(22,20,2.02)$
coordinatewise. Every $u=G(\theta)x+v\in\UGL(\theta)$ satisfies
$u\ge G(\theta)x\ge r(t)$ for some segment point $r(t)$ by Step~2 of
Appendix~\ref{app:c2-frontier}. If $\underline r$ belonged to
$\UGL(\theta)$, then $r(t)\le(22,20,2.02)$ would require
$220-t\le20$ and $4+0.99t\le2.02$, that is, $t\ge200$ and
$t\le-2$. Hence $\underline r\notin\UGL(\theta)$, although every
cluster point of $R(T;\theta)/\log T$ lies in
$G(\theta)\Xset(\theta)$ by Proposition~\ref{prop:region}.
\end{remark}

\begin{proof}[Proof of Theorem~\ref{thm:weighted}]
By \eqref{eq:regretcount} and Lemma~\ref{lem:positivity},
$w^\top R(T)\ge0$; let $L$ be the liminf, and assume $L<\infty$
(else there is nothing to prove). Pick a subsequence realizing $L$.
Since $c_w:=\min_{m\neq\mstar}\langle w,g(m;\theta)\rangle>0$, the
normalized counts are bounded along it; extract a further
subsequence converging to some $x$, which lies in $\Xset(\theta)$
by Proposition~\ref{prop:cluster}. Then
$L=\sum_{m\neq\mstar}x_m\langle w,g(m;\theta)\rangle\ge C_w(\theta)$.
\end{proof}

\begin{remark}[Nonemptiness]\label{rem:feasible}
For $\theta\in\Thsep$, every $\lambda\in\Lset(\theta)$ singles out
a pair of $\Eset$ through its earliest deviation, with
$D_m(\theta,\lambda)\ge\Delta_{i,a}^2/2$ for every $m$ with
$m(i)=a$; a uniform allocation at level
$\max_{\Eset}2/\Delta_{i,a}^2$ is therefore feasible, so
$\Xset(\theta)\neq\emptyset$ and $C_w(\theta)<\infty$.
\end{remark}

\paragraph{Details.}
Every $\lambda\in\Lset(\theta)$ singles out a pair in $\Eset$: let
$i$ be the first player with
$\mstar_\lambda(i)\neq\mstar_\theta(i)$. Players $1,\dots,i-1$ keep
their assignments, and since $\lambda$ agrees with $\theta$ on the
pairs of $\mstar(\theta)$, the arm $a:=\mstar_\lambda(i)$ satisfies
$(i,a)\in\Eset$ and $\lambda_{i,a}>\mu_{i,\mstar(i)}>\mu_{i,a}$;
hence $|\lambda_{i,a}-\mu_{i,a}|>\Delta_{i,a}$ and, by
\eqref{eq:Dm}, \emph{every} matching $m$ with $m(i)=a$ has
$D_m(\theta,\lambda)\ge\Delta_{i,a}^2/2$. If
$\Lset(\theta)=\emptyset$ (in particular when $\Eset=\emptyset$,
which happens only in the trivial market $K=1$), the constraints in
\eqref{eq:Xset} are vacuous and $x=0$ is feasible. Otherwise set
$L:=\max_{(i,a)\in\Eset}2/\Delta_{i,a}^2$ and $x_m:=L$ for every
$m\neq\mstar$. Since $K\ge N$, every pair $(i,a)\in\Eset$ extends
to a complete injective matching $m$ with $m(i)=a$, and
$m\neq\mstar$ because $a\neq\mstar(i)$; for the pair singled out by
any given $\lambda$, that matching alone yields
$x_mD_m(\theta,\lambda)\ge(2/\Delta_{i,a}^2)(\Delta_{i,a}^2/2)=1$.
Hence $x\in\Xset(\theta)$, so $\Xset(\theta)\neq\emptyset$ and
$C_w(\theta)<\infty$.

\subsection{Recovery of the two-by-two constants}
\label{app:c1-recovery}

\begin{corollary}[Two-by-two recovery]\label{cor:recovery22}
Let $N=K=2$ and let $\theta_{\Delta,\gamma}$ be the instance of
Appendix~\ref{sec:twobytwo}. Every complete-matching policy that is uniformly
good on $\Thstrict$ satisfies
\[
\liminf_{T\to\infty}\frac{R_1(T)}{\log T}\ge\frac{2}{\Delta},
\qquad
\liminf_{T\to\infty}\frac{R_2(T)}{\log T}\ge\frac{2\gamma}{\Delta^2}.
\]
\end{corollary}

\begin{proof}
Besides $\mstar=(a_1,a_2)$ there is a single matching
$m^e=(a_2,a_1)$, with gap column $g(m^e)=(\Delta,\gamma)$. An
alternative $\lambda\in\Lset(\theta_{\Delta,\gamma})$ must satisfy
$\lambda_{1,a_1}=\Delta$, $\lambda_{2,a_2}=\gamma$, and
$\mstar(\lambda)\neq\mstar$ forces $\lambda_{1,a_2}>\Delta$; the entry
$\lambda_{2,a_1}\neq\gamma$ is otherwise free. Hence
$D_{m^e}(\theta,\lambda)
=\tfrac12\lambda_{1,a_2}^2+\tfrac12\lambda_{2,a_1}^2
$ ranges over $(\Delta^2/2,\infty)$, with infimum $\Delta^2/2$
approached as $\lambda_{1,a_2}\downarrow\Delta$ and
$\lambda_{2,a_1}\to0$. The constraints
$x_{m^e}D_{m^e}(\theta,\lambda)\ge1$ for all $\lambda$ therefore
collapse to
$\Xset(\theta_{\Delta,\gamma})=\{x_{m^e}\ge2/\Delta^2\}$.
Now let $L:=\liminf R_2(T)/\log T$ and assume $L<\infty$. Along a
realizing subsequence,
$\E[N_{m^e}(T_k)]/\log T_k=R_2(T_k)/(\gamma\log T_k)\to L/\gamma$ by
\eqref{eq:regretcount}, so $L/\gamma\ge2/\Delta^2$ by
Proposition~\ref{prop:cluster}. The bound for $R_1$ is identical with
$\gamma$ replaced by $\Delta$.
\end{proof}

The quantification differs slightly from Appendix~\ref{sec:twobytwo}: here
uniform goodness is required on all of $\Thstrict$, whereas the
dedicated argument there needs it only on the one-row class $\Cclass$
and is in that respect sharper.

\section{Pairwise quotas, marginal schedules, and frontier geometry}\label{app:c2}

\subsection{Proofs of Lemma~\ref{lem:singleentry},
Theorem~\ref{thm:reduction} and Theorem~\ref{thm:compact}}
\label{app:c2-proofs}

\begin{lemma}[Single-entry alternatives]\label{lem:singleentry}
Let $\theta\in\Thstrict$ and $(i,a)\in\Eset$. For every
sufficiently small $\varepsilon>0$, the
instance $\lambda^{(i,a,\varepsilon)}$ that agrees with $\theta$
except for
$\lambda^{(i,a,\varepsilon)}_{i,a}=\mu_{i,\mstar(i)}+\varepsilon$
belongs to $\Lset(\theta)$; its serial-dictatorship matching first
deviates at player $i$, with
$\mstar_{\lambda^{(i,a,\varepsilon)}}(i)=a$; and for every
$m\in\Mset$,
\[
D_m\bigl(\theta,\lambda^{(i,a,\varepsilon)}\bigr)
=\ind{m(i)=a}\,\tfrac{1}{2}(\Delta_{i,a}+\varepsilon)^2 .
\]
\end{lemma}

\begin{proof}[Proof of Lemma~\ref{lem:singleentry}]
\emph{Row strictness.} Row $i$ of $\lambda^{(i,a,\varepsilon)}$ has
entries
$\{\mu_{i,b}\}_{b\neq a}\cup\{\mu_{i,\mstar(i)}+\varepsilon\}$. The
new entry differs from $\mu_{i,\mstar(i)}$ because $\varepsilon>0$,
and from every other $\mu_{i,b}$ once
$\varepsilon<\min_{b\neq\mstar(i)}|\mu_{i,b}-\mu_{i,\mstar(i)}|$, a
positive bound by row strictness; the other rows are unchanged. Hence
$\lambda^{(i,a,\varepsilon)}\in\Thstrict$ for all sufficiently small
$\varepsilon>0$.

\emph{Diagonal agreement.} The modified entry is $(i,a)$ with
$a\neq\mstar(i)$, so it lies off the $\mstar$-diagonal and
$D_{\mstar(\theta)}(\theta,\lambda^{(i,a,\varepsilon)})=0$. (Note
that $a$ may coincide with $\mstar(j)$ for some \emph{lower}-priority
player $j>i$; the modified entry is in row $i$, not row $j$, so the
diagonal is still untouched.)

\emph{Serial dictatorship under the alternative.} Players $j<i$
have unchanged rows and unchanged free sets, hence unchanged
assignments. At level $i$ the free set is still $F_i$; it contains
$a$ (as $(i,a)\in\Eset$), and
$\lambda_{i,a}=\mu_{i,\mstar(i)}+\varepsilon$ strictly exceeds
$\lambda_{i,b}=\mu_{i,b}\le\mu_{i,\mstar(i)}$ for every
$b\in F_i\setminus\{a\}$. Hence
$\mstar_{\lambda}(i)=a\neq\mstar(i)$, so
$\mstar(\lambda)\neq\mstar(\theta)$ and
$\lambda^{(i,a,\varepsilon)}\in\Lset(\theta)$.

\emph{Information.} The two instances differ in the single entry
$(i,a)$, and
$\lambda_{i,a}-\mu_{i,a}
=\mu_{i,\mstar(i)}+\varepsilon-\mu_{i,a}=\Delta_{i,a}+\varepsilon$;
the formula for $D_m$ follows from \eqref{eq:Dm}.
\end{proof}

\begin{proof}[Proof of Theorem~\ref{thm:reduction}]
\emph{($\subseteq$, the quotas are necessary.)}
Fix $(i,a)\in\Eset$ and $x\in\Xset(\theta)$. Applying the defining
constraint of $\Xset(\theta)$ at
$\lambda=\lambda^{(i,a,\varepsilon)}$ and using
Lemma~\ref{lem:singleentry},
\[
1\;\le\;\sum_{m\neq\mstar}x_m\,
D_m\bigl(\theta,\lambda^{(i,a,\varepsilon)}\bigr)
=z_{i,a}(x)\,\frac{(\Delta_{i,a}+\varepsilon)^2}{2},
\]
so $z_{i,a}(x)\ge2/(\Delta_{i,a}+\varepsilon)^2$ for every
sufficiently small $\varepsilon>0$; letting $\varepsilon\downarrow0$
gives the closed constraint $z_{i,a}(x)\ge2/\Delta_{i,a}^2$.

\emph{($\supseteq$, the quotas are sufficient.)}
Let $x$ satisfy all quotas and let $\lambda\in\Lset(\theta)$ be
arbitrary. Let $i$ be the first player with
$\mstar_\lambda(i)\neq\mstar_\theta(i)$. Players $1,\dots,i-1$
keep their assignments, and $\lambda$ agrees with $\theta$ on every
stable pair. Thus $a:=\mstar_\lambda(i)$ satisfies $(i,a)\in\Eset$ and
$\lambda_{i,a}>\mu_{i,\mstar(i)}>\mu_{i,a}$. This argument uses only
row strictness and the serial-dictatorship construction. Consequently,
$\lambda_{i,a}-\mu_{i,a}>\Delta_{i,a}$, so every matching $m$ with
$m(i)=a$ has $D_m(\theta,\lambda)\ge\Delta_{i,a}^2/2$. Hence
\[
\sum_{m\neq\mstar}x_m\,D_m(\theta,\lambda)
\;\ge\;\frac{\Delta_{i,a}^2}{2}\,z_{i,a}(x)
\;\ge\;\frac{\Delta_{i,a}^2}{2}\cdot\frac{2}{\Delta_{i,a}^2}
=1 . \qedhere
\]
\end{proof}

\begin{remark}[A finite-dimensional alternative class suffices]
\label{rem:finiteclass}
The necessity half of the proof uses only the single-entry
alternatives $\lambda^{(i,a,\varepsilon)}$. Consequently the
conclusions of Proposition~\ref{prop:region} and
Theorem~\ref{thm:weighted} remain valid for policies that are
uniformly good merely on the union of $|\Eset|$ one-parameter families
$\{\theta\}\cup\{\lambda^{(i,a,\varepsilon)}:(i,a)\in\Eset,\
\varepsilon>0\text{ small}\}$: Proposition~\ref{prop:cluster} then
places every cluster point in the set cut out by these
alternatives, which is $\Xset(\theta)$ by the argument above. This
mirrors the two-by-two case, where Theorem~\ref{thm:lb22} needs
uniform goodness only on the one-row class $\Cclass$.
\end{remark}

\begin{corollary}[Matching-cover LP]\label{cor:quota}
For $\theta\in\Thsep$ and $w\in\R^N_{++}$, $C_w(\theta)$ equals the
value of \eqref{eq:coverlp}.
\end{corollary}

\begin{proof}
Substitute the description of $\Xset(\theta)$ from
Theorem~\ref{thm:reduction} into the definition \eqref{eq:Cw}.
\end{proof}

\begin{proof}[Proof of Theorem~\ref{thm:compact}]
\emph{($\le$.)} Given $x$ feasible for \eqref{eq:coverlp}, set
$z_{i,a}:=z_{i,a}(x)$ and $\rho:=\sum_{m\neq\mstar}x_m$. The quota
constraints hold by assumption; the flow constraints hold because
each matching assigns every player exactly one arm and each arm to
at most one player; and
\[
\sum_{i,a}w_i\Delta_{i,a}z_{i,a}
=\sum_{m\neq\mstar}x_m\sum_iw_i\Delta_{i,m(i)}
=\sum_{m\neq\mstar}x_m\langle w,g(m;\theta)\rangle .
\]

\emph{($\ge$.)} Let $(z,\rho)$ be feasible for
\eqref{eq:compactlp}. If $\rho=0$ then $z=0$, which is feasible
only when $\Eset=\emptyset$; in that case both programs have value
$0$. If $\rho>0$, then $\zeta:=z/\rho\in P$, so by
Lemma~\ref{lem:integrality} and Minkowski's theorem
$\zeta=\sum_k\beta_k\zeta^{m_k}$ for a convex combination of
matchings; set $x'_{m_k}:=\rho\beta_k$ and $x'_m:=0$ for every other
matching $m$, so that $\sum_mx'_m=\rho$ and
the marginals of $x'$ equal $z$. Discard the possible
$\mstar$-coordinate and let $x:=(x'_m)_{m\neq\mstar}$. For
$(i,a)\in\Eset$ we have $a\neq\mstar(i)$, so
$z_{i,a}(x)=z_{i,a}-x'_{\mstar}\ind{a=\mstar(i)}=z_{i,a}\ge q_{i,a}$:
the quotas survive. The objective also survives, because the
$\mstar$-term contributes $\sum_iw_i\Delta_{i,\mstar(i)}=0$. Hence
$x$ is feasible for \eqref{eq:coverlp} with the same value.
\end{proof}

\subsection{Dual programs}\label{app:c2-dual}

The dual of \eqref{eq:coverlp} assigns a price $y_{i,a}\ge0$ to
each quota:
\begin{equation}\label{eq:duallp}
\max_{y\ge0}\;\sum_{(i,a)\in\Eset}q_{i,a}\,y_{i,a}
\quad\text{s.t.}\quad
\sum_{i} y_{i,m(i)}\;\le\;\bigl\langle w,\,g(m;\theta)\bigr\rangle
\quad\forall m\neq\mstar,
\end{equation}
with the convention $y_{i,a}=0$ for $(i,a)\notin\Eset$. Checking
dual feasibility is the weighted assignment problem
\begin{equation}\label{eq:separation}
\max_{m\in\Mset}\;
\sum_{i}\bigl(y_{i,m(i)}-w_i\Delta_{i,m(i)}\bigr)\;\le\;0 ,
\end{equation}
which is what Hungarian pricing in column generation solves. By
complementary slackness, an optimal schedule uses a matching only
if its weighted utility loss is exactly paid for by the exploration
quotas it serves, and a quota has positive price only if it is met
with equality.

\subsection{Proof of Lemma~\ref{lem:integrality}}
\label{app:c2-integrality}

\begin{lemma}[Rectangular assignment polytope]\label{lem:integrality}
The polytope
$P:=\{\zeta\in\R_+^{N\times K}:
\sum_a\zeta_{i,a}=1\ \forall i,\
\sum_i\zeta_{i,a}\le1\ \forall a\}$
has vertex set exactly the incidence matrices $\zeta^m$,
$\zeta^m_{i,a}=\ind{m(i)=a}$, of injective matchings $m\in\Mset$.
\end{lemma}

The constraint matrix of $P$ is the vertex--edge incidence matrix of
the complete bipartite graph on players and arms: the variable
$\zeta_{i,a}$ has coefficient $1$ in the row constraint of player $i$
and in the column constraint of arm $a$, and $0$ elsewhere. Incidence
matrices of bipartite graphs are totally unimodular; appending slack
variables for the inequality constraints preserves total
unimodularity, and the right-hand side is integral, so every vertex of
$P$ is integral. An integral point of $P$ is a $0/1$ matrix with unit
row sums and column sums at most $1$, which is precisely the incidence
matrix $\zeta^m$ of an injective matching $m\in\Mset$. Conversely,
every $\zeta^m$ is a vertex: it is an extreme point of the cube
$[0,1]^{N\times K}$ and lies in $P\subseteq[0,1]^{N\times K}$.
\hfill$\qed$

\subsection{Proof of Proposition~\ref{prop:sparse}}
\label{app:c2-sparse}

\begin{remark}[A harmless zero-cost ray]\label{rem:ray}
Adding $c\,\zeta^{\mstar}$ to $z$ and $c$ to $\rho$ preserves
feasibility, quotas and cost, so \eqref{eq:compactlp} is never
uniquely solved in $(z,\rho)$. Only the off-stable marginals carry
information or cost; when a specific optimizer is needed we select
canonically by minimizing $\rho$, and uniqueness assumptions refer
to this canonical off-stable solution.
\end{remark}

Let $(z,\rho)$ be feasible with $\rho>0$ and set $\zeta:=z/\rho\in P$.
The $N$ row equalities are linearly independent, so the affine hull of
$P$ has dimension at most $NK-N$. By Lemma~\ref{lem:integrality} and
Minkowski's theorem, $\zeta$ is a convex combination of vertices
$\zeta^m$, and by Carath\'eodory's theorem in the affine hull, at most
$NK-N+1$ of them suffice: $\zeta=\sum_k\beta_k\zeta^{m_k}$. Setting
$c_k:=\rho\beta_k$ gives the schedule. If some $m_k=\mstar$, dropping
that term leaves the quotas intact, because $(i,a)\in\Eset$ has
$a\neq\mstar(i)$ so $\zeta^{\mstar}$ contributes nothing to any quota
marginal, and leaves the cost intact, because
$\sum_iw_i\Delta_{i,\mstar(i)}=0$.

\paragraph{Constructive version.}
The column deficits $1-\sum_i\zeta_{i,a}\ge0$ sum to $K-N$, so any
transportation solution (northwest-corner filling suffices)
distributes them into $K-N$ dummy rows of unit row sum, extending
$\zeta$ to a $K\times K$ doubly stochastic matrix $\widetilde\zeta$.
Standard Birkhoff--von Neumann peeling decomposes $\widetilde\zeta$
into at most $(K-1)^2+1$ permutation matrices; deleting the dummy rows
of each permutation leaves injective matchings $[N]\to[K]$ whose
combination is $\zeta$. If the Carath\'eodory bound $NK-N+1$ is
desired, iteratively remove affinely dependent terms: as long as more
than $NK-N+1$ matchings carry positive weight, their incidence
vectors are affinely dependent, and moving along the dependency until
the first coefficient hits zero shrinks the support by at least one.
\hfill$\qed$

\subsection{Proof of Proposition~\ref{prop:frontier33}}
\label{app:c2-frontier}

The five non-stable matchings are listed with their gap vectors and
the quotas they serve:
\[
\begin{array}{c|c|c|c}
\text{matching} & \text{assignment} & g(m)=(g_1,g_2,g_3) &
\text{quotas served}\\ \hline
m_{\mathrm A} & (a_2,a_1,a_3) & (0.1,\;1,\;0) & (1,a_2)\\
m_{\mathrm B} & (a_2,a_3,a_1) & (0.1,\;0.1,\;1) &
 (1,a_2)\ \text{and}\ (2,a_3)\\
m_{\mathrm C} & (a_1,a_3,a_2) & (0,\;0.1,\;0.01) & (2,a_3)\\
m_{\mathrm D} & (a_3,a_1,a_2) & (1,\;1,\;0.01) & (1,a_3)\\
m_{\mathrm E} & (a_3,a_2,a_1) & (1,\;0,\;1) & (1,a_3)
\end{array}
\]

Write $x=(x_{\mathrm A},x_{\mathrm B},x_{\mathrm C},x_{\mathrm D},
x_{\mathrm E})$ for an allocation over the five non-stable matchings
of Example~\ref{ex:33}. By Theorem~\ref{thm:reduction},
\[
x\in\Xset
\iff
x_{\mathrm A}+x_{\mathrm B}\ge200,\quad
x_{\mathrm B}+x_{\mathrm C}\ge200,\quad
x_{\mathrm D}+x_{\mathrm E}\ge2 ,
\]
since $m_{\mathrm A},m_{\mathrm B}$ are the matchings containing
$(p_1,a_2)$; $m_{\mathrm B},m_{\mathrm C}$ those containing
$(p_2,a_3)$; and $m_{\mathrm D},m_{\mathrm E}$ those containing
$(p_1,a_3)$.

\paragraph{Step 1: the family and its image.}
For $h\in[0,200]$, $y\in[0,2]$, the allocation
$x(h,y)=(200-h,\,h,\,200-h,\,y,\,2-y)$ meets the three constraints
with equality. Its regret vector is
\begin{align*}
r_1&=0.1(x_{\mathrm A}+x_{\mathrm B})
   +1\cdot(x_{\mathrm D}+x_{\mathrm E})=0.1\cdot200+2=22,\\
r_2&=1\cdot x_{\mathrm A}+0.1\,x_{\mathrm B}+0.1\,x_{\mathrm C}
   +1\cdot x_{\mathrm D}
   =(200-h)+0.1h+0.1(200-h)+y=220-(h-y),\\
r_3&=1\cdot x_{\mathrm B}+0.01\,x_{\mathrm C}+0.01\,x_{\mathrm D}
   +1\cdot x_{\mathrm E}
   =h+0.01(200-h)+0.01y+(2-y)=4+0.99(h-y).
\end{align*}
With $t:=h-y\in[-2,200]$ this is exactly $r(t)$, and every
$t\in[-2,200]$ is realized (e.g., $y=0$ for $t\ge0$ and $h=0$ for
$t<0$).

\paragraph{Step 2: monotone projection.}
Let $x\in\Xset$ be arbitrary and define
\[
x''_{\mathrm B}:=\min(x_{\mathrm B},200),\quad
x''_{\mathrm A}=x''_{\mathrm C}:=200-x''_{\mathrm B},\quad
x''_{\mathrm D}:=\min(x_{\mathrm D},2),\quad
x''_{\mathrm E}:=2-x''_{\mathrm D}.
\]
Then $x''=x(h,y)$ with $h=x''_{\mathrm B}$, $y=x''_{\mathrm D}$, and
$x''\le x$ componentwise: if $x_{\mathrm B}\le200$ this reads
$200-x_{\mathrm B}\le x_{\mathrm A}$ and
$200-x_{\mathrm B}\le x_{\mathrm C}$, which are the first two quota
constraints; if $x_{\mathrm B}>200$ then
$x''_{\mathrm B}=200<x_{\mathrm B}$ and
$x''_{\mathrm A}=x''_{\mathrm C}=0$; and analogously
$x''_{\mathrm D}\le x_{\mathrm D}$,
$x''_{\mathrm E}\le x_{\mathrm E}$ from the third constraint. Since
every gap vector is componentwise nonnegative
(Lemma~\ref{lem:positivity}), $G x''\le G x$: every feasible regret
vector lies componentwise above a point of the segment.

\paragraph{Step 3: conclusion.}
Along the segment, $r_2(t)=220-t$ is strictly decreasing and
$r_3(t)=4+0.99t$ strictly increasing, so no two segment points are
componentwise ordered. Now take any Pareto-minimal $u\in\UGL$. Writing
$u=Gx+v$ with $v\ge0$, minimality forces $v=0$ (otherwise
$Gx\in\UGL$ satisfies $Gx<u$ in some coordinate and $Gx\le u$),
and Step 2 yields a segment point $r(t)\le u$; minimality forces
$u=r(t)$. Conversely, fix $t$ and suppose $u'\in\UGL$ satisfies
$u'\le r(t)$, $u'\neq r(t)$. As before we may take $u'=Gx'$, and Step
2 gives $t'$ with $r(t')\le Gx'\le r(t)$. Comparing the second
coordinates gives $t'\ge t$ and the third gives $t'\le t$, so $t'=t$
and $r(t)\le Gx'\le u'\le r(t)$ forces $u'=r(t)$, a contradiction.
Hence the Pareto-minimal boundary is exactly
$\{r(t):t\in[-2,200]\}$.
\hfill$\qed$

\subsection{Proof of Proposition~\ref{prop:family}}
\label{app:c2-family}

For $N\ge3$ and $K=N$, define $\mu^{(N)}$ by $\mu_{i,a_i}=1$,
$\mu_{i,a_{i+1}}=0.9$ for $i<N$, and $\mu_{i,a_1}=0$ for $i\ge2$.
For $2\le j<i$, let
$\mu_{i,a_j}=0.99-0.08\,(j-2)/(N-2)$; for $j\ge3$, let
$\mu_{1,a_j}=0.8\,(j-3)/(N-2)$, and for $i\ge2$ and $j\ge i+2$, let
$\mu_{i,a_j}=0.8\,(j-i-1)/(N-2)$. Thus $\mu^{(3)}$ is
Example~\ref{ex:33}.

\begin{proposition}[The $N$-player family]\label{prop:family}
(i)~$\mu^{(N)}\in\Thsep$ with $|\Eset|=N(N-1)/2$. (ii)~Every
Pareto-minimal point of $\UGL(\mu^{(N)})$ has
$r_1=c_1^{(N)}:=\sum_{a\neq a_1}2/\Delta_{1,a}$, so the top
player's cost is schedule-independent and every face of the
necessary lower boundary has dimension at most $N-2$. (iii)~For
$N\in\{3,4,5,6\}$ the boundary contains a face of dimension exactly
$N-2$, and for $N\in\{4,5,6\}$ the affine hull of the
Pareto-minimal boundary is the whole hyperplane $\{r_1=c_1^{(N)}\}$.
\end{proposition}

\emph{(i)} Within row $i$ the entries are $1$, $0.9$ (if $i<N$),
$0$ (if $i\ge2$), the mid-arm values $0.99-0.08(j-2)/(N-2)$ for
$2\le j<i$, and the far-arm values $0.8k/(N-2)$ with
$1\le k\le N-i-1$ (row $1$: $0\le k\le N-3$). The mid-arm values lie
in $(0.91,0.99]$, hence strictly between $0.9$ and $1$, and the
far-arm values are at most $0.8(N-3)/(N-2)<0.8<0.9$; so all row
entries are pairwise distinct for every $N\ge3$, $1$ is the unique
row maximum, and $a_i$ is the top choice of player $i$ alone. For
$i\ge2$ the far-arm values are positive, so $\mu_{i,a_1}=0$ is the
unique minimum of row $i$. Player $i$'s free non-stable arms are
$a_{i+1},\dots,a_N$, giving $\sum_i(N-i)=N(N-1)/2$ quotas.

\emph{(ii)} We prove a slightly more general statement: let
$\theta\in\Thsep$ with $K\ge N$ and suppose that $a_1:=\mstar(1)$
satisfies $\mu_{j,a_1}\le\mu_{j,a}$ for every player $j\ge2$ and
every arm $a$. Then every Pareto-minimal $u\in\UGL(\theta)$ has
$u_1=\sum_{a\neq a_1}2/\Delta_{1,a}$. Indeed, write $u=G(\theta)x$
with $x\in\Xset(\theta)$ (minimality forces the $\R^N_+$-part to
vanish). Since all arms are free for player $1$,
$u_1=\sum_{a\neq a_1}\Delta_{1,a}z_{1,a}(x)$ with
$z_{1,a}(x)\ge q_{1,a}=2/\Delta_{1,a}^2$, so $u_1\ge c_1$ with
equality iff every player-$1$ quota is tight. Suppose
$z_{1,a}(x)>q_{1,a}$ for some $a$ and pick $m$ with $x_m>0$,
$m(1)=a$. If $a_1$ is unmatched under $m$ (possible only when
$K>N$), let $m'$ agree with $m$ except $m'(1)=a_1$: moving a small
mass from $m$ to $m'$ (or deleting it if $m'=\mstar$) keeps
$z_{1,a}\ge q_{1,a}$, changes no off-stable marginal other than
$z_{1,a}$ (every other change is on stable pairs, such as
$(1,a_1)$), lowers $u_1$ and leaves every other coordinate unchanged, so the new
allocation is feasible, componentwise no larger than $u$ and strictly
smaller in the first coordinate, a contradiction.
Otherwise some player $j\ge2$ holds $a_1$ under $m$. Let $m'$ agree
with $m$ except $m'(1)=a_1$, $m'(j)=a$, and
move mass $\varepsilon\in(0,x_m]$ from $m$ to $m'$ (if $m'=\mstar$,
simply delete it). Quotas: $z_{1,a}$ decreases by $\varepsilon$ and
stays $\ge q_{1,a}$ for small $\varepsilon$; $(j,a_1)\notin\Eset$
since $a_1$ is taken at level $1$; every other marginal is
unchanged or increases. Regret: player $1$ loses
$\varepsilon\Delta_{1,a}>0$, player $j$ changes by
$\varepsilon(\Delta_{j,a}-\Delta_{j,a_1})\le0$ by the hypothesis on
$a_1$, and nobody else changes. The new allocation is feasible, and
its regret vector is componentwise no larger than $u$ and strictly
smaller in the first coordinate, contradicting
minimality. For $\mu^{(N)}$ the hypothesis holds because
$\mu_{j,a_1}=0$ is the unique row minimum for every $j\ge2$.

For the dimension bound, let $\mathcal F$ be a convex subset of the
Pareto-minimal boundary, for instance a face of $\UGL$ contained in
it, and pick $u$ in its relative interior. By
Proposition~\ref{prop:proper}(ii), $u$ lies in the face
$\arg\min_{\UGL}\langle w,\cdot\rangle$ of $\UGL$ for some $w>0$, and
a face of a convex set that contains a relative interior point of a
convex subset contains the whole subset; hence $\mathcal F$ lies in
that face, and so in the two hyperplanes
$\{\langle w,\cdot\rangle=C_w\}$ and $\{r_1=c_1^{(N)}\}$, whose
normals $w$ and $e_1$ are linearly independent for $N\ge2$. Thus
$\dim\mathcal F\le N-2$.

\emph{(iii)} A point $u=G x$ with $x$ optimal for some $w>0$ is
Pareto-minimal (Proposition~\ref{prop:proper}(ii)), and optimality
of a primal feasible point is certified by any dual feasible point
of equal objective value (weak duality). The supplementary script
\texttt{frontier\_family.py} therefore proves (iii) by exhibiting,
for each $N\in\{3,\dots,6\}$: (a)~a rational weight $w^\star>0$, a
rational dual-feasible solution of \eqref{eq:compactlp} at
$w^\star$ with value $V$, and $N-1$ rational primal-feasible
marginal vectors of cost $V$ whose regret vectors are affinely
independent, so that the optimal face at $w^\star$ has dimension
$N-2$; (b)~for $N\ge4$, $N$ rational weights with an exact
primal--dual pair each, whose $N$ regret vectors are affinely
independent. The weights in (a) are $w^\star=(1,0.99,1)$,
$(1,0.99,1,1)$, $(1,0.99,\tfrac{289}{300},1,1)$ and
$(1,0.99,0.97,1,1,1)$.
All certificates are stored in
\texttt{results/frontier\_family\_certificates.json} and
re-verified from the file by the function \texttt{verify}, which
uses rational arithmetic only, no floating point and no LP solver.
\hfill$\qed$

\paragraph{The $N=4$ face certificate.}
For concreteness we print the certificate for $N=4$ (the others have the
same format). Order the variables of \eqref{eq:compactlp} as
$(z_{1,a_1},\dots,z_{1,a_4},z_{2,a_1},\dots,z_{4,a_4},\rho)$ and its
constraints as the $N$ player rows $\sum_a z_{i,a}-\rho=0$, the $N$
arm rows $\sum_i z_{i,a}-\rho\le0$ and the $|\Eset|=6$ quota rows
$-z_{i,a}\le-q_{i,a}$ for $(i,a)\in\Eset$ in the order
$(1,a_2),\ (1,a_3),\ (1,a_4),\ (2,a_3),\ (2,a_4),\ (3,a_4)$.
At $w^\star=(1,\tfrac{99}{100},1,1)$ the cost vector is
$c_{i,a}=w^\star_i(\mu_{i,a_i}-\mu_{i,a})$. The dual certificate is
\[
y_{\mathrm{eq}}=\bigl(0,\ \tfrac{99}{100},\ 1,\ 1\bigr),\qquad
y_{\mathrm{arm}}=\bigl(0,\ -\tfrac{99}{100},\ -1,\ -1\bigr),\qquad
y_{\mathrm{quota}}=\bigl(-\tfrac{109}{100},\ -2,\ -\tfrac{8}{5},\ -\tfrac{109}{1000},\ -\tfrac{151}{250},\ -\tfrac{1}{10}\bigr),
\]
with $y_{\mathrm{arm}},y_{\mathrm{quota}}\le0$, $A^\top y\le c$
coordinatewise and objective $b^\top y=\tfrac{12422}{45}$. (Here
$\mu^{(4)}$ has rows $(1,0.9,0,0.4)$, $(0,1,0.9,0.4)$,
$(0,0.99,1,0.9)$, $(0,0.99,0.95,1)$.) The three
optimal primal points (rows $i=1,\dots,4$, columns $a_1,\dots,a_4$, then $\rho$) follow;
write $r(z):=(\sum_a\Delta_{i,a}z_{i,a})_{i=1}^4$ for their regret vectors.
\begin{align*}
z^{(1)}&=\begin{pmatrix}\tfrac{32}{9} & 200 & 2 & \tfrac{50}{9}\\ 0 & \tfrac{50}{9} & 200 & \tfrac{50}{9}\\ 0 & 2 & \tfrac{82}{9} & 200\\ \tfrac{1868}{9} & \tfrac{32}{9} & 0 & 0\end{pmatrix},\ \rho=\tfrac{1900}{9},\ r(z^{(1)})=\bigl(\tfrac{76}{3},\ \tfrac{70}{3},\ \tfrac{1001}{50},\ \tfrac{46708}{225}\bigr),\\
z^{(2)}&=\begin{pmatrix}\tfrac{1850}{9} & 200 & 2 & \tfrac{50}{9}\\ \tfrac{1868}{9} & 0 & 200 & \tfrac{50}{9}\\ 0 & 2 & \tfrac{1900}{9} & 200\\ 0 & \tfrac{1900}{9} & 0 & 202\end{pmatrix},\ \rho=\tfrac{3718}{9},\ r(z^{(2)})=\bigl(\tfrac{76}{3},\ \tfrac{2078}{9},\ \tfrac{1001}{50},\ \tfrac{19}{9}\bigr),\\
z^{(3)}&=\begin{pmatrix}\tfrac{1832}{9} & 200 & 2 & \tfrac{50}{9}\\ \tfrac{1850}{9} & 0 & 200 & \tfrac{50}{9}\\ 2 & 0 & \tfrac{1882}{9} & 200\\ 0 & \tfrac{1900}{9} & 0 & 200\end{pmatrix},\ \rho=\tfrac{3700}{9},\ r(z^{(3)})=\bigl(\tfrac{76}{3},\ \tfrac{2060}{9},\ 22,\ \tfrac{19}{9}\bigr).
\end{align*}
Each $z^{(k)}$ is primal feasible with $c^\top z^{(k)}=b^\top y$, so all
three are optimal at $w^\star>0$ by weak duality and hence Pareto-minimal
(Proposition~\ref{prop:proper}(ii)); their regret vectors are affinely
independent, so the optimal face has dimension $2=N-2$. These four
checks (primal feasibility, dual feasibility, equal objective, affine
independence) are exactly what \texttt{verify} performs, in rational
arithmetic, for every stored certificate; the $N\in\{3,5,6\}$
certificates and the affine-hull certificates of (iii) are in the same
file. Part (iii) is thus a finite, machine-checkable proof, not a
numerical experiment.

\section{Proofs for Section~\ref{sec:algorithm}}\label{app:c4}

\paragraph{Canonical policy and guarantee.}\label{app:c4-canonical}
The canonical policy $\pi_w$ solves the compact LP \eqref{eq:compactlp} with quotas
$\widehat q$ and the \emph{surrogate} costs
$w_i(\widehat\mu_{i,\widehat m_k(i)}-\widehat\mu_{i,a})_+$, which keep
the objective bounded on every instance, separated or not, and selects
the canonical solution of Remark~\ref{rem:ray} by three ordered
programs: (i) minimize the cost, with value $\widehat v_k$; (ii)
minimize $\rho$ subject to
$\text{cost}\le\widehat v_k+\varepsilon_k$, where the vanishing slack
$\varepsilon_k:=k^{-1/4}$ dominates the $O(k^{-1/3})$ estimation
error of the data (Appendix~\ref{app:c4-plugin}); (iii) among the
solutions of (ii), take the lexicographically smallest. This yields
off-stable marginals $\widehat z_k$ and full marginals
$\widehat z^{\mathrm{full}}_k$ (row sums $\widehat\rho_k$), and the
constructive decomposition of Proposition~\ref{prop:sparse} turns
$\widehat z^{\mathrm{full}}_k$ into a schedule
$\{(m_j,c_j)\}_j$ of at most $NK-N+1$ matchings, from which any
$\widehat m_k$-component is discarded.

\begin{assumption}[Canonical marginal uniqueness]\label{ass:unique}
At $(\theta,w)$, $\theta\in\Thsep$, the canonical program (minimal
cost, then minimal $\rho$) applied to \eqref{eq:compactlp} with the
true quotas has an optimal set whose off-stable projection is a
singleton $\{z^\circ_w(\theta)\}$.
\end{assumption}

Uniqueness assumptions of this kind are common in the Graves--Lai
literature \citep{combes2017minimal}, but Assumption~\ref{ass:unique}
is not automatic. Lemma~\ref{lem:generic} confines its failure to
finitely many hyperplanes of weights only under a noncoincidence
condition that must be checked for each instance; for $\mu^{(3)}$ the
only failing weight ratio is $w_2/w_3=0.99$.
Theorem~\ref{thm:frontier-achieve} does not need the assumption.

\begin{theorem}[Achievability under canonical marginal uniqueness]
\label{thm:achieve}
For every $w\in\R^N_{++}$ and $\alpha>\xi>3$, the policy $\pi_w$ is
uniformly good on $\Thstrict$. Moreover, for every
$\theta\in\Thsep$ satisfying Assumption~\ref{ass:unique},
\[
\lim_{T\to\infty}\frac{\E_\theta[N_{i,a}(T)]}{\log T}
=z^\circ_{w,i,a}(\theta)
\]
for every off-stable pair $(i,a)$, and consequently, with
$r^\circ_i(\theta):=\sum_{a}\Delta_{i,a}\,z^\circ_{w,i,a}(\theta)$,
$R_i(T;\theta)/\log T\to r^\circ_i(\theta)$ for every $i$ and
$R_w(T;\theta)/\log T\to C_w(\theta)$: the bound of
Theorem~\ref{thm:weighted} is attained exactly, and the regret vector
converges to the canonical boundary point.
\end{theorem}

\begin{remark}[Sensitivity near a supporting weight]
\label{rem:c4-scope}
The failing ratio $0.99$ is the slope of the segment of
Proposition~\ref{prop:frontier33}: a face of slope $-s$ is optimal
at weight ratio $s$. Equal weights are nearly degenerate when
the face redistributes cost nearly one-for-one,
so at finite horizons the empirical LP can select different endpoints
across seeds (\S\ref{sec:experiments}). Section~\ref{sec:c4-frontier}
allows selection on a nonunique face.
\end{remark}

The proof combines identification, plug-in LP stability, schedule
tracking, certification and regret accounting; details are in
Appendices~\ref{app:c4-ident}--\ref{app:c4-ug}.

For small $k$ the prescribed rounds could exceed the epoch length:
their total is a deterministic $O(NKk^2)$ for every variant (the
pathwise mass bound $NK\kappa_k$ of Lemmas~\ref{lem:mass}
and~\ref{lem:capfeas}) against an epoch length of order $e^{k^2}$,
so the policy starts the epoch schedule at the smallest
deterministic $k_0=k_0(N,K,\alpha,\xi)$ beyond which every epoch
accommodates its phases, playing the round-robin prefix of
Algorithm~\ref{alg:policy} until then; the prefix affects no
asymptotic statement and $k_0$ is suppressed from the notation.
One may take $k_0$ as the smallest $k\ge2$ with $T_{k-1}\ge K$
such that, for every $j\ge k$,
\[
T_j-T_{j-1}\ \ge\ NK\lceil\sqrt{b(T_j)}\rceil
+\lceil(NK+1)j(b(T_j)-b(T_{j-1}))\rceil+NK+1.
\]
The right-hand side bounds repair, schedule execution, replay and
rounding on every path and is $O(NKj^2)$, so such a $k_0$ exists.

For convergence and coefficient limits we fix $\theta\in\Thsep$,
unless the statement specifies another instance. Canonical marginal
uniqueness (Assumption~\ref{ass:unique}) is used only for the
canonical plug-in limit and its tracking/accounting consequences;
it is not assumed for the split, regularized or uncosted variants.
The uniform-goodness Lemma~\ref{lem:ug} in
Appendix~\ref{app:c4-strict} works at an arbitrary strict instance
and assumes neither separation nor canonical uniqueness.
Let $\delta_\theta>0$ be half the smallest
absolute difference between two entries in the same row of $\theta$
(positive by row strictness), let
$\Delta_{\min}:=\min_{(i,a)\in\Eset}\Delta_{i,a}>0$, and let
\[
d_\theta:=\min_{m\neq\mstar}\sum_iw_i
\bigl(\mu_{i,\mstar(i)}-\mu_{i,m(i)}\bigr)_+\;>\;0
\]
(positive because on $\Thsep$ every non-stable matching has a strictly
positive gap somewhere, Lemma~\ref{lem:positivity}); the regular
events defined below have radius $\bar\delta_\theta:=\delta_\theta/2$.
Throughout, $C$ and $c$ denote
constants depending on $(\theta,w,\alpha,\xi,N,K)$ whose value may
change between displays. As in \S\ref{sec:model} we assume $K\ge2$; when $K=1$ the market
admits a single complete matching, $\Eset=\emptyset$
(\S\ref{sec:general-lb}), $C_w(\theta)=0$, and every policy attains
it trivially. We also use the standard stack-of-rewards convention:
the successive observations of each pair $(i,a)$ form an i.i.d.\
sequence fixed in advance, and $\widehat\mu^{(n)}_{i,a}$ denotes the
mean of its first $n$ samples, so that events indexed by
\emph{counts} are events about these sequences alone, independent of
the policy's sampling times. We abbreviate $b_k:=b(T_k)$ throughout
this appendix.
Two concentration tools are used. First, the union-over-counts bound
for unit-variance Gaussian sample means: for every $\varepsilon>0$
and $n_0\ge1$,
\begin{equation}\label{eq:maximal}
\Prob\Bigl(\exists n\ge n_0:\
\bigl|\widehat\mu_{i,a}^{(n)}-\mu_{i,a}\bigr|\ge\varepsilon\Bigr)
\;\le\;\frac{2\,e^{-n_0\varepsilon^2/2}}{1-e^{-\varepsilon^2/2}} .
\end{equation}
Second, a peeling bound at the self-normalized scale
$\sqrt{2c/n}$.

\begin{lemma}[Peeling]\label{lem:peel}
There is an absolute constant $C_1$ such that for every pair, every
$t\ge3$ and every $c\ge1$,
\[
\Prob\Bigl(\exists\,n\le t:\
\bigl|\widehat\mu^{(n)}-\mu\bigr|\ge\sqrt{2c/n}\Bigr)
\;\le\;C_1\,c\,(\log t)\,e^{-c}.
\]
\end{lemma}

\begin{proof}
Write $S_n$ for the centered partial sum, so the event at count $n$
is $|S_n|\ge\sqrt{2cn}$. \emph{Small counts.} For each fixed
$n\le\lceil c\rceil$, $\Prob(|S_n|\ge\sqrt{2cn})\le2e^{-c}$; a union
over these counts costs at most $2(c+1)e^{-c}\le4c\,e^{-c}$.
\emph{Large counts.} Set $n_j:=\lceil c(1+1/c)^j\rceil$ for $j\ge0$
and partition $[\lceil c\rceil,t]$ into the blocks $[n_j,n_{j+1})$;
their number is at most $\log t/\log(1+1/c)+1\le C\,c\log t$ for
$t\ge3$, $c\ge1$. Because now $n_j\ge c(1+1/c)^j\ge c$ and
$n_{j+1}\le c(1+1/c)^{j+1}+1$,
\[
\frac{n_j}{n_{j+1}}
\;\ge\;\frac{c(1+1/c)^j}{c(1+1/c)^{j+1}+1}
\;=\;\frac{1}{(1+1/c)+1/\bigl(c(1+1/c)^j\bigr)}
\;\ge\;\frac{1}{1+2/c}
\;\ge\;1-\frac{2}{c},
\]
the ceiling now being harmless since the block starts are all
$\ge c$. For $n\in[n_j,n_{j+1})$ the event implies
$|S_n|\ge\sqrt{2cn}\ge\sqrt{2c\,n_j}$, so by the two-sided
exponential-martingale maximal inequality for Gaussian increments,
\[
\Prob\Bigl(\max_{n<n_{j+1}}|S_n|\ge\sqrt{2c\,n_j}\Bigr)
\le2\exp\Bigl(-\frac{2c\,n_j}{2n_{j+1}}\Bigr)
=2\,e^{-c\,n_j/n_{j+1}}
\le2\,e^{-c(1-2/c)}
=2e^2\cdot e^{-c}.
\]
Summing the block bounds and adding the small-count union bound
gives $C_1\,c\,(\log t)\,e^{-c}$.
\end{proof}

\begin{algofloat}[H]
\centering
\fbox{\begin{minipage}{0.94\textwidth}
In policy formulas at round $t$, counts and estimates use only
observations before that round; cumulative $N_{i,a}(T)$ includes
round $T$. Off-stable vectors are extended by zero on stable pairs
when summed over all arms.

\emph{Completion rule.} For a pair $(i,a)$, $\mathrm{comp}(i,a)$ is
the complete matching that assigns $a$ to level $i$ and then, for
$j=1,\dots,N$ in increasing order skipping $i$, assigns level $j$
the smallest-index arm still free.

\emph{Prefix} ($t\le T_{k_0-1}$): play the cyclic shift
$m_r(i):=a_{((i+r-1)\bmod K)+1}$ with $r=t\bmod K$; the $K$ shifts
cover every pair. Take $k_0\ge2$ with $T_{k_0-1}\ge K$, so every
pair is observed before the first solve. Empirical arm ties and ties
between the two least-sampled pairs are broken by increasing arm
index; order the recovered schedule lexicographically by matching.

\emph{Epoch $k\ge k_0$:}
\begin{enumerate}\itemsep2pt
  \item \emph{Repair.} While some pair has
  $N_{i,a}<\sqrt{b(T_k)}$: among such pairs take the least-sampled,
  ties broken lexicographically; play $\mathrm{comp}(i,a)$.
  \item \emph{Solve.} Three ordered programs of
  Appendix~\ref{app:c4-canonical} (min cost; min $\rho$ with slack
  $\varepsilon_k$; lexicographic), then the constructive
  decomposition of Proposition~\ref{prop:sparse}; discard
  $\widehat m_k$-components.
  \item \emph{Explore.} For each $(m_j,c_j)$ of the decomposition,
  in order, play $m_j$ exactly
  $\lceil(b(T_k)-b(T_{k-1}))c_j\rceil$ times; then play
  $\widehat m_k$ exactly $\lceil(b(T_k)-b(T_{k-1}))k\rceil$ times
  (stable-matching replay).
  \item \emph{Exploit/estimate.} At each remaining round: recompute
  $\widetilde m_t=\mstar(\widehat\mu(t))$; scan levels
  $i=1,\dots,N$ in increasing order and, within level $i$, the free
  arms $a\neq\widetilde m_t(i)$ in increasing index order; the
  first comparison failing its certificate is the \emph{uncertain
  comparison}. If none fails, play $\widetilde m_t$; otherwise play
  $\mathrm{comp}(j,b)$ for the least-sampled pair $(j,b)$ of that
  comparison.
\end{enumerate}

\emph{Variant $\pi_s$} (\S\ref{sec:c4-frontier}): when the
empirical eligible structure matches Example~\ref{ex:33} (labels
$\widehat q_{12},\widehat q_{13},\widehat q_{23}$), step 2 is
replaced by the allocation
$\widehat x_{\mathrm B}=s\min(\widehat q_{12},\widehat q_{23})$,
$\widehat x_{\mathrm A}=\widehat q_{12}-\widehat x_{\mathrm B}$,
$\widehat x_{\mathrm C}=\widehat q_{23}-\widehat x_{\mathrm B}$,
$\widehat x_{\mathrm D}=(1-s)\widehat q_{13}$,
$\widehat x_{\mathrm E}=s\widehat q_{13}$ over the five non-stable
matchings labeled by $\widehat m_k$; otherwise each eligible pair
receives its own dedicated matching $\mathrm{comp}(i,a)$ with mass
$\widehat q_{i,a}\le\kappa_k$.

\emph{Variant $\pi_{w,\Psi}$} (\S\ref{sec:c4-frontier}): step 2
evaluates the target
$\tau_k=\Psi(\widehat m_k,\widehat\Eset_k,\widehat q_k,\widehat\Delta_k)$,
solves the capped regularized program \eqref{eq:regprog}, then
takes the minimal-$\rho$ completion; the three-stage canonical
selection is skipped and everything else is unchanged. A constant
rule is the fixed-target policy $\pi_{w,\tau}$.

\emph{Variant $\pi_\tau$} (\S\ref{sec:c4-attainable}): as
$\pi_{w,\tau}$ with the cost term of \eqref{eq:regprog} removed
(program \eqref{eq:trackprog}).
\end{minipage}}
\caption{The policy $\pi_w$ with all tie-breaks explicit. The analysis
uses only the stated sampling guarantees, so any other fixed
tie-break rule, and any decomposition of the same marginals into at
most $NK-N+1$ matchings, works verbatim. With the factor $k$ in the
replay length, the replays up to epoch $k$ give every stable pair
$\Theta(k\,b(T_k))$ samples rather than $\Theta(b(T_k))$, so the
stable-side certification radius is $O(k^{-1/2})$ (proof of
Lemma~\ref{lem:certify}).}
\label{alg:policy}
\end{algofloat}

\paragraph{Computational cost.}
A round costs $O(NK)$ operations: one serial-dictatorship pass and at
most $NK$ certificate checks. In each epoch $\pi_w$ solves $O(NK)$
linear programs (cost, mass, lexicographic selection), $\pi_{w,\Psi}$
and $\pi_\tau$ solve one convex quadratic program and one linear
program, and $\pi_s$ solves none. Each program has $NK+1$ variables
and $O(NK)$ constraints, and its solution is decomposed into at most
$NK-N+1$ matchings. Since $T_k=\lceil e^{k^2}\rceil$, a horizon $T$
contains $O(\sqrt{\log T})$ epochs, so $\pi_w$ solves
$O(NK\sqrt{\log T})$ linear programs in total.

We now fix the \emph{regular event} of epoch $k$. Because the solve
step runs \emph{after} the repair phase, whose additional samples
change the empirical means, the event must control the estimates at
\emph{all} counts the epoch can exhibit, not at the deterministic
time $T_{k-1}$. With $c_1\in(0,1]$ a constant such that
$\sqrt{b(T_{k-1})}\ge c_1k$ for all $k\ge1$, set
\[
A_k:=\bigl\{\,\bigl|\widehat\mu^{(n)}_{i,a}-\mu_{i,a}\bigr|
<\bar\delta_\theta\ \ \forall(i,a),\ \forall n\ge c_1k\,\bigr\},
\]
an event about the reward stacks alone. Since after the repair phase
of epoch $k$ every pair count is at least $\sqrt{b(T_k)}\ge c_1k$
and counts never decrease, on $A_k$ \emph{every} estimate consulted
at the solve step, and at any later round of epoch $k$, is within
$\bar\delta_\theta$ of its mean. Two distinct entries of a row of
$\theta$ differ by at least $2\delta_\theta=4\bar\delta_\theta$, so on
$A_k$ every within-row difference of empirical means has the sign of
the true difference and at least half its absolute value. This gives
three properties of these estimates on $A_k$.
(A1)~Every within-row ordering of the empirical means agrees with
that of $\theta$; in particular the empirical structure is
$(\mstar,\Eset)$.
(A2)~Every $(i,a)\in\Eset$ has empirical gap
$\widehat\mu_{i,\mstar(i)}-\widehat\mu_{i,a}
\ge\Delta_{i,a}/2\ge\Delta_{\min}/2$, so every
noise-floored empirical quota is at most $8/\Delta_{\min}^2$.
(A3)~Every non-stable matching $m$ has empirical surrogate cost
$\sum_iw_i(\widehat\mu_{i,\mstar(i)}-\widehat\mu_{i,m(i)})_+
\ge\tfrac12\sum_iw_i(\mu_{i,\mstar(i)}-\mu_{i,m(i)})_+\ge d_\theta/2$.
The radius $\bar\delta_\theta$ does not involve $w$, so $A_k$ is the
same event for every variant, including the weight-free $\pi_s$ and
$\pi_\tau$, whose analyses use only (A1) and (A2).

Finally, a deterministic bound that the initialization, the
bad-epoch accounting and the uniform-goodness argument all rely on:
capping the quotas does bound the \emph{mass} of the selected
schedule, which is not automatic from capping the constraint
right-hand sides alone.

\begin{lemma}[Capped quotas bound the schedule mass]\label{lem:mass}
On every sample path and in every epoch $k$, any minimal-$\rho$
solution of the stage-two program (hence the canonical solution)
satisfies
\[
\widehat\rho_k\;\le\;\sum_{(i,a)\in\widehat\Eset_k}\widehat q_{i,a}
\;\le\;NK\kappa_k .
\]
Consequently the exploration phase of epoch $k$ plays at most
$(b_k-b_{k-1})NK\kappa_k+NK$ rounds besides the replay, pathwise,
with no condition on the data.
\end{lemma}

\begin{proof}
Work in the empirical cover formulation. The flow decomposition
in the proof of Theorem~\ref{thm:compact} applies to any nonnegative
marginals with the stated row sums and column bounds; it preserves
arbitrary empirical quota constraints and linear surrogate costs.
This decomposition needs no separation assumption on the empirical
means. Let $x$ be a cover decomposition of a
minimal-$\rho$ solution, so $\widehat\rho_k=\sum_{m}x_m$. Mass on
$\widehat m_k$ covers no quota and only raises $\rho$, so
$x_{\widehat m_k}=0$. Suppose $x_m>0$ and every pair
$(i,a)\in\widehat\Eset_k$ with $m(i)=a$ had strict quota slack.
Reducing $x_m$ by a small $\varepsilon>0$ keeps every quota
satisfied, can only decrease the cost (the surrogate cost
coefficients are nonnegative), hence preserves the stage-two cost
constraint, and reduces $\rho$: a contradiction. So every
positive-mass matching covers at least one \emph{tight} quota.
Charge each $x_m$ to one tight quota that $m$ covers. The total mass
charged to a tight quota $(i,a)$ is at most
$z_{i,a}=\widehat q_{i,a}$, whence
$\widehat\rho_k\le\sum_{(i,a)\,\text{tight}}\widehat q_{i,a}
\le|\widehat\Eset_k|\,\kappa_k\le NK\kappa_k$. The round count
follows from the increment execution with its at most $NK$ ceilings.
\end{proof}

\subsection{Estimation and identification}\label{app:c4-ident}

\begin{lemma}[Estimation and identification]\label{lem:ident}
Almost surely $\widehat\mu(t)\to\mu$, and there is an a.s.\ finite
$K_0$ with $\widehat m_k=\mstar(\theta)$ and
$\widehat\Eset_k=\Eset$ for all $k\ge K_0$. Moreover
$\Prob\bigl(\widehat m_k\neq\mstar\ \text{or}\
\widehat\Eset_k\neq\Eset\bigr)\le C\,e^{-c_\theta k}$.
\end{lemma}

The repair phase of epoch $k$ guarantees
$N_{i,a}\ge\sqrt{b(T_k)}\ge c_1k$ for every pair from the solve step
onward, and repair rounds through epoch $k$ total at most
$NK\lceil\sqrt{b(T_k)}\rceil$. For $T\in(T_{k-1},T_k]$ this is
$O(NK\sqrt{\log T})$, because the per-pair repair
level only ever needs topping up from $\sqrt{b(T_{k-1})}$ to
$\sqrt{b(T_k)}$. Since every pair count tends to infinity, the
strong law gives $\widehat\mu(t)\to\mu$ a.s.

Both $\widehat m_k$ and $\widehat\Eset_k$ are functions of the
within-row orderings of the empirical means \emph{at the solve step};
there every pair count is at least $c_1k$, so on $A_k$ all orderings
agree with those of $\theta$ by (A1), and $A_k$ implies
$\widehat m_k=\mstar$ and $\widehat\Eset_k=\Eset$. By
\eqref{eq:maximal} with $n_0=c_1k$ and a union over $NK$ pairs,
\begin{equation}\label{eq:badprob}
\Prob(A_k^c)\;\le\;NK\,
\frac{2\,e^{-c_1k\bar\delta_\theta^2/2}}{1-e^{-\bar\delta_\theta^2/2}}
\;=\;Ce^{-c_\theta k},
\end{equation}
and Borel--Cantelli yields the a.s.\ finite $K_0$ with $A_k$ holding
for all $k\ge K_0$. \hfill$\qed$

\subsection{Genericity of canonical marginal uniqueness}\label{app:c4-generic}

\begin{lemma}[Where canonical marginal uniqueness fails]\label{lem:generic}
Fix $\theta\in\Thsep$ and let $\mathcal V$ be the finite vertex
set of the feasible polyhedron $\mathcal P(\theta)$ of
\eqref{eq:compactlp}. For $v\in\mathcal V$ write
$G v:=(\sum_a\Delta_{i,a}v_{i,a})_{i\in[N]}$ for its regret vector.
The set $W_{\mathrm{fail}}$ of $w\in\R^N_{++}$ at which
Assumption~\ref{ass:unique} fails is contained in
\[
\bigcup\bigl\{\,w:\langle w,Gv-Gv'\rangle=0\,\bigr\},
\]
the union over pairs $v,v'\in\mathcal V$ with distinct off-stable
projections. In particular $W_{\mathrm{fail}}$ is contained in a
finite union of linear subspaces, and it has Lebesgue measure zero
as soon as no two vertices with distinct off-stable projections
share a regret vector.
\end{lemma}

\begin{proof}
Fix $w>0$ and let $c=c(\theta,w)$. The cost-optimal face
$\widetilde F_w:=\arg\min_{\mathcal P(\theta)}c^\top z$ is
nonempty (Lemma~\ref{lem:fwcompact}) and its recession cone
consists of the directions $(d,\delta)$ of $\mathcal P(\theta)$
with $c^\top d=0$; since $c_{i,a}>0$ off the stable pairs on
$\Thsep$, such a $d$ is supported on stable pairs, and the flow
constraints force $d=\delta\,\zeta^{\mstar}$. Minimizing $\rho$
over $\widetilde F_w$ removes this ray, so the canonical optimal set
$Z_w$ is a compact face of $\mathcal P(\theta)$, the convex hull of
the vertices of $\mathcal P(\theta)$ it contains. Its off-stable
projection is a singleton if and only if all these vertices share
one off-stable block. Hence $w\in W_{\mathrm{fail}}$ implies that
two vertices $v\neq v'$ with distinct off-stable projections are
both cost-optimal at $w$, whence
$0=c^\top(v-v')=\sum_iw_i\sum_a\Delta_{i,a}(v-v')_{i,a}
=\langle w,Gv-Gv'\rangle$. Each such set is a hyperplane if
$Gv\neq Gv'$ and all of $\R^N$ if $Gv=Gv'$, which gives the
containment and the measure statement.
\end{proof}

For the instance $\mu^{(3)}$, $W_{\mathrm{fail}}$ is exactly the
hyperplane $w_2/w_3=0.99$. We argue in the marginal coordinates of
\eqref{eq:compactlp}, in which Assumption~\ref{ass:unique} is stated.
At any $w>0$ a cost-optimal point of \eqref{eq:compactlp} decomposes
into a cost-optimal allocation (Theorem~\ref{thm:compact}), whose
regret vector minimizes $\langle w,\cdot\rangle$ over $\UGL$, hence is
Pareto-minimal and equals some $r(t)$ of
Proposition~\ref{prop:frontier33}; two distinct points $r(t)\neq
r(t')$ tie in weighted cost only if $w_2=0.99\,w_3$. The allocations
with regret vector $r(t)$ are exactly the family points $x(h,y)$ with
$h-y=t$: by Step~2 of Appendix~\ref{app:c2-frontier} such an $x$
dominates a family point $x''\le x$ with
$G(\mu^{(3)})x''\le r(t)$, hence $G(\mu^{(3)})x''=r(t)$ by Step~3, and
$x=x''$ because every column of $G(\mu^{(3)})$ has a positive entry. The off-stable marginals of these family points depend
on $t$ only: $z_{1,a_2}=z_{2,a_3}=200$, $z_{1,a_3}=2$,
$z_{2,a_1}=z_{3,a_2}=200-t$ and $z_{3,a_1}=t+2$. Given this
off-stable block, the row constraints of \eqref{eq:compactlp} force
$\rho\ge202$ and $\rho\ge400-t$, and $\rho=\max(202,400-t)$ is
feasible, with every column constraint then holding with equality.
Off the hyperplane the optimal $t$ is unique, so the cost-optimal set
has a single off-stable block and Assumption~\ref{ass:unique} holds.
On the hyperplane every $t\in[-2,200]$ is cost-optimal; minimizing
$\rho=\max(202,400-t)$ keeps every $t\in[198,200]$, whose off-stable
blocks differ, and the assumption fails.

\subsection{Plug-in convergence}\label{app:c4-plugin}

\begin{lemma}[Plug-in convergence]\label{lem:plugin}
Under Assumption~\ref{ass:unique},
$\widehat z_k^{\mathrm{off}}\to z^\circ_w(\theta)$ almost surely.
\end{lemma}

\emph{A uniform bound on the canonical solution.} On $A_k$, by (A3)
the empirical surrogate cost of every non-stable matching is at least
$d_\theta/2$, and by (A2) the uniform reference cover of
Remark~\ref{rem:feasible} at level $8/\Delta_{\min}^2$, evaluated at
the empirical data, is feasible for the empirical quotas with cost at
most a constant $C_{\mathrm{ref}}$; in particular
$\widehat v_k\le C_{\mathrm{ref}}$. Passing through the cover
formulation \eqref{eq:coverlp} (Theorem~\ref{thm:compact}), let $x$ be
a cover decomposition of the canonical solution. It has no
$\widehat m_k$-component, because such mass covers no quota, costs
nothing and only raises $\rho$ (first step of the proof of
Lemma~\ref{lem:mass}); hence
$\widehat\rho_k=\sum_{m\neq\widehat m_k}x_m$. By stage (ii) its cost
is at most $\widehat v_k+\varepsilon_k\le C_{\mathrm{ref}}+1$, as
$\varepsilon_k=k^{-1/4}\le1$, so
$(d_\theta/2)\sum_{m\neq\widehat m_k}x_m\le C_{\mathrm{ref}}+1$ and
\begin{equation}\label{eq:rhobound}
\widehat\rho_k\;=\;\sum_{m\neq\widehat m_k}x_m\;\le\;
\frac{2(C_{\mathrm{ref}}+1)}{d_\theta}\;=:\;C_\theta
\qquad\text{on }A_k .
\end{equation}

\emph{Data rate.} Analogously to $A_k$, let
$B_k:=\{|\widehat\mu^{(n)}_{i,a}-\mu_{i,a}|<k^{-1/3}\ \forall(i,a),\
\forall n\ge c_1k\}$, again an event about the reward stacks, hence
insensitive to when the repair samples arrive. By \eqref{eq:maximal}
with $n_0=c_1k$ and $\varepsilon=k^{-1/3}$ (whose denominator
contributes a factor $O(k^{2/3})$),
$\Prob(B_k^c)\le CNK\,k^{2/3}e^{-c_1k^{1/3}/2}$, which is summable,
so almost surely $A_k\cap B_k$ holds for all large $k$. All estimates
consulted at the solve step carry counts $\ge c_1k$, so on
$A_k\cap B_k$ the surrogate costs err by $O(k^{-1/3})$; the noise-floor radii are
$O(\sqrt{\log k/k})=o(k^{-1/3})$ (counts at least $c_1k$, $\beta$ of
order $\log k$); the map $g\mapsto2/g^2$ is Lipschitz on
$[\Delta_{\min}/2,\infty)$, so the quotas err by $O(k^{-1/3})$ as
well; and the cap is eventually inactive. The LP values obey
$|\widehat v_k-v^*|=O(k^{-1/3})$: start with a true cost optimizer
and top up each empirical quota deficit by a dedicated completion
matching, as below. The total added mass is $O(k^{-1/3})$, giving
$\widehat v_k\le v^*+O(k^{-1/3})$. Conversely, after discarding stable mass, an empirical cost
optimizer has bounded mass by the argument for \eqref{eq:rhobound};
topping up its
true quota deficits and transferring the costs gives
$v^*\le\widehat v_k+O(k^{-1/3})$. Thus
$|\widehat v_k-v^*|=o(\varepsilon_k)$ for the stage-two slack
$\varepsilon_k=k^{-1/4}$.

\emph{A feasible proxy for the true optimizer.} The true canonical
optimum $(z^*,\rho^*)$ need not satisfy the \emph{empirical} quotas
when $\widehat q_k>q$. Top it up: for each $(i,a)\in\Eset$ fix one
matching $m^{(i,a)}$ with $m^{(i,a)}(i)=a$ and add the mass
$(\widehat q_{k,i,a}-z^*_{i,a})_+\le(\widehat q_{k,i,a}-q_{i,a})_+
=O(k^{-1/3})$ to it (using $z^*\ge q$). The resulting proxy meets
every empirical quota, satisfies the flow constraints after
increasing $\rho$ by the total added mass, and its empirical cost and
its $\rho$ exceed $v^*$ and $\rho^*$ by $O(k^{-1/3})$. Since
$\widehat v_k\ge v^*-o(\varepsilon_k)$, the proxy is feasible for
stage (ii) once $O(k^{-1/3})\le\varepsilon_k/2$: its empirical cost
is at most $v^*+o(\varepsilon_k)\le\widehat v_k+\varepsilon_k$.
Hence $\limsup_k\widehat\rho_k\le\rho^*$, the minimal $\rho$ over the
true optimal face.

\emph{Limit points.} Any limit point $(\bar z,\bar\rho)$ of the
selected solutions is feasible for the true program (the constraint
data converge), has cost $\le\lim(\widehat v_k+\varepsilon_k)=v^*$,
hence is cost-optimal, and has $\bar\rho\le\rho^*$, hence lies in the
minimal-$\rho$ face. By Assumption~\ref{ass:unique} the off-stable
projection of that face is the singleton $\{z^\circ_w(\theta)\}$, so
$\widehat z_k^{\mathrm{off}}\to z^\circ_w(\theta)$ a.s., regardless of the
lexicographic tie-break.
\hfill$\qed$

\subsection{Tracking}\label{app:c4-track}

\begin{lemma}[Tracking]\label{lem:track}
The exploration phases contribute
$\E[N^{\mathrm{exp}}_{i,a}(T)]=z^\circ_{w,i,a}\log T+o(\log T)$ to
every off-stable pair, and bad epochs (those with a wrong
$\widehat m_k$ or aberrant quotas) contribute only $O(1)$ in
expectation.
\end{lemma}

Fix an off-stable pair $p$. The exploration phase of epoch $k$
contributes $(b_k-b_{k-1})\widehat z_{k,p}+O(NK)$ to $N_p$, the
rounding term coming from at most $NK-N+1$
ceilings (Proposition~\ref{prop:sparse}); on $A_k$ the discarded
$\widehat m_k$-component does not affect off-stable pairs.

\emph{Regular epochs.} On $A_k$, \eqref{eq:rhobound} gives the
\emph{uniform} bound $\widehat z_{k,p}\le\widehat\rho_k\le C_\theta$,
and $\widehat z_{k,p}\mathbf1_{A_k}\to z^\circ_p$ a.s.\
(Lemma~\ref{lem:plugin} and $\mathbf1_{A_k}\to1$). Dominated
convergence with the constant dominator $C_\theta$ gives
$\E[\widehat z_{k,p}\mathbf1_{A_k}]\to z^\circ_p$, and the weighted
Ces\`aro average yields
$\sum_{k\le K}(b_k-b_{k-1})\,\E[\widehat z_{k,p}\mathbf1_{A_k}]
=b_Kz^\circ_p+o(b_K)$.

\emph{Bad epochs.} Pathwise, epoch $k$ explores at most
$(b_k-b_{k-1})NK\kappa_k+NK\le CNKk^2$ rounds
(the capped-mass Lemma~\ref{lem:mass}), so by \eqref{eq:badprob},
$\sum_k\E[(b_k-b_{k-1})\widehat z_{k,p}\mathbf1_{A_k^c}]
\le\sum_kCNKk^2e^{-c_\theta k}=O(1)$.

\emph{Anytime.} For $T\in(T_{K-1},T_K]$ the counts are sandwiched
between their values at $T_{K-1}$ and $T_K$, and
$b(T_{K-1})/b(T_K)\to1$, so
$\E[N^{\mathrm{exp}}_p(T)]=z^\circ_p\log T+o(\log T)$.
\hfill$\qed$

\subsection{Certified exploitation}\label{app:c4-certify}

\begin{lemma}[Certified exploitation]\label{lem:certify}
The expected number of exploitation rounds in which the played
matching differs from $\mstar(\theta)$ is $O(1)$, and the expected
number of estimation rounds is $o(\log T)$.
\end{lemma}

\paragraph{(a) Wrong certified exploitation.}
Suppose at round $t$ the policy exploits the certified candidate
$\widetilde m_t\neq\mstar$, and let $i$ be the earliest deviation
level, $a^\star=\mstar(i)$, $\widetilde a=\widetilde m_t(i)$, both
free at level $i$. If both estimates were within their certification
radii, then
$\widehat\mu_{i,\widetilde a}-\widehat\mu_{i,a^\star}
<(\mu_{i,\widetilde a}-\mu_{i,a^\star})
+\mathrm{rad}_t(i,\widetilde a)+\mathrm{rad}_t(i,a^\star)
<\mathrm{rad}_t(i,\widetilde a)+\mathrm{rad}_t(i,a^\star)$,
contradicting the certificate; hence some pair deviates by its
$c(t)$-radius at its current random count, which is at most $t$. By
Lemma~\ref{lem:peel} with $c=c(t)$ and a union over $NK$ pairs,
\[
\Prob\bigl(\text{some pair deviates by its }c(t)\text{-radius at
round }t\bigr)
\;\le\;NK\,C_1\,c(t)\log t\,e^{-c(t)}
\;\le\;\frac{C}{t(\log t)^{\xi-2}},
\]
which is summable over $t$ because $\xi>3$; the expected number of
wrong certified exploitation rounds is $O(1)$.

\paragraph{(b) Classification, and rounds involving a deviated
entry.}
Fix the threshold
$\delta'':=\min\bigl(\delta_\theta/2,\
\min_{(i,a)\in\Eset}\Delta_{i,a}/4\bigr)$ and call an entry
\emph{deviated} at a given moment if its estimate errs by at least
$\delta''$ at its current count. Classify every estimation round by
the two entries of its \emph{sampled comparison}: \emph{type B} if
one of them is deviated, \emph{type C} otherwise. For type B, the
round samples the least-sampled of the two involved pairs. A given
pair $p$ participates in at most $2K$ distinct comparisons (as the
competing arm against any candidate entry of its level, or as the
candidate entry against any free arm), and while $p$ is deviated at
count $n$, each such comparison can produce at most $n+1$ type-B
rounds that sample the \emph{other} pair (whose count must first
catch up to $n$), plus the single round that samples $p$ itself and
moves it off count $n$: at most $2K(n+1)+1$ rounds per pair and
count in total.
Charging each type-B round to its deviated pair at its current count,
\[
\E[\text{type-B rounds}]
\;\le\;NK\sum_{n\ge1}\bigl(2K(n+1)+1\bigr)\cdot2\,e^{-n\delta''^2/2}
\;=\;O(1).
\]

\paragraph{(c) Type-C rounds.}
In a type-C round both compared entries are clean. The candidate's
choice at the sampled level is the empirical argmax over the free
arms, so if the competitor is clean too, the two true means are
ordered as the empirical ones, and the compared \emph{true} gap is
positive, hence at least $2\delta_\theta$; the compared empirical gap
is therefore at least $2\delta_\theta-2\delta''\ge\delta_\theta$. In
particular a clean comparison is certified as soon as both its counts
reach $8c(t)/\delta_\theta^2$, and since the round samples the
smaller count, \emph{every} type-C round samples a pair whose count
is below $8c(t)/\delta_\theta^2$; the type-C rounds within epoch $k$
therefore number at most $NK(8c(T_k)/\delta_\theta^2+1)\le CNKk^2$
pathwise. Two remarks on this cap. First, it is a \emph{local} bound
used only to multiply small probabilities below; summing it over
epochs would give $O((\log T)^{3/2})$, which is \emph{not} how the
global count is obtained. Second, the same argument applied at the
horizon already gives a global pathwise bound: every type-C round
increments the count of a pair currently below
$8c(T)/\delta_\theta^2$, counts never decrease, and $c$ is
increasing, so each pair can be sampled in type-C rounds at most
$8c(T)/\delta_\theta^2+1$ times before it crosses the demand
threshold once and for all; type-C rounds thus total at most
$NK(8c(T)/\delta_\theta^2+1)=O(\log T)$ pathwise. The refinement
below improves this $O(\log T)$ to expected $O(1)+o(\log T)$ by
splitting on where the deviation sits.

\emph{Shadow rounds} (some entry elsewhere is deviated, so the
candidate need not equal $\mstar$). A shadow round in epoch $k$
requires an entry deviated at its current count, which is at least
$\sqrt{b(T_{k-1})}\ge c_1k$ by repair; by \eqref{eq:maximal} this has
probability at most $Ce^{-c''k}$. With the pathwise per-epoch cap
above,
\[
\E[\text{shadow rounds}]
\;\le\;\sum_kCNKk^2\cdot Ce^{-c''k}\;=\;O(1).
\]

\emph{Equilibrium rounds} (no entry is deviated anywhere, so
$\widetilde m_t=\mstar$ and the sampled comparison is
$(\mstar(i),a)$ with $(i,a)\in\Eset$; which of its two pairs is
sampled matters only in the late regime below, where it is settled,
and the early regime is handled pathwise). Fix $\gamma\in(0,1/2)$, set the
threshold
$\varepsilon_\gamma:=\min\bigl(\delta''/2,\
\gamma\Delta_{\min}/8\bigr)$, and let
$k_\gamma(T):=\lceil C_2\log\log T\rceil$ with
$C_2=C_2(\theta,\gamma)$ large enough that the event
\[
\mathcal G_T:=\bigl\{\text{every entry is within }\varepsilon_\gamma
\text{ of its mean at every count}\ge c_1k_\gamma(T)\bigr\}
\]
has $\Prob(\mathcal G_T^c)\le(\log T)^{-2}$, which \eqref{eq:maximal}
provides. The threshold plays two roles: since
$\varepsilon_\gamma<\delta''$, on $\mathcal G_T$ \emph{no entry is deviated}
at any count $\ge c_1k_\gamma$; and since
$2\varepsilon_\gamma\le\gamma\Delta_{\min}/4$, the empirical gaps at
such counts satisfy
$\widehat\Delta_{i,a}\ge\Delta_{i,a}(1-\gamma/2)$. Call epoch $k$
\emph{early} if $k\le2k_\gamma(T)$ and \emph{late} otherwise; the
buffer $k_\gamma<k\le2k_\gamma$ is deliberate, because the first
epochs after $k_\gamma$ have not yet accumulated enough trustworthy
replay for the stable-side radius to be small.

\emph{Early epochs, pathwise.} Every equilibrium round is a type-C
round and samples a pair whose count is below
$8c(t)/\delta_\theta^2\le8c(T_{2k_\gamma})/\delta_\theta^2$; counts
never decrease, so the equilibrium rounds of early epochs total at
most $NK(8c(T_{2k_\gamma})/\delta_\theta^2+1)
=O\bigl((\log\log T)^2\bigr)$ on every path.

\emph{Late epochs, on $\mathcal G_T$.} For every epoch $j>k_\gamma$ the
counts consulted at its solve step exceed $c_1j\ge c_1k_\gamma$, so
on $\mathcal G_T$ no entry is deviated there and
$\widehat m_j=\mstar$: the replay of \emph{every} epoch
$j>k_\gamma$ plays the true stable matching. At any estimation round $t$ of a
late epoch $k>2k_\gamma$, the replays of epochs
$k_\gamma+1,\dots,k$ have all been executed (the replay is part of
phase 3), so every stable pair carries at least
\[
\sum_{j=k_\gamma+1}^{k}(b_j-b_{j-1})\,j
\;\ge\;c'\bigl(k^3-(k/2)^3\bigr)\;=\;\Theta\bigl(k\,b(T_k)\bigr)
\]
trustworthy samples, using $k_\gamma<k/2$. The stable-side
certification radius is therefore bounded by $Ck^{-1/2}$, which is below
$\gamma\Delta_{i,a}/4$ for all late epochs once $T$ is large
(late epochs have $k>2k_\gamma(T)\to\infty$). The demand
\[
\mathrm{Dem}_{i,a}(t):=\frac{2c(t)}
{\bigl(\widehat\Delta_{i,a}(t)-\mathrm{rad}^{\mathrm{st}}_t\bigr)^2}
\]
therefore has a strictly positive denominator at every late round on
$\mathcal G_T$; outside this regime we never use $\mathrm{Dem}$, since for
$\widehat\Delta\le\mathrm{rad}^{\mathrm{st}}$ the failed
certificate carries no count information. At a late equilibrium
round the failed certificate gives
$N_{i,a}(t)<\mathrm{Dem}_{i,a}(t)$, and
$\mathrm{Dem}_{i,a}(t)=O(b(T_k))$ by the bound below, whereas the
stable side carries $\Theta(k\,b(T_k))$ samples; so for $T$ large the
least-sampled pair of the comparison, which the round samples, is
$(i,a)$ itself. Let $k(t)$ be the epoch
containing $t$ and
$\mathrm{Sup}_{i,a}(t):=\sum_{k_\gamma<k\le k(t)}(b_k-b_{k-1})
\widehat z_{k,(i,a)}$ the late-epoch supply; every estimation round
lies in phase 4 of its epoch, after that epoch's exploration phase
has completed, so at any equilibrium round $t$ the count
$N_{i,a}(t)$ already includes $\mathrm{Sup}_{i,a}(t)$, current epoch included.
Since at each late equilibrium round the current count, which is at
least the supply plus all previous late equilibrium samples, still
lies below the demand, and counts never decrease, the late
equilibrium rounds on $(i,a)$ obey the ratchet bound (its sharper
form, evaluated at the last such round, is (L4) and the ratchet
step in the proof of Lemma~\ref{lem:ug})
\[
N^{\mathrm{eq,late}}_{i,a}(T)\;\le\;
\max_{t\le T\ \text{late}}
\bigl[\mathrm{Dem}_{i,a}(t)-\mathrm{Sup}_{i,a}(t)\bigr]_+ +1
\qquad\text{on }\mathcal G_T.
\]
On $\mathcal G_T$ the late quotas obey
$\widehat q_{k,(i,a)}\ge q_{i,a}(1-\gamma)$ (the noise floor is
deterministic and vanishing at these counts, and the cap
$\kappa_k=k>q_{i,a}$ for $k>k_\gamma$ and $T$ large), so with
$\widehat z_{k,(i,a)}\ge\widehat q_{k,(i,a)}$ and
$b(T_{k(t)})\ge b(t)$,
\[
\mathrm{Dem}_{i,a}(t)\le\frac{2c(t)}{\Delta_{i,a}^2}(1+C\gamma),
\qquad
\mathrm{Sup}_{i,a}(t)\ge\bigl(b(t)-b(T_{k_\gamma})\bigr)q_{i,a}(1-\gamma),
\]
and hence, using $c(t)\le b(t)$ and $b(T_{k_\gamma})=
O((\log\log T)^2)$,
\[
N^{\mathrm{eq,late}}_{i,a}(T)
\;\le\;C'\gamma\log T+O\bigl((\log\log T)^2\bigr)
\qquad\text{on }\mathcal G_T.
\]
\emph{Off $\mathcal G_T$.} All equilibrium rounds are type-C rounds and
their pathwise total is at most
$NK(8c(T)/\delta_\theta^2+1)=O(\log T)$, so their expected
contribution is at most $O(\log T)\cdot(\log T)^{-2}=o(1)$.
Therefore
$\limsup_T\E[N^{\mathrm{eq}}(T)]/\log T\le C'\gamma$ for every
$\gamma>0$, i.e., $\E[N^{\mathrm{eq}}(T)]=o(\log T)$. Summing the
three contributions, $\E[N^{\mathrm{est}}(T)]=O(1)+o(\log T)=o(\log T)$.
\hfill$\qed$

\subsection{Accounting and proof of Theorem~\ref{thm:achieve}}
\label{app:c4-ug}

\begin{lemma}[Regret accounting]\label{lem:account}
Under Assumption~\ref{ass:complete},
$R_i(T;\theta)=\sum_{a}\Delta_{i,a}\,\E_\theta[N_{i,a}(T)]$, and the
total off-stable counts decompose into exploration
(Lemma~\ref{lem:track}), repair ($O(NK\sqrt{\log T})$), estimation
($o(\log T)$) and wrong exploitation ($O(1)$), yielding
$\E[N_{i,a}(T)]=z^\circ_{w,i,a}\log T+o(\log T)$.
\end{lemma}

\paragraph{Accounting.}
Under Assumption~\ref{ass:complete}, \eqref{eq:regretcount} in pair
form reads $R_i(T)=\sum_a\Delta_{i,a}\E[N_{i,a}(T)]$. Every round
belongs to exactly one phase, so for an off-stable pair $p$,
\[
\E[N_p(T)]
=\underbrace{z^\circ_p\log T+o(\log T)}_{\text{explore
(Lem.~\ref{lem:track})}}
+\underbrace{O(\sqrt{\log T})}_{\text{repair}}
+\underbrace{o(\log T)}_{\text{estimation
(Lem.~\ref{lem:certify})}}
+\underbrace{O(1)}_{\text{wrong exploit}}
=z^\circ_p\log T+o(\log T),
\]
while correctly certified exploitation and the stable-matching
replay on regular epochs play $\mstar$ and touch only zero-gap
pairs; the replay on irregular epochs contributes $O(k^2)$ rounds
with probability $Ce^{-c_\theta k}$, hence $O(1)$ in expectation, and
is absorbed into the last term. Dividing by $\log T$ gives the count
limits of
Theorem~\ref{thm:achieve}; the regret limits follow from the
accounting identity, and the weighted limit from
$\sum_iw_ir^\circ_i=\sum_{i,a}w_i\Delta_{i,a}z^\circ_{i,a}
=C_w(\theta)$, the cost optimality of $z^\circ_w$. Matching this
against Theorem~\ref{thm:weighted} shows the bound is attained
exactly.

The uniform-goodness claim of Theorem~\ref{thm:achieve} is the
statement of Lemma~\ref{lem:ug} for $\pi_w$; it is proved in
Appendix~\ref{app:c4-strict} at an arbitrary strict instance, for
the four policy variants at once.
\hfill$\qed$

\subsection{Proof of Proposition~\ref{prop:frontier-attain}}
\label{app:c4-frontier33}

\begin{proposition}[Attainability of the full $3\times3$ segment]
\label{prop:frontier-attain}
Fix $s\in[0,1]$ and let $\pi_s$ use the empirical $s$-split whenever
the eligible structure matches Example~\ref{ex:33}, with capped
dedicated completions otherwise (Algorithm~\ref{alg:policy}). Then
$\pi_s$ is uniformly good on
$\Thstrict$, and on $\mu^{(3)}$,
$R(T)/\log T\to r(202s-2)=(22,\,222-202s,\,2.02+199.98s)$: as $s$
ranges over $[0,1]$ the policy attains every expected-regret
coefficient on the segment
of Proposition~\ref{prop:frontier33}.
\end{proposition}

The policy $\pi_s$ differs from $\pi_w$ only in step 2, so
Lemma~\ref{lem:ident} (repair, identification, \eqref{eq:badprob})
and the certification analysis of Lemma~\ref{lem:certify} apply
verbatim; for the supply side of the demand--supply argument, note that
 the $s$-split meets every empirical quota by construction
($\widehat x_{\mathrm A}+\widehat x_{\mathrm B}=\widehat q_{12}$,
$\widehat x_{\mathrm B}+\widehat x_{\mathrm C}=\widehat q_{23}$,
$\widehat x_{\mathrm D}+\widehat x_{\mathrm E}=\widehat q_{13}$, all
components nonnegative since
$\widehat x_{\mathrm B}\le\min(\widehat q_{12},\widehat q_{23})$),
which is the only property of the schedule that
Lemma~\ref{lem:certify}(c) uses.

The plug-in step becomes elementary: the target allocation is an
explicit \emph{continuous} function of
$(\widehat q_{12},\widehat q_{13},\widehat q_{23})$, so no LP
perturbation or uniqueness argument is required. On $A_k$ (eventually
a.s.), the structure matches Example~\ref{ex:33}, the labels agree
with the true ones, and $(\widehat q_{12},\widehat q_{13},\widehat q_{23})
\to(200,2,200)$ (noise floor
vanishing, cap inactive), whence
\[
\widehat x_k\;\longrightarrow\;
 x(s):=\bigl(200(1-s),\,200s,\,200(1-s),\,2(1-s),\,2s\bigr)
\quad\text{a.s.},
\]
which is the family point of Proposition~\ref{prop:frontier33} with
$h=200s$, $y=2(1-s)$, i.e., $t=h-y=202s-2$. The tracking argument of
Lemma~\ref{lem:track} carries over with $\widehat z_k$ replaced by
the pair marginals of $\widehat x_k$: on $A_k$ these are uniformly
bounded by $\widehat q_{12}+\widehat q_{13}+\widehat q_{23}\le
C_\theta'$, giving the dominated Ces\`aro limit, and bad epochs are
$O(1)$ in expectation through the caps and \eqref{eq:badprob}
exactly as before (the fallback dedicated cover is also capped).
The accounting of Lemma~\ref{lem:account} then yields
\[
\begin{aligned}
\frac{\E[N_{i,a}(T)]}{\log T}&\longrightarrow z_{i,a}(s)
:=\sum_{m}x_m(s)\,\ind{m(i)=a},
\quad a\neq\mstar(i),\\
\frac{R(T)}{\log T}&\longrightarrow G\bigl(\mu^{(3)}\bigr)x(s)
=r(202s-2),
\end{aligned}
\]
the last identity by Step 1 of Appendix~\ref{app:c2-frontier}.
Uniform goodness on $\Thstrict$ is Lemma~\ref{lem:ug}
(Appendix~\ref{app:c4-strict}), which covers $\pi_s$.
\hfill$\qed$

\begin{corollary}[The $3\times3$ frontier is closed]
\label{cor:frontier-closed}
For $\mu^{(3)}$ the necessary lower boundary is the Pareto-minimal
set of asymptotic regret vectors attainable by uniformly good
complete-matching policies: no attainable vector lies below it
(Proposition~\ref{prop:region}), every segment point is attained by
some $\pi_s$, and every attainable vector lies weakly above a segment
point.
Theorem~\ref{thm:closure} extends this to every separated
instance.
\end{corollary}

\subsection{The cost-optimal face}\label{app:c4-fw}

\begin{lemma}[The projected cost-optimal face is a compact polytope]
\label{lem:fwcompact}
For $\theta\in\Thsep$ and $w>0$, $F_w(\theta)$ is a nonempty
compact polytope.
\end{lemma}

\begin{proof}
$F_w(\theta)$ is a polyhedron as the projection of one. The LP is
feasible by Remark~\ref{rem:feasible} and has finite nonnegative
value, hence attains its minimum, so $F_w(\theta)$ is nonempty.
For boundedness, note that
on $\Thsep$ the stable arm $\mstar(i)$ is player $i$'s globally
best arm, so by row strictness every off-stable coefficient is
strictly positive, $c_{i,a}=w_i\Delta_{i,a}>0$ for all
$a\neq\mstar(i)$. With
$c_{\min}:=\min_{(i,a):\,a\neq\mstar(i)}w_i\Delta_{i,a}>0$, every
$y\in F_w(\theta)$ satisfies
$c_{\min}\|y\|_1\le c(\theta,w)^\top y=C_w(\theta)$, so
$F_w(\theta)$ is contained in a ball of radius
$C_w(\theta)/c_{\min}$.
\end{proof}

In the $3\times3$ instance at the knife-edge weight, $F_w$ is the
off-stable image of the two-parameter allocation family of
Proposition~\ref{prop:frontier33}, not just its minimal-$\rho$ region.

\subsection{Proof of Theorem~\ref{thm:frontier-achieve}}
\label{app:c4-c4b}

Throughout this subsection fix $w\in\R^N_{++}$ and a target rule
$\Psi$ as in \S\ref{sec:c4-frontier}, and write
$\tau_k:=\Psi(\widehat m_k,\widehat\Eset_k,\widehat q_k,
\widehat\Delta_k)$ for the target used in epoch $k$, where
$\widehat\Delta_{k,i,a}:=\widehat\mu_{i,\widehat m_k(i)}
-\widehat\mu_{i,a}$ on the pairs with $a\neq\widehat m_k(i)$. For
the convergence statements fix also $\theta\in\Thsep$, write
$\mstar=\mstar(\theta)$, let ``off'' denote the coordinates $(i,a)$
with $a\neq\mstar(i)$, let $\tau:=\Psi(\mstar,\Eset,q,\Delta)$ be
the rule's value at the true structure, and abbreviate
$\Phi_\tau(y):=\|y-\tau^{\mathrm{off}}\|^2$ on the off-stable
space; a constant rule is the fixed-target policy $\pi_{w,\tau}$.
The uniform-goodness part at the end works on an arbitrary strict
instance and uses only the cap. Constants depend on
$(\theta,w,\Psi,\alpha,\xi,N,K)$.

\paragraph{The projected program.}
The surrogate cost vanishes on stable coordinates, so
$c^\top z=c_{\mathrm{off}}^\top z^{\mathrm{off}}$ for every feasible
$(z,\rho)$. Hence for every $\eta>0$ and every cap level $B$,
$(z,\rho)$ minimizes $c^\top z+\eta\Phi_\tau(z^{\mathrm{off}})$ over
$\mathcal P(\theta)\cap\{\rho\le B\}$ if and only if
$z^{\mathrm{off}}$ minimizes
\[
\psi_\eta(y)\;:=\;c_{\mathrm{off}}^\top y+\eta\,\Phi_\tau(y)
\]
over the projected polyhedron
$Y_B(\theta):=\operatorname{proj}_{\mathrm{off}}
\bigl(\mathcal P(\theta)\cap\{\rho\le B\}\bigr)$, and the off-stable
projection of the argmin equals the argmin of $\psi_\eta$; lifting
and projection preserve both feasibility and the objective. The same
identification at $\eta=0$ shows that the minimizers of
$c_{\mathrm{off}}^\top y$ over
$Y_\infty(\theta):=\operatorname{proj}_{\mathrm{off}}
\mathcal P(\theta)$ form exactly $F_w(\theta)$
(Lemma~\ref{lem:fwcompact}), with minimum value $C_w(\theta)$
(Theorem~\ref{thm:compact}).

\begin{lemma}[Exact selection, cap-compatible]\label{lem:projsel}
Let $y^\tau$ be the Euclidean projection of $\tau^{\mathrm{off}}$
onto $F_w(\theta)$ and
$\rho^\tau_{\min}:=\min\{\rho:(z,\rho)\in\mathcal P(\theta),\
z^{\mathrm{off}}=y^\tau\}$; the minimum is attained because the
completion fiber is a nonempty closed polyhedron on which
$\rho\ge0$. There exists
$\bar\eta=\bar\eta(\theta,w,\tau)>0$ such that for every
$0<\eta<\bar\eta$ and every $B\ge\rho^\tau_{\min}$,
\[
\arg\min_{y\in Y_B(\theta)}\psi_\eta(y)=\{y^\tau\}.
\]
\end{lemma}

\begin{proof}
\emph{Uncapped.} The selection problem, minimizing $\Phi_\tau$ over
$F_w(\theta)=\{y\in Y_\infty(\theta):
c_{\mathrm{off}}^\top y\le C_w(\theta)\}$, is a strictly convex
minimization over a nonempty compact polyhedron, with unique
solution $y^\tau$. ($F_w(\theta)$ is bounded: on $\Thsep$ every
off-stable cost $c_{i,a}=w_i\Delta_{i,a}$ is positive and every
recession direction of $Y_\infty(\theta)$ is nonnegative, so a
nonzero recession direction has positive cost.) We show directly
that for small $\eta$ the regularized minimizers lie in
$F_w(\theta)$; this is the exact-regularization theorem of
\citet[Theorem~2.1 and Corollary~2.3]{friedlander2007exact}, going
back to \citet{mangasarian1979nonlinear} for linear programs, with
an explicit threshold. Since $F_w(\theta)$ is cut out of the
polyhedron $Y_\infty(\theta)$ by the single inequality
$c_{\mathrm{off}}^\top y\le C_w(\theta)$, Hoffman's bound
\citep{hoffman1952approximate} gives a constant
$H=H(\theta,w)\in(0,\infty)$ with
\[
\operatorname{dist}\bigl(y,F_w(\theta)\bigr)
\;\le\;H\bigl(c_{\mathrm{off}}^\top y-C_w(\theta)\bigr)
\qquad\text{for all }y\in Y_\infty(\theta),
\]
where the right-hand side is nonnegative because
$C_w(\theta)=\min_{Y_\infty(\theta)}c_{\mathrm{off}}^\top y$
(Euclidean norms throughout).
Let $y_\eta$ minimize $\psi_\eta$ over $Y_\infty(\theta)$ (the
argmin is nonempty because $\Phi_\tau$ is coercive and
$c_{\mathrm{off}}\ge0$), let $p$ be the projection of $y_\eta$ onto
$F_w(\theta)$ and $d:=\|y_\eta-p\|$. Convexity of $\Phi_\tau$ and
$\Phi_\tau(p)\ge\Phi_\tau(y^\tau)$ give
$\Phi_\tau(y_\eta)\ge\Phi_\tau(p)-2\|p-\tau^{\mathrm{off}}\|\,d
\ge\Phi_\tau(y^\tau)-D\,d$ with
$D:=2\max_{F_w(\theta)}\|\cdot-\tau^{\mathrm{off}}\|$. Hence
\[
0\;\ge\;\psi_\eta(y_\eta)-\psi_\eta(y^\tau)
\;=\;\bigl(c_{\mathrm{off}}^\top y_\eta-C_w(\theta)\bigr)
+\eta\bigl(\Phi_\tau(y_\eta)-\Phi_\tau(y^\tau)\bigr)
\;\ge\;(1/H-\eta D)\,d,
\]
so $d=0$ whenever $\eta<\bar\eta:=1/(HD)$ (if $D=0$ then
$F_w(\theta)=\{\tau^{\mathrm{off}}\}$ and any $\bar\eta$ works).
Thus every minimizer of $\psi_\eta$ lies in $F_w(\theta)$, where
$c_{\mathrm{off}}^\top y\equiv C_w(\theta)$ and the regularized
objective ranks by $\Phi_\tau$ alone; strict convexity gives
$y_\eta=y^\tau$.

\emph{Capped.} $Y_B(\theta)\subseteq Y_\infty(\theta)$, and
$y^\tau\in Y_B(\theta)$ because $B\ge\rho^\tau_{\min}$ provides a
feasible completion. The unique global minimizer of $\psi_\eta$ over
the larger set thus lies in the smaller one, so the capped and
uncapped minima coincide, and every capped minimizer is an uncapped
minimizer, hence equals $y^\tau$.
\end{proof}

\begin{lemma}[Feasibility and pathwise mass under the cap]
\label{lem:capfeas}
On every sample path and in every epoch $k$, the capped empirical
program \eqref{eq:regprog} is feasible, every solution satisfies
$\widehat\rho_k\le NK\kappa_k$, and the exploration phase of epoch
$k$ plays at most $(b_k-b_{k-1})NK\kappa_k+NK$ rounds besides the
replay.
\end{lemma}

\begin{proof}
The dedicated cover that plays $\mathrm{comp}(i,a)$ with mass
$\widehat q_{i,a}$ for each $(i,a)\in\widehat\Eset_k$ meets every
quota, and its marginals, completed with
$\rho=\sum\widehat q\le|\widehat\Eset_k|\,\kappa_k\le NK\kappa_k$,
satisfy the flow constraints: each row sums to the total mass and
each column is used at most once per matching. The mass bound is the
cap constraint itself, and the round count follows from the
increment execution with its at most $NK$ ceilings, as in
Lemma~\ref{lem:mass}.
\end{proof}

\begin{remark}[Why the cap is explicit]\label{rem:whycap}
For $\pi_w$ the mass bound is a theorem (the capped-mass
Lemma~\ref{lem:mass}): reducing the mass of a matching that covers
no tight quota preserves feasibility, does not increase the linear
cost, and lowers $\rho$. Under regularization this surgery can
\emph{increase} the objective, since removing mass moves
$z^{\mathrm{off}}$ away from $\tau$ wherever
$z^{\mathrm{off}}<\tau$; a one-dimensional instance is
$\min_{x\ge1}\eta(x-10)^2$, whose minimizer is $10$, not $1$. The
pathwise bound that the bad-epoch accounting and the
uniform-goodness argument rely on is therefore imposed as a
constraint rather than derived.
\end{remark}

\begin{lemma}[Regularized plug-in convergence]\label{lem:regplugin}
Let $\theta\in\Thsep$ and $\delta_k:=k^{-1/3}$. There are a
deterministic index $k_1=k_1(\theta,w,\Psi)$ and constants such
that on $A_k\cap B_k$ with $k\ge k_1$,
\[
\bigl\|\widehat z_k^{\mathrm{off}}-y^\tau\bigr\|
\;\le\;C\Bigl(\delta_k+\sqrt{\delta_k/\eta_k}\Bigr).
\]
In particular $\widehat z_k^{\mathrm{off}}\to y^\tau$ almost surely.
\end{lemma}

\begin{remark}[Choice of $\eta_k$]\label{rem:etarate}
The proof uses only $\eta_k<\bar\eta$ eventually and
$\delta_k/\eta_k\to0$, so any $\eta_k\downarrow0$ satisfying the latter
condition is admissible. Once $\eta_k$ is below the instance-dependent
threshold $\bar\eta$, a larger $\eta_k$ reduces the displayed error bound.
The logarithmic schedule of \eqref{eq:regprog} gives
$O(k^{-1/6}\sqrt{\log k})$, whereas $\eta_k=k^{-1/4}$ gives
$O(k^{-1/24})$; the latter factor is still about $0.83$ at $k=100$.
These bounds do not compare the policies' finite-horizon regret.
For stage two, $\varepsilon_k=k^{-1/4}$ satisfies the required
$\varepsilon_k\to0$ and $\delta_k/\varepsilon_k\to0$.
\end{remark}

\begin{proof}
Choose $k_1$ so that for $k\ge k_1$: on $A_k$ the empirical
structures agree, $\widehat m_k=\mstar$ and
$\widehat\Eset_k=\Eset$ (proof of Lemma~\ref{lem:ident});
$\eta_k<\bar\eta$; $NK\kappa_k\ge\rho^\tau_{\min}+1$; and on
$A_k\cap B_k$ the quota caps and noise floors are inactive with
$\|\widehat q_k-q\|_\infty\le C\delta_k$ (proof of
Lemma~\ref{lem:plugin}).

\emph{Step 0: target drift.} On $A_k\cap B_k$ with $k\ge k_1$ the
empirical structure is the true one, every count consulted at the
solve step is at least $c_1k$, and the off-stable gaps and the
quotas are within $C\delta_k$ of $(\Delta,q)$ (proof of
Lemma~\ref{lem:plugin}); enlarging $k_1$ so that these data lie in
the neighborhood on which the rule is Lipschitz gives
$\|\tau_k-\tau\|\le C\delta_k$. Hence for $v$ in any fixed ball,
$|\Phi_{\tau_k}(v)-\Phi_\tau(v)|
=|\langle\tau-\tau_k,\,2v-\tau_k-\tau\rangle|\le C\delta_k$, and
since $\eta_k\le1$ the empirical regularized objective
$\widehat J_k(z):=\widehat c_k^\top z
+\eta_k\Phi_{\tau_k}(z^{\mathrm{off}})$ differs on that ball from
$\widehat c_k^\top z+\eta_k\Phi_\tau(z^{\mathrm{off}})$ by at most
$C\delta_k$. For a constant rule, $\tau_k\equiv\tau$ and this step
is void.

\emph{Step 1: constant mass on regular epochs.} Let the proxy
$\widetilde z_k$ be the minimal-$\rho$ completion of $y^\tau$ topped
up to the empirical quotas: for each $(i,a)\in\Eset$ add mass
$(\widehat q_{k,i,a}-y^\tau_{i,a})_+\le(\widehat q_{k,i,a}
-q_{i,a})_+\le C\delta_k$ on a fixed completion matching with
$m(i)=a$, using $y^\tau\ge q$. The proxy is feasible for
\eqref{eq:regprog}, its $\rho$ being at most
$\rho^\tau_{\min}+CNK\delta_k\le NK\kappa_k$, and
$\widehat J_k(\widetilde z_k)\le C_0$ for a constant $C_0$: the
costs are bounded on $A_k$ and, by Step~0,
$\Phi_{\tau_k}(\widetilde z^{\mathrm{off}}_k)\le
\Phi_\tau(\widetilde z^{\mathrm{off}}_k)+C\delta_k\le
\Phi_\tau(y^\tau)+C\delta_k$. On $A_k$ the empirical surrogate cost
of every non-stable matching is at least $d_\theta/2$ by (A3). Take a
cover decomposition of the
minimal-$\rho$ solution $\widehat z_k$; it has no
$\widehat m_k$-component, since deleting positive mass on
$\widehat m_k$ leaves the off-stable block, hence the regularized
objective, unchanged while decreasing $\rho$. Therefore
$\widehat\rho_k=\sum_{m\neq\widehat m_k}\widehat x_{k,m}$, and
optimality of $\widehat z_k$ against the proxy together with
$\Phi_\tau\ge0$ gives
\[
\frac{d_\theta}{2}\,\widehat\rho_k
\;\le\;\widehat c_k^\top\widehat z_k
\;\le\;\widehat J_k(\widehat z_k)
\;\le\;\widehat J_k(\widetilde z_k)\;\le\;C_0,
\qquad\text{so }\ \widehat\rho_k\le\frac{2C_0}{d_\theta}.
\]
All points
appearing below therefore lie in a fixed box, on which the empirical
and true objectives differ by at most $C\delta_k$ and both are
Lipschitz.

\emph{Step 2: Hoffman repair.} On $A_k\cap B_k$ with $k\ge k_1$ the capped
empirical and true feasible systems consist of the \emph{same}
constraint matrix (flow rows, quota rows indexed by $\Eset$, the cap
row, nonnegativity) with right-hand sides differing by at most
$C\delta_k$ (only the quota entries differ). By Hoffman's error
bound \citep{hoffman1952approximate}, with the constant
$H_{\mathcal P}$ of this fixed matrix (distinct from the constant $H$
of Lemma~\ref{lem:projsel}, which refers to the single cost
inequality), there is a point $\bar z_k$ of the true capped
feasible set with $\|\bar z_k-\widehat z_k\|\le CH_{\mathcal P}\delta_k$.

\emph{Step 3: objective transfer.} Write
$J_k(z):=c^\top z+\eta_k\Phi_\tau(z^{\mathrm{off}})$ for the true
regularized objective; $J_k$ and $\widehat J_k$ depend on $z$ only
through $z^{\mathrm{off}}$ and differ by at most $C\delta_k$ on the
box (costs on $A_k\cap B_k$, targets by Step~0). Chaining optimality of
$\widehat z_k$ against the proxy,
\[
J_k(\bar z_k)\;\le\;\widehat J_k(\widehat z_k)+C\delta_k
\;\le\;\widehat J_k(\widetilde z_k)+C\delta_k
\;\le\;J_k\bigl(\text{completion of }y^\tau\bigr)+C\delta_k,
\]
that is, $\psi_{\eta_k}(\bar z_k^{\mathrm{off}})
\le\psi_{\eta_k}(y^\tau)+C\delta_k$ with
$\bar z_k^{\mathrm{off}}\in Y_{NK\kappa_k}(\theta)$.

\emph{Step 4: exact selection and strong convexity.} By
Lemma~\ref{lem:projsel} ($\eta_k<\bar\eta$,
$NK\kappa_k\ge\rho^\tau_{\min}$), $y^\tau$ minimizes
$\psi_{\eta_k}$ over $Y_{NK\kappa_k}(\theta)$. The function
$\psi_{\eta_k}$ is $2\eta_k$-strongly convex, so first-order
optimality over the convex set gives
$\psi_{\eta_k}(y)-\psi_{\eta_k}(y^\tau)
\ge\eta_k\|y-y^\tau\|^2$ for every $y\in Y_{NK\kappa_k}(\theta)$,
whence $\|\bar z_k^{\mathrm{off}}-y^\tau\|
\le\sqrt{C\delta_k/\eta_k}$; adding the repair distance of Step 2
proves the display. Almost-sure convergence follows since
$\delta_k/\eta_k=k^{-1/3}\log(k+1)\to0$,
$\Prob\bigl((A_k\cap B_k)^c\bigr)$ is summable by
\eqref{eq:badprob} and the proof of Lemma~\ref{lem:plugin}, and
$k_1$ is deterministic.
\end{proof}

\paragraph{Proof of Theorem~\ref{thm:frontier-achieve}.}
Identification and repair are those of Lemma~\ref{lem:ident},
unchanged. \emph{Tracking:} the proof of Lemma~\ref{lem:track}
applies with $z^\circ$ replaced by $y^\tau$. On
$A_k\cap B_k\cap\{k\ge k_1\}$ the off-stable marginals are bounded
by the constant of Step 1 and converge to $y^\tau$
(Lemma~\ref{lem:regplugin}), giving the dominated Ces\`aro limit.
On the complement, epoch $k$ contributes at most
$(b_k-b_{k-1})\widehat z_{k,p}+O(NK)
\le(b_k-b_{k-1})NK\kappa_k+O(NK)=O(k^2)$ to the pair count
(the cover-feasibility Lemma~\ref{lem:capfeas} and
$b_k-b_{k-1}=O(k)$), and
\[
\sum_kk^2\Bigl(Ce^{-c_\theta k}
+CNK\,k^{2/3}e^{-c_1k^{1/3}/2}\Bigr)<\infty,
\]
an $O(1)$ total in expectation.
Hence $\E[N^{\mathrm{exp}}_{i,a}(T)]=y^\tau_{i,a}\log T+o(\log T)$
for every off-stable pair. \emph{Certification:} the analysis of
Lemma~\ref{lem:certify} applies verbatim; its supply argument uses
only $\widehat z_k\ge\widehat q_k$ on the eligible pairs, which
holds because the selected schedule is feasible for the empirical
quotas, exactly as for $\pi_s$ in
Appendix~\ref{app:c4-frontier33}. \emph{Accounting:} as in
Lemma~\ref{lem:account},
$\E[N_{i,a}(T)]=y^\tau_{i,a}\log T+o(\log T)$, so
$R_i(T)/\log T\to\sum_a\Delta_{i,a}y^\tau_{i,a}$ and, since
$y^\tau\in F_w(\theta)$ is cost-optimal,
$R_w(T)/\log T\to c_{\mathrm{off}}^\top y^\tau=C_w(\theta)$.
\emph{Uniform goodness:} Lemma~\ref{lem:ug}
(Appendix~\ref{app:c4-strict}) applies to $\pi_{w,\Psi}$ directly.
\hfill$\qed$

\begin{corollary}[Pointwise attainability of the boundary]
\label{cor:frontier-general}
Fix $\theta\in\Thsep$. For every $w>0$ and every
$y\in F_w(\theta)$, the constant rule $\tau^{\mathrm{off}}=y$,
extended by zeros on stable coordinates, gives
$R(T)/\log T\to(\sum_a\Delta_{i,a}y_{i,a})_i$. Consequently every
Pareto-minimal point $r$ of $\UGL(\theta)$ is attained by some policy
$\pi_{w,\tau}$ with $\tau=\tau(\theta,r)$, uniformly good on all of
$\Thstrict$.
\end{corollary}

\paragraph{Proof of Corollary~\ref{cor:frontier-general}.}
A constant rule is Lipschitz, so Theorem~\ref{thm:frontier-achieve}
covers the fixed-target policy $\pi_{w,\tau}$. For
$y\in F_w(\theta)$, the choice $\tau^{\mathrm{off}}=y$ projects to
itself, so the theorem yields
$R(T)/\log T\to(\sum_a\Delta_{i,a}y_{i,a})_i$. Now let $r$ be any
Pareto-minimal point of $\UGL(\theta)$. By
Proposition~\ref{prop:proper}(ii) there is $w>0$ with
$\langle w,r\rangle=C_w(\theta)$. Writing $r=G(\theta)x+v$ with
$x\in\Xset(\theta)$ and $v\ge0$, Pareto minimality forces $v=0$,
since $G(\theta)x\in\UGL(\theta)$ lies componentwise below $r$.
Support gives $\langle w,G(\theta)x\rangle=C_w(\theta)$, i.e., $x$
is cost-optimal; its pair marginals with any stable completion are
feasible for \eqref{eq:compactlp} at the same cost (proof of
Theorem~\ref{thm:compact}), so
$y:=(z_{i,a}(x))_{\mathrm{off}}\in F_w(\theta)$ and
$(G(\theta)x)_i=\sum_a\Delta_{i,a}y_{i,a}$. The policy
$\pi_{w,\tau}$ with $\tau^{\mathrm{off}}=y$ attains $r$.
\hfill$\qed$

\begin{remark}[The $s$-split as a target rule]\label{rem:usplit-rule}
On the structure of Example~\ref{ex:33} let $\Psi_s$ map the
empirical quotas $(\widehat q_{12},\widehat q_{13},\widehat q_{23})$
to the pair marginals of the allocation $\widehat x(s)$ of the
variant $\pi_s$ in Algorithm~\ref{alg:policy}, and every other
structure to $\tau=0$; $\Psi_s$ is Lipschitz, the split being built
from sums and a minimum of the quotas. At the knife-edge weight
$w^\ast=(1,0.99,1)$ every family point of
Proposition~\ref{prop:frontier33} is cost-optimal (all segment
points have the same $w^\ast$-cost), so $F_{w^\ast}(\mu^{(3)})$
contains the rule's value at the true quotas, the family point with
$h=200s$, $y=2(1-s)$. Theorem~\ref{thm:frontier-achieve} therefore
gives $R(T)/\log T\to r(202s-2)$ for $\pi_{w^\ast,\Psi_s}$: the
conclusion of Proposition~\ref{prop:frontier-attain}, whose policy
$\pi_s$ reaches it without solving any program.
\end{remark}

\subsection{Proof of Theorem~\ref{thm:closure} and
Corollary~\ref{cor:rigid}}\label{app:c4-closure}

\paragraph{The uncosted policy $\pi_\tau$.}
Fix $\tau\in\R_+^{N\times K}$. The policy $\pi_\tau$ is
$\pi_{w,\tau}$ with the cost term of \eqref{eq:regprog} removed:
step 2 of epoch $k$ solves
\begin{equation}\label{eq:trackprog}
\min\bigl\{\Phi_\tau(z^{\mathrm{off}}):\ (z,\rho)\in
\widehat{\mathcal P}_k,\ \rho\le NK\kappa_k\bigr\},
\qquad \Phi_\tau(y):=\|y-\tau^{\mathrm{off}}\|^2,
\end{equation}
off-stable coordinates taken with respect to $\widehat m_k$, and
takes the minimal-$\rho$ completion of the unique optimal
off-stable block (the feasible set is a polytope and $\Phi_\tau$ is
strictly convex). No weight, surrogate cost or regularization
parameter enters. The cover of Lemma~\ref{lem:capfeas} is feasible
for the constraints of \eqref{eq:trackprog} whatever the objective,
so that lemma applies verbatim: the program is feasible on every
sample path, every solution has $\widehat\rho_k\le NK\kappa_k$, and
the exploration phase of epoch $k$ plays at most
$(b_k-b_{k-1})NK\kappa_k+NK$ rounds besides the replay.

\emph{(i) Convergence.} Fix $\theta\in\Thsep$ and write
$Y_\infty(\theta)=\operatorname{proj}_{\mathrm{off}}\mathcal P(\theta)$
as in Appendix~\ref{app:c4-c4b}, a nonempty closed polyhedron
(Remark~\ref{rem:feasible}). Let $y^\tau_\infty$ be the Euclidean
projection of $\tau^{\mathrm{off}}$ onto $Y_\infty(\theta)$ and
$\rho^\tau_\infty:=\min\{\rho:(z,\rho)\in\mathcal P(\theta),\
z^{\mathrm{off}}=y^\tau_\infty\}$, attained as in
Lemma~\ref{lem:projsel}. Since $Y_\infty(\theta)$ is convex, the
projection inequality reads
\begin{equation}\label{eq:projineq}
\|y-y^\tau_\infty\|^2\;\le\;\Phi_\tau(y)-\Phi_\tau(y^\tau_\infty)
\qquad\text{for all }y\in Y_\infty(\theta).
\end{equation}

\begin{lemma}[Plug-in convergence of the uncosted program]
\label{lem:trackplugin}
Let $\theta\in\Thsep$ and $\delta_k:=k^{-1/3}$. There are a
deterministic index $k_2=k_2(\theta,\tau)$ and a constant $C$ such
that on $A_k\cap B_k$ with $k\ge k_2$,
\[
\bigl\|\widehat z_k^{\mathrm{off}}-y^\tau_\infty\bigr\|\le C\sqrt{\delta_k},
\qquad\text{and}\qquad
\bigl\|\widehat z_k^{\mathrm{off}}-\tau^{\mathrm{off}}\bigr\|\le C\delta_k
\ \text{ if }\tau^{\mathrm{off}}\in Y_\infty(\theta).
\]
In particular $\widehat z_k^{\mathrm{off}}\to y^\tau_\infty$ almost
surely.
\end{lemma}

\begin{proof}
Choose $k_2$ so that for $k\ge k_2$: on $A_k$,
$\widehat m_k=\mstar$ and $\widehat\Eset_k=\Eset$ (proof of
Lemma~\ref{lem:ident}); $NK\kappa_k\ge\rho^\tau_\infty+1$; and on
$A_k\cap B_k$ the quota caps and noise floors are inactive with
$\|\widehat q_k-q\|_\infty\le C\delta_k$ (proof of
Lemma~\ref{lem:plugin}).

\emph{Proxy.} As in Step~1 of the proof of
Lemma~\ref{lem:regplugin}, let $\widetilde z_k$ be a minimal-$\rho$
completion of $y^\tau_\infty$ topped up to the empirical quotas
along fixed completion matchings. Because $y^\tau_\infty\ge q$ on
$\Eset$ (it is the off-stable block of a point of
$\mathcal P(\theta)$), the mass added for $(i,a)\in\Eset$ is at most
$(\widehat q_{k,i,a}-q_{i,a})_+\le C\delta_k$; hence $\widetilde z_k$
is feasible for \eqref{eq:trackprog}, its $\rho$ being at most
$\rho^\tau_\infty+CNK\delta_k\le NK\kappa_k$, and
$\|\widetilde z_k^{\mathrm{off}}-y^\tau_\infty\|\le C\delta_k$.

\emph{Feasible target.} If $\tau^{\mathrm{off}}\in Y_\infty(\theta)$
then $y^\tau_\infty=\tau^{\mathrm{off}}$, and optimality of
$\widehat z_k$ against the proxy gives
$\|\widehat z_k^{\mathrm{off}}-\tau^{\mathrm{off}}\|
=\Phi_\tau(\widehat z_k^{\mathrm{off}})^{1/2}
\le\Phi_\tau(\widetilde z_k^{\mathrm{off}})^{1/2}\le C\delta_k$.

\emph{General target.} Optimality against the proxy gives
$\Phi_\tau(\widehat z_k^{\mathrm{off}})
\le\Phi_\tau(\widetilde z_k^{\mathrm{off}})
\le\Phi_\tau(y^\tau_\infty)+C\delta_k$, so $\widehat z_k^{\mathrm{off}}$
lies in a fixed ball on which $\Phi_\tau$ is Lipschitz. The Hoffman
repair of Step~2 of the proof of Lemma~\ref{lem:regplugin} applies
unchanged (same constraint matrix, right-hand sides within
$C\delta_k$) and yields $\bar z_k$ in the true capped feasible set
with $\|\bar z_k-\widehat z_k\|\le C\delta_k$, whence
$\Phi_\tau(\bar z_k^{\mathrm{off}})\le\Phi_\tau(y^\tau_\infty)
+C\delta_k$ with $\bar z_k^{\mathrm{off}}\in Y_{NK\kappa_k}(\theta)
\subseteq Y_\infty(\theta)$. By \eqref{eq:projineq},
$\|\bar z_k^{\mathrm{off}}-y^\tau_\infty\|^2\le C\delta_k$, and
adding the repair distance proves the display. Almost-sure
convergence follows because $\Prob((A_k\cap B_k)^c)$ is summable
(\eqref{eq:badprob} and the proof of Lemma~\ref{lem:plugin}) and
$k_2$ is deterministic.
\end{proof}

The remainder of part~(i) is the proof of
Theorem~\ref{thm:frontier-achieve} with $y^\tau$ replaced by
$y^\tau_\infty$: tracking by Lemma~\ref{lem:track}, the off-stable
marginals being bounded on the good events by the fixed ball above
and convergent to $y^\tau_\infty$, while bad epochs cost $O(k^2)$
rounds with summable probability by Lemma~\ref{lem:capfeas};
certification by Lemma~\ref{lem:certify}, whose supply argument
uses only the quota constraints of \eqref{eq:trackprog}; accounting
by Lemma~\ref{lem:account}; and uniform goodness on $\Thstrict$ by
Lemma~\ref{lem:ug}, which lists $\pi_\tau$. Hence
$\E_\theta[N_{i,a}(T)]/\log T\to y^\tau_{\infty,i,a}$ for every
off-stable pair and
$R_i(T;\theta)/\log T\to\sum_a\Delta_{i,a}y^\tau_{\infty,i,a}$.

\emph{(ii) The attainable set.} If $r\in\mathcal A(\theta)$ is
attained by a policy $\pi$, then $R(T;\theta)/\log T\to r$ along the
full sequence and Proposition~\ref{prop:region} gives
$r\in G(\theta)\Xset(\theta)$. Conversely, let
$x\in\Xset(\theta)$, let $z:=z(x)$ be its pair marginals
($z_{i,a}=\sum_{m\neq\mstar:\,m(i)=a}x_m$) and
$\rho:=\sum_{m\neq\mstar}x_m$. Then $(z,\rho)\in\mathcal P(\theta)$:
every matching gives each player exactly one arm and each arm at
most one player, so the rows of $z$ sum to $\rho$ and its columns
to at most $\rho$, and $z\ge q$ on $\Eset$ is
Theorem~\ref{thm:reduction}. Taking $\tau:=z$ in part~(i),
$\tau^{\mathrm{off}}\in Y_\infty(\theta)$ is tracked exactly and
\[
\frac{R_i(T;\theta)}{\log T}\;\longrightarrow\;
\sum_{a\neq\mstar(i)}\Delta_{i,a}\,z_{i,a}
=\sum_{m\neq\mstar}x_m\bigl(\mu_{i,\mstar(i)}-\mu_{i,m(i)}\bigr)
=\bigl(G(\theta)x\bigr)_i ,
\]
regrouping the sum over $m$ by the arm $m(i)$ and using
$\Delta_{i,\mstar(i)}=0$. Since $\pi_\tau$ plays complete matchings
and is uniformly good on $\Thstrict$, $G(\theta)x\in\mathcal A(\theta)$.
Thus $\mathcal A(\theta)=G(\theta)\Xset(\theta)$, the image of the
polyhedron $\Xset(\theta)$ under a linear map and hence a
polyhedron (Lemma~\ref{lem:polyhedral}), and
$\UGL(\theta)=\mathcal A(\theta)+\R^N_+$ by \eqref{eq:UGL}.

\emph{Pareto-minimal sets.} For any $S\subseteq\R^N$,
$\mathrm{Pmin}(S+\R^N_+)=\mathrm{Pmin}(S)$: if
$u=s+v\in\mathrm{Pmin}(S+\R^N_+)$ with $s\in S$ and $v\ge0$, then
$s\le u$ forces $s=u\in S$, and any point of $S$ strictly below $u$
would lie in $S+\R^N_+$; conversely, if $u\in\mathrm{Pmin}(S)$ and
$s+v\le u$ with $s\in S$, $v\ge0$, then $s\le u$ forces $s=u$ and
then $v=0$. With $S=\mathcal A(\theta)$ this is
$\mathrm{Pmin}(\mathcal A(\theta))=\mathrm{Pmin}(\UGL(\theta))$,
the Graves--Lai necessary lower boundary; by
Proposition~\ref{prop:proper}(i) every attainable vector moreover
lies weakly above an attainable Pareto-minimal point.
\hfill$\qed$

\paragraph{Proof of Corollary~\ref{cor:rigid}.}
The inequalities defining $\Xset(\theta)$ in
Theorem~\ref{thm:reduction} have nonnegative coefficients, so the
recession cone of $\Xset(\theta)$ is the whole orthant
$\R_+^{\Mset\setminus\{\mstar\}}$, and the recession cone of its
linear image $\mathcal A(\theta)=G(\theta)\Xset(\theta)$ is
$G(\theta)\R_+^{\Mset\setminus\{\mstar\}}
=\mathrm{cone}\{g(m;\theta):m\neq\mstar\}$. A closed convex set
$S$ satisfies $S+\R^N_+=S$ if and only if
$\R^N_+\subseteq\mathrm{rec}(S)$; applied to the polyhedron
$\mathcal A(\theta)$ this is the stated equivalence. If $K>N$,
some arm $a_0$ is unmatched by $\mstar$, and the matching that
sends player $i$ to $a_0$ and every other player to its stable arm
has gap vector $\Delta_{i,a_0}e_i$ with $\Delta_{i,a_0}>0$
(separation and row strictness), so every $e_i$ lies in the cone.
For $N=1$, the standing assumption $K\ge2$ puts us in this case.
If $K=N$, then $N\ge2$, and every complete matching permutes the
stable arms. Any $m\neq\mstar$ changes at least two assignments,
so by separation and row strictness its gap vector has at least
two strictly positive coordinates and all remaining coordinates
nonnegative. Every nonzero nonnegative combination of these
vectors still has at least two strictly positive coordinates;
hence no $e_i$ lies in the cone. This proves the equivalence with
$K>N$.
For $\mu^{(3)}$, the five gap vectors listed in
Appendix~\ref{app:c2-frontier} show directly that no $e_i$ lies in
the cone: a nonnegative combination with vanishing second
coordinate can use only $m_{\mathrm E}$, whose third coordinate is
$1$, and one with vanishing first coordinate can use only
$m_{\mathrm C}$, whose second and third coordinates are both
positive. For the explicit witness,
$(22,20,202)=G(\mu^{(3)})x$ with $x_{m_{\mathrm B}}=200$,
$x_{m_{\mathrm E}}=2$ (the parallel end, $s=1$ in
Proposition~\ref{prop:frontier-attain}). If
$(22,20,r_3)=G(\mu^{(3)})x$ for some $x\in\Xset(\mu^{(3)})$, then
$r_1=0.1(x_{\mathrm A}+x_{\mathrm B})+(x_{\mathrm D}+x_{\mathrm E})=22$
together with $x_{\mathrm A}+x_{\mathrm B}\ge200$ and
$x_{\mathrm D}+x_{\mathrm E}\ge2$ forces both quotas tight, and then
$r_2=x_{\mathrm A}+0.1x_{\mathrm B}+0.1x_{\mathrm C}+x_{\mathrm D}
=200-0.9x_{\mathrm B}+0.1x_{\mathrm C}+x_{\mathrm D}
\ge220-x_{\mathrm B}\ge20$, using $x_{\mathrm C}\ge200-x_{\mathrm B}$
and $x_{\mathrm B}\le200$, with equality only if
$x_{\mathrm B}=200$ and $x_{\mathrm C}=x_{\mathrm D}=0$; then
$x_{\mathrm A}=0$, $x_{\mathrm E}=2$ and
$r_3=x_{\mathrm B}+0.01(x_{\mathrm C}+x_{\mathrm D})+x_{\mathrm E}=202$.
Hence $(22,20,202+\varepsilon)\notin\mathcal A(\mu^{(3)})$ for every
$\varepsilon>0$, although it lies in $\UGL(\mu^{(3)})$.
\hfill$\qed$

\subsection{Uniform goodness on
\texorpdfstring{$\Thstrict$}{Theta-strict}}\label{app:c4-strict}

This subsection proves the uniform-goodness claims of Theorems~\ref{thm:achieve}, \ref{thm:frontier-achieve} and~\ref{thm:closure}(i) and Proposition~\ref{prop:frontier-attain} in a single statement. The
standing conventions of the appendix are suspended: the instance is
an arbitrary $\lambda\in\Thstrict$, neither top-choice separation
nor Assumption~\ref{ass:unique} is assumed, and constants may
depend on $(\lambda,w,s,\Psi,\tau,\alpha,\xi,N,K)$.

Write $\mstar=\mstar(\lambda)$ for the serial-dictatorship matching
of $\lambda$, $F_i(\lambda)$ for the free set at level $i$ along
$\mstar$, and recall from \S\ref{sec:information} that the eligible
set $\Eset(\lambda)$ and the gaps
$\Delta_{i,a}(\lambda)=\lambda_{i,\mstar(i)}-\lambda_{i,a}$,
$(i,a)\in\Eset(\lambda)$, are defined for every strict instance,
with $\Delta_{i,a}(\lambda)>0$ because $\mstar(i)$ is by row
strictness the unique maximizer of row $i$ over $F_i(\lambda)$;
these \emph{free-set} gaps are positive even when $\mstar(i)$ is
not the global maximizer of its row. Let $\delta_\lambda>0$ be half
the smallest absolute difference between two entries in the same
row of $\lambda$, and
$\delta''_\lambda:=\min\bigl(\delta_\lambda/2,\
\min_{(i,a)\in\Eset(\lambda)}\Delta_{i,a}(\lambda)/4\bigr)$. Since
$K\ge N$ and $K\ge2$, level~$1$ has at least two free arms and hence
a free non-stable one, so $\Eset(\lambda)\neq\emptyset$ for every
$\lambda\in\Thstrict$ and both minima are over nonempty sets.

\begin{lemma}[Uniform goodness on $\Thstrict$]\label{lem:ug}
Fix $\alpha>\xi>3$ and $K\ge2$, $w\in\R^N_{++}$, $s\in[0,1]$, a
target rule $\Psi$ and $\tau\in\R_+^{N\times K}$, and let $\pi$ be
any of $\pi_w$, $\pi_s$, $\pi_{w,\Psi}$, $\pi_\tau$. For every $\lambda\in\Thstrict$:
\begin{enumerate}
  \item[(i)] prefix, repair, exploration and replay rounds total
  $O\bigl((\log T)^{3/2}\bigr)$ on every sample path;
  \item[(ii)] the expected number of wrong certified exploitation
  rounds is $O(1)$;
  \item[(iii)] the expected number of estimation rounds is
  $o(\log T)$.
\end{enumerate}
Consequently
$S(T;\lambda)=O\bigl((\log T)^{3/2}\bigr)=o(T^{\alpha'})$ for every
$\alpha'>0$, and Lemma~\ref{lem:nocancel} (direction
(ii)$\Rightarrow$(iii)) yields
$|R_i(T;\lambda)|=o(T^{\alpha'})$ for every player: each of the
four policies is componentwise uniformly good on $\Thstrict$.
\end{lemma}

\begin{proof}
\emph{Two properties shared by the four step-2 variants.} First,
the schedule of epoch $k$ is feasible for the empirical quotas,
$\widehat z_{k}\ge\widehat q_k$ on $\widehat\Eset_k$: for $\pi_w$
the quotas are constraints of the solved program; for $\pi_s$ the
split meets them with equality by construction, as does the
dedicated fallback (Appendix~\ref{app:c4-frontier33}); for $\pi_{w,\Psi}$ and $\pi_\tau$ they are constraints of the capped
programs \eqref{eq:regprog} and \eqref{eq:trackprog}. Second, the
scheduled mass is bounded pathwise by $NK\kappa_k$: for $\pi_w$
this is the capped-mass Lemma~\ref{lem:mass}, for $\pi_{w,\Psi}$
and $\pi_\tau$ the cover-feasibility Lemma~\ref{lem:capfeas}
(whose cover is feasible whatever the objective), and for $\pi_s$ both the split and the
fallback assign total mass at most
$\sum_{(i,a)\in\widehat\Eset_k}\widehat q_{i,a}\le NK\kappa_k$. All
phases other than step 2, in particular the certificate and the
estimation-round sampling rule of phase 4, are identical across the
variants, so one argument covers all four.

\emph{(i) Pathwise rounds.} The prefix is a deterministic $O(1)$.
Repair only ever tops each pair count up from $\sqrt{b_{k-1}}$ to
$\sqrt{b_k}$, hence totals $O(NK\sqrt{\log T})$ rounds. In epoch
$k$, exploration plays at most $(b_k-b_{k-1})NK\kappa_k+NK$ rounds
by the mass bounds above, and the replay
$\lceil(b_k-b_{k-1})k\rceil$ more; with $b_k-b_{k-1}=O(k)$,
$\kappa_k=k$ and $O(\sqrt{\log T})$ epochs up to $T$, the total is
$O\bigl(NK(\log T)^{3/2}\bigr)$, with no condition on the data or
the instance.

\emph{(ii) Wrong certified exploitation.} Part (a) of the proof of
Lemma~\ref{lem:certify} transfers verbatim. If the certified
candidate $\widetilde m_t$ deviates from $\mstar$ earliest at level
$i$, then $\widetilde m_t(j)=\mstar(j)$ for $j<i$, so the level-$i$
free set is $F_i(\lambda)$ and contains both $\mstar(i)$ and
$\widetilde a:=\widetilde m_t(i)$; since $\mstar(i)$ is the unique
maximizer of row $i$ over $F_i(\lambda)$, the true gap opposes the
certificate, and a certified pass forces some entry to deviate by
its $c(t)$-radius at its current count. The peeling bound of
Lemma~\ref{lem:peel} is instance-independent, and
$\sum_tC/(t(\log t)^{\xi-2})=O(1)$ as before.

\emph{(iii) Estimation rounds.} This is the one place where the
instance enters through the policy's decisions rather than through
concentration alone, so we give the argument in full rather than by
transfer. Call an entry $(j,b)$ \emph{deviated at count $n$} if
$|\widehat\mu^{(n)}_{j,b}-\lambda_{j,b}|\ge\delta''_\lambda$ and
\emph{clean} otherwise; ``deviated at round $t$'' refers to its
current count $N_{j,b}(t)$. Two facts about the phase order of
Algorithm~\ref{alg:policy} are used throughout.
\begin{enumerate}
  \item[(O1)] Within epoch $k$ the phases run repair, solve,
  explore (schedule, then replay of $\widehat m_k$), and
  exploitation/estimation, in this order. Repair tops every pair
  count up to $\sqrt{b_k}\ge c_1k$, and counts never decrease, so
  \emph{every} estimate consulted at the solve step of epoch $k$,
  and at every estimation round of epoch $k$, has count at least
  $c_1k$.
  \item[(O2)] Every estimation round of epoch $k$ takes place after
  the exploration and replay phases of epochs $k_0,\dots,k$ have
  been executed in full. For $a\neq\widehat m_j(i)$, the exploration
  phase of epoch $j$ plays pair $(i,a)$ at least
  $(b_j-b_{j-1})\widehat z_{j,(i,a)}$ times, because increment
  execution rounds the count of every matching up; discarding the
  empirical stable component does not change these marginals. The replay
  plays $\widehat m_j$ exactly $\lceil(b_j-b_{j-1})\,j\rceil$ times.
\end{enumerate}

\emph{Claim (clean ordering).} Suppose that at some moment no entry
is deviated. Then $\mstar(\widehat\mu)=\mstar$, the empirical free
set at every level $i$ equals $F_i(\lambda)$, the empirical eligible
set equals $\Eset(\lambda)$, and for every $(i,a)\in\Eset(\lambda)$
\[
\bigl|(\widehat\mu_{i,\mstar(i)}-\widehat\mu_{i,a})
-\Delta_{i,a}(\lambda)\bigr|<2\delta''_\lambda
\le\Delta_{i,a}(\lambda)/2 .
\]
\emph{Proof of the claim.} Two entries in the same row of $\lambda$
differ by at least $2\delta_\lambda$, while clean estimates err by
less than $\delta''_\lambda\le\delta_\lambda/2$ each, so every
within-row empirical comparison has the sign of the true one.
Serial dictatorship is computed level by level: at level~$1$ the
free set is $[K]$ for both $\lambda$ and $\widehat\mu$, and the
empirical maximizer is the true one; if the first $i-1$ assignments
agree, the level-$i$ free sets agree and so do the maximizers over
them. The eligible set consists of the free off-stable pairs
$\{(i,a):a\in F_i\setminus\{\mstar(i)\}\}$ and is therefore
determined by the matching and the free sets. The displayed bound
is the triangle inequality together with
$\delta''_\lambda\le\Delta_{i,a}(\lambda)/4$. \hfill$\square$

Applied at the solve step of epoch $k$, the claim gives
$\widehat m_k=\mstar$ and $\widehat\Eset_k=\Eset(\lambda)$; applied
at an exploitation/estimation round $t$, it gives
$\widetilde m_t=\mstar$. Neither use involves top-choice separation
or Assumption~\ref{ass:unique}.

\emph{Classification.} Classify every estimation round by the two
entries of its sampled comparison: \emph{type B} if one of them is
deviated at that round, \emph{type C} otherwise.

\emph{Type B.} The bookkeeping of part (b) of the proof of
Lemma~\ref{lem:certify} is purely combinatorial: a pair
participates in at most $2K$ comparisons, and while it is deviated
at count $n$ each comparison produces at most $n+1$ type-B rounds
that sample the other pair, plus one that samples the pair itself.
Charging each type-B round to its deviated pair at its current
count and using $\Prob(\text{deviated at count }n)\le
2e^{-n\delta''^2_\lambda/2}$,
\[
\E[\text{type-B rounds}]\le
NK\sum_{n\ge1}\bigl(2K(n+1)+1\bigr)\,2e^{-n\delta''^2_\lambda/2}
=O(1).
\]

\emph{Type C: the local cap.} Let the sampled comparison be at level
$i$ between the candidate's entry $\widetilde a=\widetilde m_t(i)$
and a free arm $a\neq\widetilde a$ of the candidate's level-$i$ free
set. Whatever that free set is, $\widetilde a$ is the empirical
argmax over it, so $\widehat\mu_{i,\widetilde a}>\widehat\mu_{i,a}$;
both entries are clean, and the two true means differ by at least
$2\delta_\lambda$, so they are ordered as the empirical ones and
$\widehat\mu_{i,\widetilde a}-\widehat\mu_{i,a}\ge
2\delta_\lambda-2\delta''_\lambda\ge\delta_\lambda$. The
certificate for this comparison therefore passes as soon as both
counts reach $8c(t)/\delta_\lambda^2$, and since the round samples
the smaller count, every type-C round increments a count that is
below $8c(t)/\delta_\lambda^2$. As counts never decrease and $c$ is
increasing, the type-C rounds of epoch $k$ number at most
$NK(8c(T_k)/\delta_\lambda^2+1)\le CNKk^2$ and the type-C rounds
up to horizon $T$ at most $NK(8c(T)/\delta_\lambda^2+1)=O(\log T)$,
both on every sample path.

\emph{Shadow rounds} (type C, but some entry elsewhere is
deviated). By (O1) every count at such a round of epoch $k$ is at
least $c_1k$, so the round requires an entry deviated at some count
$\ge c_1k$; by \eqref{eq:maximal} with $\varepsilon=\delta''_\lambda$,
$n_0=c_1k$ and a union over $NK$ pairs this has probability at most
$Ce^{-c_1k\delta''^2_\lambda/2}$. With the per-epoch cap,
$\E[\text{shadow rounds}]\le\sum_kCNKk^2\cdot
Ce^{-c_1k\delta''^2_\lambda/2}=O(1)$.

\emph{Equilibrium rounds} (type C and no entry deviated). By the
claim, $\widetilde m_t=\mstar$ and the candidate's free sets are
$F_i(\lambda)$, so the sampled comparison is $(\mstar(i),a)$ with
$(i,a)\in\Eset(\lambda)$. Fix $\gamma\in(0,1/2)$, write
$\Delta_{\min}(\lambda):=\min_{\Eset(\lambda)}\Delta_{i,a}(\lambda)$,
$q_{i,a}(\lambda):=2/\Delta_{i,a}(\lambda)^2$, and set
$\varepsilon_\gamma:=\min\bigl(\delta''_\lambda/2,\
\gamma\Delta_{\min}(\lambda)/8\bigr)$,
$k_\gamma(T):=\lceil C_2\log\log T\rceil$ and
\[
\mathcal G_{\lambda,T}:=\bigl\{\,\bigl|\widehat\mu^{(n)}_{j,b}-\lambda_{j,b}\bigr|
<\varepsilon_\gamma\ \ \forall(j,b),\ \forall n\ge c_1k_\gamma(T)\bigr\},
\]
an event about the reward stacks alone. By \eqref{eq:maximal} and a
union over $NK$ pairs,
$\Prob(\mathcal G_{\lambda,T}^c)\le
CNK(\log T)^{-c_1C_2\varepsilon_\gamma^2/2}\le(\log T)^{-2}$ once
$C_2=C_2(\lambda,\gamma)$ is large and $T\ge T_0(\lambda,\gamma)$.
On $\mathcal G_{\lambda,T}$: (G1) no entry is deviated at any count
$\ge c_1k_\gamma(T)$, because $\varepsilon_\gamma<\delta''_\lambda$;
(G2) for $(i,a)\in\Eset(\lambda)$ and both counts
$\ge c_1k_\gamma(T)$,
$|\widehat\Delta_{i,a}-\Delta_{i,a}(\lambda)|<2\varepsilon_\gamma
\le\gamma\Delta_{i,a}(\lambda)/4$, where
$\widehat\Delta_{i,a}:=\widehat\mu_{i,\mstar(i)}-\widehat\mu_{i,a}$.
Call epoch $k$ \emph{early} if $k\le2k_\gamma(T)$ and \emph{late}
otherwise.

\emph{Early epochs, pathwise.} Equilibrium rounds are type-C rounds,
so those in early epochs each increment a count below
$8c(T_{2k_\gamma})/\delta_\lambda^2$; they total at most
$NK(8c(T_{2k_\gamma})/\delta_\lambda^2+1)=O((\log\log T)^2)$ on
every path, as $c(T_k)=O(k^2)$.

\emph{Late epochs, on $\mathcal G_{\lambda,T}$.} We establish four
consequences, for $T\ge T_0(\lambda,\gamma)$.
\begin{enumerate}
  \item[(L1)] \emph{Every epoch $j>k_\gamma$ solves and replays
  $\mstar$.} By (O1) the counts at the solve step of epoch $j$ are
  at least $c_1j>c_1k_\gamma$, so by (G1) no entry is deviated
  there, and the claim gives $\widehat m_j=\mstar$,
  $\widehat\Eset_j=\Eset(\lambda)$. Hence the replay of every epoch
  $j>k_\gamma$ plays the true stable matching.
  \item[(L2)] \emph{The stable side is well sampled.} Let $t$ be an
  estimation round of a late epoch $k>2k_\gamma$. By (O2) the replays of
  epochs $k_\gamma+1,\dots,k$ have been executed, so by (L1) every
  stable pair $(i,\mstar(i))$ carries at least
  $\sum_{j=k_\gamma+1}^{k}(b_j-b_{j-1})\,j\ge c'(k^3-(k/2)^3)
  =\Theta(k\,b_k)$ samples, using $b_j-b_{j-1}\ge cj$ and
  $k_\gamma<k/2$. Its certification radius is therefore
  $\mathrm{rad}^{\mathrm{st}}_t\le\sqrt{2c(T_k)/\Theta(k\,b_k)}
  =O(k^{-1/2})\le\gamma\Delta_{\min}(\lambda)/4$, since late epochs
  have $k>2k_\gamma(T)\to\infty$.
  \item[(L3)] \emph{Quotas after $k_\gamma$ are nearly the true ones.} At the
  solve step of any epoch $j>k_\gamma$, by (L1) the surrogate gap of
  $(i,a)\in\Eset(\lambda)$ is $\widehat\Delta_{i,a}$, which by (G2)
  is at most $\Delta_{i,a}(\lambda)(1+\gamma/4)$; the noise floor is
  at most $2\sqrt{2\beta(T_{j-1})/(c_1j)}=O(\sqrt{\log j/j})$ by
  (O1), hence also below $\Delta_{i,a}(\lambda)(1+\gamma/4)$ for
  $j>k_\gamma(T)$ and $T$ large. Thus
  $\widehat g_{i,a}\le\Delta_{i,a}(\lambda)(1+\gamma/4)$, and with
  $\kappa_j=j>q_{i,a}(\lambda)$ for $j$ large,
  $\widehat q_{j,(i,a)}=\min(2/\widehat g_{i,a}^{\,2},\kappa_j)
  \ge q_{i,a}(\lambda)(1-\gamma/2)$. By the first shared property,
  $\widehat z_{j,(i,a)}\ge\widehat q_{j,(i,a)}\ge
  q_{i,a}(\lambda)(1-\gamma/2)$ for every epoch $j>k_\gamma$ and every
  $(i,a)\in\Eset(\lambda)$.
  \item[(L4)] \emph{Demand and supply at a late equilibrium round.}
  Let $t$ be an equilibrium round of a late epoch $k(t)$ whose
  sampled comparison is $(\mstar(i),a)$, $(i,a)\in\Eset(\lambda)$.
  The certificate failed there, so
  $\widehat\Delta_{i,a}(t)<\mathrm{rad}^{\mathrm{st}}_t
  +\sqrt{2c(t)/N_{i,a}(t)}$; by (G2) and (L2) this forces
  \[
  N_{i,a}(t)\;<\;\mathrm{Dem}^{\lambda}_{i,a}(t):=
  \frac{2c(t)}{\Delta_{i,a}(\lambda)^2(1-\gamma/2)^2}
  \;\le\;q_{i,a}(\lambda)\,c(t)\,(1+C\gamma).
  \]
  In particular $N_{i,a}(t)<\mathrm{Dem}^{\lambda}_{i,a}(t)=O(b_{k(t)})$ while the
  stable side carries $\Omega(k(t)\,b_{k(t)})$ samples by (L2), so
  for $T$ large the least-sampled pair of the comparison is $(i,a)$
  and the round samples $(i,a)$ itself. On the supply side, by (O2)
  and (L3),
  \[
  \begin{aligned}
  N_{i,a}(t)&\;\ge\;\mathrm{Sup}_{i,a}(t)+n_{i,a}(t),\\
  \mathrm{Sup}_{i,a}(t)&:=\sum_{k_\gamma<j\le k(t)}
  (b_j-b_{j-1})\,\widehat z_{j,(i,a)}
  \;\ge\;\bigl(b(t)-b_{k_\gamma}\bigr)q_{i,a}(\lambda)(1-\gamma/2),
  \end{aligned}
  \]
  where $n_{i,a}(t)$ counts the earlier late equilibrium rounds that
  sampled $(i,a)$, and we used $b_{k(t)}\ge b(t)$.
\end{enumerate}
\emph{Ratchet.} Let $N^{\mathrm{eq,late}}_{i,a}(T)$ be the number
of late equilibrium rounds sampling $(i,a)$ up to $T$ and let $t$ be
the last of them. By (L4),
$n_{i,a}(t)=N^{\mathrm{eq,late}}_{i,a}(T)-1$ and
$\mathrm{Sup}_{i,a}(t)+n_{i,a}(t)\le N_{i,a}(t)<\mathrm{Dem}^{\lambda}_{i,a}(t)$, so on
$\mathcal G_{\lambda,T}$, using $c(t)\le b(t)$ and
$b_{k_\gamma}=O\bigl((\log\log T)^2\bigr)$,
\begin{align*}
N^{\mathrm{eq,late}}_{i,a}(T)
&\;\le\;\mathrm{Dem}^{\lambda}_{i,a}(t)-\mathrm{Sup}_{i,a}(t)+1\\
&\;\le\;q_{i,a}(\lambda)\,b(t)\bigl[(1+C\gamma)-(1-\gamma/2)\bigr]
+q_{i,a}(\lambda)\,b_{k_\gamma}+1\\
&\;\le\;C'\gamma\log T+O\bigl((\log\log T)^2\bigr).
\end{align*}
\emph{Off $\mathcal G_{\lambda,T}$.} Equilibrium rounds are type-C rounds, at
most $O(\log T)$ on every path, so their expected contribution is
$O(\log T)\cdot(\log T)^{-2}=o(1)$. Summing over the $|\Eset(\lambda)|$
pairs and the early epochs,
$\limsup_T\E[N^{\mathrm{eq}}(T)]/\log T\le C'|\Eset(\lambda)|\gamma$
for every $\gamma\in(0,1/2)$, i.e., $\E[N^{\mathrm{eq}}(T)]=o(\log T)$.
Adding the three contributions,
$\E[N^{\mathrm{est}}(T)]=O(1)+O(1)+o(\log T)=o(\log T)$. Nothing in
(iii) used top-choice separation or Assumption~\ref{ass:unique}:
the argument needs only row strictness (through $\delta_\lambda$),
positivity of the free-set gaps, and the two shared properties.
\emph{Conclusion.} Every round belongs to the prefix or to one of
the four phases. Correctly certified exploitation plays $\mstar$
and is not a deviation; the replay deviates only on epochs with
$\widehat m_k\neq\mstar$ and is covered by the pathwise count (i)
either way. Since $S(T;\lambda)$ is by definition an expectation,
the pathwise bound (i) and the expected bounds (ii)--(iii) add to
$S(T;\lambda)=O\bigl((\log T)^{3/2}\bigr)$, and
Lemma~\ref{lem:nocancel} converts this deviation-count bound into
componentwise uniform goodness. No pathwise bound on the
estimation phase is claimed or needed.
\end{proof}

\FloatBarrier
\section{Additional experiments}\label{app:experiments}

All conventions of \S\ref{sec:experiments} apply: sample-path
pseudo-regret $\widetilde R$, normalization by $b(T)$, and the
practical variants of the policy unless stated otherwise.

The practical variants use the exploration and certification mechanisms of
\S\ref{sec:c4-policy} with schedules adapted to feasible horizons:
geometric epochs $T_k=200\cdot2^k$
instead of $e^{k^2}$, threshold exponents $\alpha=1.1$ and (when the
certification test is on) $\xi=1.05$ instead of $\alpha>\xi>3$,
deficit-based execution of the current target $b(t)\widehat z_k$
instead of increment execution, quota cap $2000$, and batched
exploitation. The thresholds $\alpha=1.1$ and $\xi=1.05$ lie outside
the theory: for $\xi\le3$ the failure bound $C/(t(\log t)^{\xi-2})$ in
the proof of Lemma~\ref{lem:certify} is not summable, and the plug-in
tracker runs without a certificate. The practical variants therefore
carry no guarantee; the theorem-style diagnostic below keeps
$\alpha>\xi>3$.

The following comparison separates the main practical suite from the
truncated diagnostic with theorem-style thresholds.
\begin{center}
\small
\setlength{\tabcolsep}{3pt}
\begin{tabular}{@{}>{\raggedright\arraybackslash}p{0.14\textwidth}
>{\raggedright\arraybackslash}p{0.26\textwidth}
>{\raggedright\arraybackslash}p{0.26\textwidth}
>{\raggedright\arraybackslash}p{0.26\textwidth}@{}}
\toprule
Feature & Theory: Algorithm~\ref{alg:policy} & Main practical suite & Theorem-style diagnostic \\
\midrule
Epochs & $T_k=\lceil e^{k^2}\rceil$ & $T_k=200\cdot2^k$, up to $10^7$ & $T_k=\lceil e^{k^2}\rceil$, $k=2,3,4$ \\
Thresholds & $\alpha>\xi>3$ & $\alpha=1.1$; $\xi=1.05$ when certified & $\alpha=3.3$, $\xi=3.1$ \\
Quota cap & $\kappa_k=k$ & $2000$ & $2000$ \\
Repair / solve & Repair to $\sqrt{b(T_k)}$, then solve & Solve at epoch start; repair to $\sqrt{b(t+1)}$ during execution & Repair to $\sqrt{b(T_k)}$, then solve \\
Execution / replay & Increment execution; stable replay & Deficit execution; no separate replay phase & Increment execution; stable replay \\
Certification & Refresh each round; first failing comparison & Batched; test at $c(T_k)$ when on & Batched; least-certified comparison \\
Canonical solve & Cost, mass with slack $k^{-1/4}$, lexicographic & Numerical cost/mass LPs & Numerical cost/mass LPs \\
\bottomrule
\end{tabular}
\end{center}
Both diagnostics start with six cyclic-matching rounds. Their numerical
cost slack is $10^{-9}\max(1,|\widehat v|)$, rather than $k^{-1/4}$.
The regularized face-selection experiment below uses its own growing
quota cap and regularization schedules.

\paragraph{Mechanism ablation.}
Table~\ref{tab:ablation} reports the seven configurations discussed
in \S\ref{sec:experiments}.

\begin{table}[H]
\centering\small
\setlength{\tabcolsep}{4pt}
\caption{Mechanism ablation at $T=10^7$, $50$ seeds, $w=(1,1,1)$,
$C_w=244$: persistent quota exploration (E), certification (C),
noise floor (F), candidate (dyn/frozen/lock-in); median [IQR] of
$w^\top\widetilde R(T)/b(T)$; median exploration, wrong-exploitation
and estimation rounds; fraction of seeds with an incorrect fixed
matching (defined below). Count medians
are rounded to the nearest integer, with halves rounded up.}
\label{tab:ablation}
\begin{tabular}{lcccc r@{\;}l r r r r}
\toprule
variant & E & C & F & cand.\ &
\multicolumn{2}{c}{$w^\top\!\widetilde R/b$ med [IQR]} &
expl.\ & wrong & est.\ & lock-wr.\\
\midrule
oracle       & \checkmark & --         & --         & --   &
$244.4$ & $[244.4,244.4]$ & $3{,}872$ & $0$ & $0$ & --\\
plug-in      & \checkmark & --         & \checkmark & dyn  &
$299.0$ & $[230.5,382.3]$ & $6{,}122$ & $970$ & $0$ & --\\
certified    & \checkmark & \checkmark & \checkmark & dyn  &
$752.2$ & $[650.3,872.4]$ & $4{,}073$ & $0$ & $27{,}589$ & --\\
frozen cand. & \checkmark & \checkmark & \checkmark & frz  &
$28{,}117$ & $[874,273{,}425]$ & $4{,}473$ & $0$ & $\approx10^7$ & $62\%$\\
no floor     & \checkmark & --         & --         & dyn  &
$311.3$ & $[242.7,382.8]$ & $8{,}772$ & $333$ & $0$ & --\\
exploit-only & stops      & --         & \checkmark & lock &
$72.1$ & $[1.8,57{,}386]$ & $4$ & $1{,}366$ & $0$ & $44\%$\\
cert-only    & --         & \checkmark & \checkmark & dyn  &
$659.7$ & $[556.2,871.2]$ & $0$ & $0$ & $34{,}691$ & --\\
\bottomrule
\end{tabular}
\end{table}

Here \texttt{lock-wr.} is the fraction of seeds with an incorrect
fixed matching. For \emph{frozen cand.}, the certification candidate is
frozen at epoch $1$; for \emph{exploit-only}, the policy commits when
the empirical stable matching repeats at two consecutive epochs
($k\ge2$), then stops repair, quota exploration and further solves.
The flag records whether that fixed matching differs from $\mstar$;
it is not a final-window error rate. A dash means that the variant
makes no fixed commitment. The \texttt{wrong} column counts actual
wrong exploitation rounds, excluding estimation rounds. Thus a wrong
frozen candidate can remain uncertified while accumulating estimation
cost and no wrong exploitation. For the three variants with
certification (certified, frozen cand.\ and cert-only) the
\texttt{wrong} count is zero in every one of the $50$ seeds, not only
in the median.

The exploit-only median lies below $C_w$ because the lower bound
constrains the expected regret of uniformly good policies only:
exploit-only is not uniformly good, and $44\%$ of its seeds lock in a
wrong matching, which shows up in its upper quartile ($57{,}386$)
rather than in its median. Certified is no cheaper than cert-only at
this horizon (medians $752.2$ and $659.7$; Welch's $t=1.3$): the quota
supply makes certification free only asymptotically
(Remark~\ref{rem:free}) and at $T=10^7$ does not yet cover the
certification demand (see the certification margin below), so
certified pays for both persistent exploration and certification.

\paragraph{Per-player convergence curves.}
Figure~\ref{fig:tracking} shows $\widetilde R_i(T)/b(T)$ for the oracle and
plug-in trackers of \S\ref{sec:experiments}.

\begin{figure}[t]
\centering
\includegraphics[width=\textwidth]{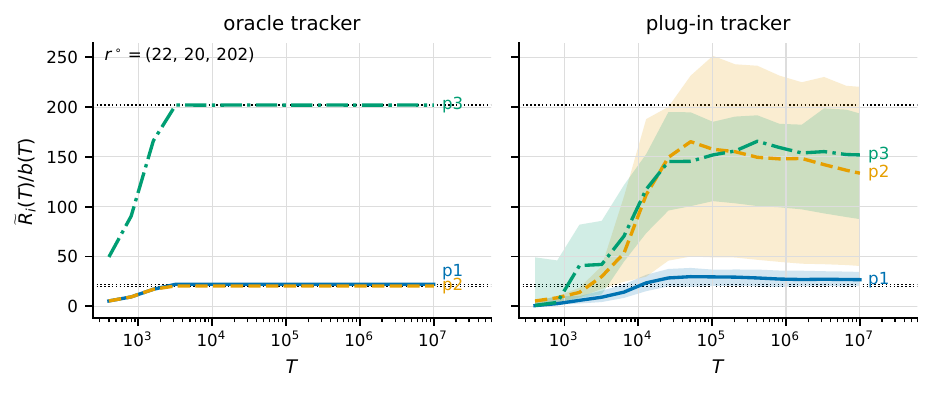}
\caption{$\widetilde R_i(T)/b(T)$ against $T$ (median and interquartile band
over $50$ seeds) for the oracle tracker (left) and the plug-in
tracker (right), with the canonical targets
$r^\circ=(22,20,202)$ dotted. The oracle closely tracks all three constants.
The plug-in medians for $p_2,p_3$ lie between the two
frontier endpoints with wide interquartile bands: at finite horizons
the empirical optimizers select different endpoints across seeds (Remark~\ref{rem:c4-scope}), because at
$w=(1,1,1)$ the two vertices differ in weighted cost by only
$0.8\%$.}
\label{fig:tracking}
\end{figure}

\paragraph{Convergence of the weighted ratio.}
Figure~\ref{fig:weighted} shows $w^\top\widetilde R(T)/(C_w b(T))$ for the
oracle and plug-in trackers. The oracle sits on $1$ from
$T\approx3\times10^3$ onward. The plug-in overshoots to $1.38$ near
$T\approx5\times10^4$, where the empirical quotas are noisiest, and
decays toward $1$ ($1.23$ at $T=10^7$) as the quota estimates
concentrate. Their fluctuations are of the central-limit order
$1/\sqrt n$, where $n=z^\circ b(T)\approx4\times10^3$ is the number of
samples on a binding pair at $T=10^7$.

\paragraph{Certification margin.}
Certifying a binding pair
needs about $2c(t)/\widehat\Delta^2$ samples, with $\widehat\Delta$
its current gap estimate, while its quota supplies about
$2b(t)/\widehat g^{\,2}$, with $\widehat g$ the noise-floor
regularized gap of the last solve. At the practical thresholds the
relative margin $(b-c)/c$ is about $0.7\%$ at $T=10^7$ (the absolute
margin $b(t)-c(t)=(\alpha-\xi)\log\ell(t)$ is
$0.05\,\log\ell(10^7)\approx0.14$), whereas gap-estimate standard
errors of $\approx0.02$ on counts of $4\times10^3$ translate into
quota fluctuations of order $40\%$; the test therefore keeps
demanding samples at feasible horizons.

\begin{figure}[H]
\centering
\includegraphics[width=0.72\textwidth]{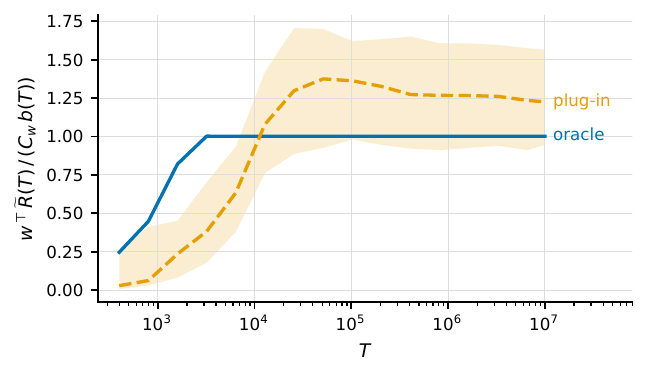}
\caption{$w^\top\widetilde R(T)/(C_w b(T))$ (median and interquartile band,
$50$ seeds) for the oracle and plug-in trackers on $\mu^{(3)}$.}
\label{fig:weighted}
\end{figure}

\paragraph{LP verification on the hand-analyzed instance.}
For the instance $\mu^{(3)}$ of \S\ref{sec:frontier}: over $300$
generic random weight vectors, every returned optimal regret point
is one of the two frontier endpoints $(r_2,r_3)=(20,202)$ and
$(222,2.02)$, always with $r_1=22$; $45$ sampled family points are
optimal under the knife-edge weight $w=(1,0.99,1)$; and on $200$
random weight vectors the matching-cover LP \eqref{eq:coverlp} and the
compact LP \eqref{eq:compactlp} agree to machine precision
(\texttt{frontier\_3x3.py}).

\paragraph{LP size.}
At $N=K=50$ the compact LP has $2{,}501$ variables and $1{,}325$
constraints, and the cost and mass LPs of the practical solve take
$0.28$ seconds together (median over five top-choice-separated
instances; each row is a jittered grid in random order with its
maximum moved to the diagonal; \texttt{lp\_scaling.py}).

\paragraph{A higher-dimensional certificate.}
The $N$-player construction is not only a collection of one-dimensional
examples. For $N=4$, the exact rational certificate in
\texttt{frontier\_family\_certificates.json} contains three affinely independent
optimal regret vectors at the positive weight $w=(1,0.99,1,1)$, hence a
two-dimensional Pareto face. Figure~\ref{fig:frontier-n4} plots its
$(r_2,r_3,r_4)$ coordinates. This is a structural certificate, not a
Monte Carlo estimate; the supplementary script rechecks all primal and
dual equalities using rational arithmetic.

\begin{figure}[t]
\centering
\includegraphics[width=0.72\textwidth]{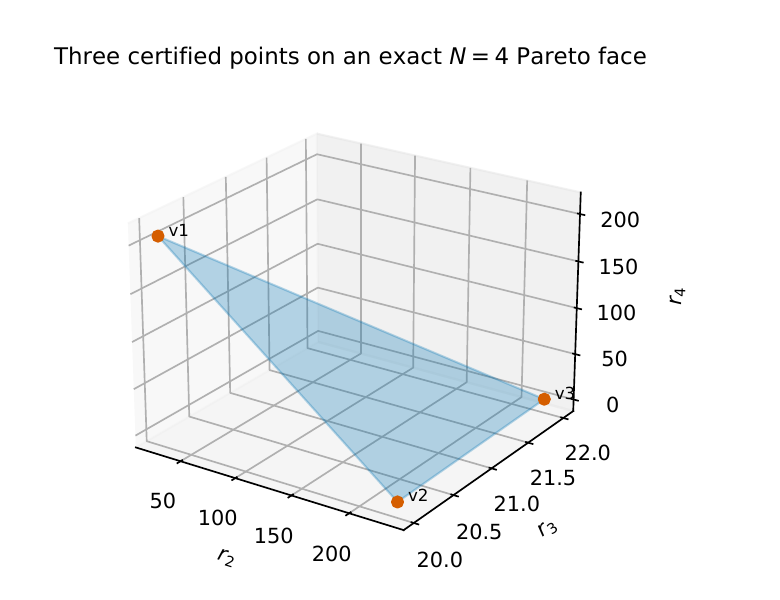}
\caption{Three exact $N=4$ regret vectors spanning a two-dimensional
Pareto face. The top-player coordinate $r_1=76/3$ is constant and is
omitted from the projection.}
\label{fig:frontier-n4}
\end{figure}

\paragraph{Regularized face selection.}
At the knife-edge weight $w=(1,0.99,1)$, where
Assumption~\ref{ass:unique} fails and the cost-optimal face is the
off-stable image of the two-parameter allocation family, we centered the capped program
\eqref{eq:regprog} at the marginal vectors of the three interior
family points $t\in\{50,99,150\}$ of
Proposition~\ref{prop:frontier33}, with quota cap
$\kappa_k=\min(2000,25\cdot2^{k/2})$ and mass cap $NK\kappa_k$
in this diagnostic. Oracle runs use $5$ seeds per target; plug-in
runs use $20$ seeds per target and regularization schedule.
With oracle data the selector recovers every target to solver precision and the executed schedules satisfy
$\max_i|\widetilde R_i(T)/b(T)-r_i(t)|\le0.065$ on a coordinate scale of about
$220$; the mass cap binds at most once, during the early
quota-clipping epochs, and the late-epoch mass ratio is
$\widehat\rho_k/(NK\kappa_k)\approx0.02$, so the cap does not
distort the selection. With plug-in data the median vectors remain
ordered by $t$. Runs have a median of three or four epochs with an
incorrect empirical stable structure, depending on the target and
schedule. The median cap-hit count is zero. Across the nine
configurations, median $w^\top\widetilde R/(C_wb(T))$ ranges from
$1.18$ to $1.46$, with $C_w=243.8$ at this supporting weight.
The diagnostic ran with $\eta_k=k^{-1/4}$, $k^{-1/6}$,
and $0.1k^{-1/4}$
(Remark~\ref{rem:etarate} motivates the logarithmic schedule
displayed in \eqref{eq:regprog}); the runs did not show a consistent
advantage for any schedule at this horizon.

\paragraph{Random markets: LP consistency and finite-horizon tracking.}
We generate $33$ random instances: ten each with $N=K\in\{3,4,5\}$
and three with $(N,K)=(3,4)$, rows drawn i.i.d.\ uniform on $[0,1]$ and
rejection-sampled to top-choice separation and a within-row gap
floor of $0.05$; weights $w=(1,\dots,1)$. On \emph{every} instance
the compact edge-marginal LP \eqref{eq:compactlp} and the
matching-cover LP \eqref{eq:coverlp} agree to machine precision
(maximum relative error $3.55\times10^{-16}$), a numerical
check of the marginal reformulation in Theorem~\ref{thm:compact}
beyond the hand-analyzed instance. Both programs already use
the pairwise quotas, so this is not an independent verification
of the semi-infinite equivalence in Theorem~\ref{thm:reduction}.
The plug-in tracker ($50$ seeds per instance, $T=10^7$) reaches the
pooled ratios $w^\top\widetilde R/(C_w\,b(T))$ of
Table~\ref{tab:random} (top): means $1.6$ to $1.8$ with substantial
dispersion. The overhead is the finite-horizon surcharge of
\S\ref{sec:experiments}, driven by identification and noisy empirical
quotas, and is relatively larger when $C_w$ is small: across the $33$
instances, the Spearman correlation between $C_w$ and the
per-instance mean ratio is $-0.44$. On the three probes of
Table~\ref{tab:random} (bottom), the tracker's ratio falls from
$T=10^5$ to $T=10^7$ ($1.92\to1.64$, $2.15\to1.78$, $1.65\to1.45$;
paired $t=-3.6$, $-10.1$ and $-7.5$ over the $50$ shared seeds); the
large spread on the $3\times3$ probe comes from a few seeds with
ratios above $3$.

\begin{table}[H]
\centering\small
\caption{Random markets and baselines, $w=(1,\dots,1)$, mean $\pm$
sd of $w^\top\widetilde R(T)/(C_w\,b(T))$ over $50$ seeds.
\textbf{Top:} plug-in tracker at $T=10^7$ on the $33$ random
instances, pooled by size. \textbf{Bottom:} one further probe
instance per size and the three-player instance $\mu^{(3)}$, with the
dedicated-cover and layered KL-UCB baselines; KL-UCB is run at
$T=10^7$ and at $T=10^5$, where a plug-in run is added for reference.
$^{*}$From the main suite of \S\ref{sec:experiments} (median $1.23$).}
\label{tab:random}
\begin{tabular}{lcc}
\toprule
Group & instances $\times$ seeds & $w^\top\widetilde R/(C_w b)$ \\
\midrule
$3\times3$ & $10\times50$ & $1.63\pm0.86$ \\
$4\times4$ & $10\times50$ & $1.63\pm0.74$ \\
$5\times5$ & $10\times50$ & $1.80\pm0.78$ \\
$3\times4$ & $\phantom{0}3\times50$ & $1.71\pm0.61$ \\
\bottomrule
\end{tabular}
\vspace{0.6em}

\begin{tabular}{lrccccc}
\toprule
Probe & $C_w$ & plug-in@$10^7$ & cover@$10^7$ & KL-UCB@$10^7$ &
KL-UCB@$10^5$ & plug-in@$10^5$ \\
\midrule
$3\times3$ & $73.9$ & $1.64\pm1.40$ & $3.27\pm0.80$ & $1.20\pm0.25$
& $1.07\pm0.27$ & $1.92\pm1.91$ \\
$4\times4$ & $102.6$ & $1.78\pm0.73$ & $3.03\pm0.85$ & $2.50\pm0.61$
& $1.93\pm0.50$ & $2.15\pm0.95$ \\
$5\times5$ & $309.0$ & $1.45\pm0.34$ & $2.69\pm0.60$ & $1.29\pm0.17$
& $1.02\pm0.13$ & $1.65\pm0.48$ \\
$\mu^{(3)}$ & $244.0$ & $1.28\pm0.46^{*}$ & -- & $1.28\pm0.32$
& $1.00\pm0.28$ & $1.38\pm0.56$ \\
\bottomrule
\end{tabular}
\end{table}

\paragraph{Baselines: scheduling is what the LP buys.}
Two comparators isolate the value of the schedule
(Table~\ref{tab:random}, bottom). Both run at the tracker's horizon
$T=10^7$; the $T=10^5$ columns are kept for reference.
\emph{Dedicated-cover} exploration keeps the tracker's quotas but
serves each quota with its own completion matching instead of the LP
schedule: on the three probe instances it pays $1.7\times$ to
$2.0\times$ the tracker's weighted pseudo-regret, the price of
forgoing parallelization and cost-aware routing. \emph{Layered
KL-UCB} (each priority level plays the remaining arm with the highest
Gaussian index, the natural centralized index policy in the spirit of
\citealp{liu2020competing}) is, by contrast, a strong
weighted-pseudo-regret baseline. At $T=10^7$ its ratio is below the
tracker's on the $3\times3$ and $5\times5$ probes ($1.20$ against
$1.64$, $1.29$ against $1.45$), equal on $\mu^{(3)}$ ($1.28$) and
above it on the $4\times4$ probe ($2.50$ against $1.78$); the index
dynamics parallelize \emph{incidentally}, because while $p_1$
explores $a_2$ the displaced $p_2$ explores its own uncertain arm
$a_3$. From $T=10^5$ to $T=10^7$, however, its ratio rises on all four
instances ($1.07\to1.20$, $1.93\to2.50$, $1.02\to1.29$,
$1.00\to1.28$; paired $t$ from $6.0$ to $17.4$), so over this range
layered KL-UCB moves away from the weighted optimum $C_w$ while the
tracker's ratio falls on the three probes.

The two policies differ most in the regret \emph{vector}. The KL-UCB
mean vector on $\mu^{(3)}$ is $(22.1,\,26.6,\,196.4)$ at $T=10^5$,
close to the parallel end $(22,20,202)$ of the frontier, and
$(27.8,\,32.9,\,252.4)$ at $T=10^7$, beyond it: the lowest-priority
player absorbs essentially the whole externality. Unlike $\pi_s$,
which sweeps the segment (Figure~\ref{fig:frontier}), and
$\pi_{w,\tau}$, which selects arbitrary face points, layered KL-UCB
has no parameter that moves it along the segment. Externality
\emph{scheduling}, not weighted-regret optimality alone, is what
distinguishes the LP-driven policies.

\FloatBarrier
\paragraph{Theorem-style certification thresholds with a practical cap.}
Finally we use the thresholds $\alpha=3.3$, $\xi=3.1$ and epoch
boundaries $T_k=\lceil e^{k^2}\rceil$, executing epochs $k=2,3,4$
after the six-round prefix ($T=8{,}886{,}111$,
$b(T)/\log T=1.57$), with increment execution and
stable-matching replay, $20$ seeds per mode. The diagnostic
retains a cap of $2000$, a short initialization, numerical LP
tie-breaking, and batched exploitation with estimates held fixed
within each batch; it samples the least-certified comparison instead
of the first failing one. Both this diagnostic and
Algorithm~\ref{alg:policy} solve after repair; the diagnostic evaluates
the noise-floor scale at the current time $t$, whereas the algorithm
uses $\beta(T_{k-1})$. The algorithm also refreshes and certifies the
candidate each round. Its cap $\kappa_k=k$ is below the true quota
$200$ at all three simulated epoch indices; it would substantially
restrict quota-driven exploration.
No run has a wrong-exploitation round in either mode. The per-round
certification guarantee does not cover the batched implementation.
The certification surcharge dominates at
this horizon: estimation rounds have medians $28{,}846$ (oracle)
and $34{,}073$ (plug-in), rounded as in Table~\ref{tab:ablation}, and
$w^\top\widetilde R/b(T)$ has mean $\pm$ sample standard deviation
$563\pm189$ and $495\pm188$, respectively, against $C_w=244$; the
difference between the two modes is within seed-to-seed variability
(Welch's $t=1.15$, $p=0.26$, with $20$ seeds each). This
quantifies Remark~\ref{rem:free} at these certification thresholds:
the relative margin $(b-c)/c$ is about $2.3\%$ at the final horizon
(the absolute margin $b(t)-c(t)=0.2\log\ell(t)$ is $\approx0.55$),
far below the finite-sample fluctuation of the empirical quotas, so
the test keeps demanding samples. Certification is asymptotically free
for the theoretical policy of Remark~\ref{rem:free}; this
batched diagnostic has a measurable surcharge at $T\approx10^7$.

\end{document}